\documentclass[11pt]{article}
\usepackage[a4paper,left=20mm,top=20mm,right=20mm,bottom=30mm]{geometry}

\usepackage{amsmath,amsopn,amsthm,amssymb}

\usepackage{authblk}
\usepackage[numbers,sort&compress]{natbib}

\usepackage{graphicx}
\usepackage{mathptmx} 
\usepackage[mathscr]{euscript}
\usepackage{mathtools}
\usepackage{amssymb}
\usepackage{wasysym}
\usepackage{float}
\usepackage{enumitem}
\usepackage{xspace}
\usepackage{url}
\usepackage[margin=30pt,font=small,labelfont=bf]{caption}

\usepackage{placeins}

\usepackage[
 colorlinks=true,
urlcolor=black,
linkcolor=blue
]{hyperref}

\newcommand{\FIGDIR}{./images}

\usepackage{algorithmicx}
\usepackage{algorithm} 
\usepackage[noend]{algpseudocode}
\algrenewcommand\algorithmicrequire{\textbf{Input:}}
\algrenewcommand\algorithmicensure{\textbf{Output:}}

\newtheorem{theorem}{Theorem}
\newtheorem{lemma}{Lemma}
\newtheorem{corollary}{Corollary}
\newtheorem{proposition}{Proposition}
\newtheorem{fact}{Observation}
\newtheorem{definition}[theorem]{Definition}%
\newtheorem*{remark}{Remark}%

\newtheorem{claim}{Claim}
\newenvironment{claim-proof}%
{\begin{description}[leftmargin = 0.2cm, labelsep = 0.2cm]
    \item[\emph{Proof of claim:}]}{\hfill$\diamond$\end{description}}

\newtheorem{problem}{Problem}[section]
\newcommand{\PROBLEM}[1]{{\sc #1}} 

\makeatletter
\def\moverlay{\mathpalette\mov@rlay}
\def\mov@rlay#1#2{\leavevmode\vtop{%
   \baselineskip\z@skip \lineskiplimit-\maxdimen
   \ialign{\hfil$\m@th#1##$\hfil\cr#2\crcr}}}
\newcommand{\charfusion}[3][\mathord]{
    #1{\ifx#1\mathop\vphantom{#2}\fi
        \mathpalette\mov@rlay{#2\cr#3}
      }
    \ifx#1\mathop\expandafter\displaylimits\fi}
\makeatother

\newcommand{\cupdot}{\charfusion[\mathbin]{\cup}{\cdot}}

\usepackage{tabularx}
\makeatletter
\newcommand{\multiline}[1]{%
  \begin{tabularx}{\dimexpr\linewidth-\ALG@thistlm}[t]{@{}X@{}}
    #1
  \end{tabularx}
}
\makeatother
\newcommand{\AX}[1]{{\upshape(#1)}}
\newcommand{\scen}{\ensuremath{\mathscr{S}}}
\newcommand{\graphs}{\ensuremath{\mathscr{G}}}
\newcommand{\R}{\mathscr{R}}
\newcommand{\F}{\mathscr{F}}
\newcommand{\tT}{\ensuremath{\tau_{T}}}
\newcommand{\tS}{\ensuremath{\tau_{S}}}
\newcommand{\Gg}{\ensuremath{G_{_{=}}}}
\newcommand{\Gu}{\ensuremath{G_{_{<}}}}
\newcommand{\Ga}{\ensuremath{G_{_{>}}}}
\newcommand{\Sri}{\ensuremath{\mathfrak{S}}}
\newcommand{\rpair}[3]{\{{#1}{#2}\vert {#3},{#1}{#3}\vert {#2}\}}
\newcommand{\CF}{C_F}
\newcommand{\CR}{C_R}
\DeclareMathOperator{\parent}{par}
\DeclareMathOperator{\lca}{lca}
\DeclareMathOperator{\child}{child}
\DeclareMathOperator{\depth}{depth}
\DeclareMathOperator{\LA}{LA}
\DeclareMathOperator{\gfitch}{\digamma}
\DeclareMathOperator{\gjoin}{\triangledown}

\newcommand{\wQO}{\ensuremath{\Psi^{w}}}  
\newcommand{\sQO}{\ensuremath{\Psi^{s}}}			
\newcommand{\wO}{\ensuremath{\Theta^{w}}}  
\newcommand{\sO}{\ensuremath{\Theta^{s}}}			

\providecommand{\keywords}[1]{\textbf{\textit{Keywords: }} #1}

  \title{Relative Timing Information and Orthology in Evolutionary Scenarios}

\author[1,2]{David Schaller}
\author[1]{Tom Hartmann} 
\author[3]{Manuel Lafond} 
\author[8]{Nicolas Wieseke} 
\author[1,2,4-7,*]{Peter F. Stadler}    
\author[9,*]{Marc Hellmuth}

\affil[1]{Bioinformatics Group, Department of Computer Science \&
  Interdisciplinary Center for Bioinformatics, Universit{\"a}t Leipzig,
  H{\"a}rtelstra{\ss}e~16--18, D-04107 Leipzig, Germany.}

\affil[2]{Max Planck Institute for Mathematics in the Sciences,
  Inselstra{\ss}e 22, D-04103 Leipzig, Germany}

\affil[3]{  Department of Computer Science, Universit{\'e} de Sherbrooke,
  2500 boul. de l'Universit{\'e} Sherbrooke, Qc, Canada, J1K 2R1}

\affil[4]{German Centre for Integrative Biodiversity Research
  (iDiv) Halle-Jena-Leipzig, Competence Center for Scalable Data Services
  and Solutions Dresden-Leipzig, Leipzig Research Center for Civilization
  Diseases, and Centre for Biotechnology and Biomedicine at Leipzig
  University at Universit{\"a}t Leipzig}

\affil[5]{Institute for Theoretical Chemistry, University of Vienna,
  W{\"a}hringerstrasse 17, A-1090 Wien, Austria}

\affil[6]{Facultad de Ciencias, Universidad National de Colombia, Sede
  Bogot{\'a}, Colombia}

\affil[7]{Santa Fe Insitute, 1399 Hyde Park Rd., Santa Fe NM 87501,
  USA}

\affil[8]{Swarm Intelligence and Complex Systems Group, 
  Department of Computer Science, Leipzig University,
  Augustusplatz 10, D-04109 Leipzig, Germany}

\affil[9]{Department of Mathematics, Faculty of Science, Stockholm University,
  SE - 106 91 Stockholm,   Sweden \newline \texttt{marc.hellmuth@math.su.se}}

\affil[*]{corresponding authors}

\date{\ }

\begin{document}
\sloppy

\maketitle

\abstract{ 
  Evolutionary scenarios describing the evolution of a family of
  genes within a collection of species comprise the mapping of the vertices
  of a gene tree $T$ to vertices and edges of a species tree $S$. The
  relative timing of the last common ancestors of two extant genes (leaves
  of $T$) and the last common ancestors of the two species (leaves of $S$)
  in which they reside is indicative of horizontal gene transfers (HGT) and
  ancient duplications. Orthologous gene pairs, on the other hand, require
  that their last common ancestors coincides with a corresponding
  speciation event. The relative timing information of gene and species
  divergences is captured by three colored graphs that have the extant
  genes as vertices and the species in which the genes are found as vertex
  colors: the equal-divergence-time (EDT) graph, the later-divergence-time
  (LDT) graph and the prior-divergence-time (PDT) graph, which together
  form an edge partition of the complete graph. 
  
  Here we give a complete
  characterization in terms of informative and forbidden triples that can
  be read off the three graphs and provide a polynomial time algorithm for
  constructing an evolutionary scenario that explains the graphs, provided
  such a scenario exists. While both LDT and PDT graphs are cographs, this
  is not true for the EDT graph in general. We show that every EDT graph is
  perfect.  While the information about LDT and PDT graphs is necessary to
  recognize EDT graphs in polynomial-time for general scenarios, this extra
  information can be dropped in the HGT-free case.  However, recognition of
  EDT graphs without knowledge of putative LDT and PDT graphs is
  NP-complete for general scenarios. In contrast, PDT graphs can be
  recognized in polynomial-time.  We finally connect the EDT graph to the
  alternative definitions of orthology that have been proposed for
  scenarios with horizontal gene transfer. With one exception, the
  corresponding graphs are shown to be colored cographs.
}
\smallskip

\noindent 
\keywords{gene tree, species tree, cograph, perfect graph, orthology,
  xenology, horizontal gene transfer, informative and forbidden triples,
  relative timing, NP-hardness
}

\sloppy

\maketitle

\section{Introduction}

An \emph{evolutionary scenario} describes the history of a gene family
relative to the phylogeny of a set of species. Formally, it comprises a
mapping $\mu$ of the gene tree $T$ into the species tree $S$, usually
called the \emph{reconciliation} of $S$ and $T$. The conceptual relevance
of scenarios in evolutionary biology derives from the fact that they define
key relationships between genes, in particular orthology, paralogy, and
xenology \cite{Fitch:00}. On the practical side, scenarios also imply
relations on the set of genes that can be inferred directly from sequence
similarity data, such as the \emph{best match} relation
\cite{Geiss:20b,Stadler:20a} or the \emph{later divergence time} (LDT)
relation \cite{Schaller:21f}, which is closely related to the inference of
horizontal gene transfer (HGT) events.

In the absence of horizontal transfer, orthology is characterized by the
fact that the last common ancestor of two genes $x$ and $y$ is exactly the
speciation event that separated the two species $\sigma(x)$ and $\sigma(y)$
in which $x$ and $y$, resp., reside \cite{Fitch:00}. A necessary condition
for orthology, therefore, is that the last common ancestor of the genes $x$
and $y$ and the last common ancestor of the species $\sigma(x)$ and
$\sigma(y)$ have the same evolutionary age. Whether or not $x,y$ and
$\sigma(x),\sigma(y)$ have \emph{equal divergence time} (EDT) can be
decided (at least at some level of accuracy) directly from sequence
data. The graph $\Gg$ whose vertices are the genes and whose edges are the
pairs of genes with equal divergence time of $x,y$ and
$\sigma(x),\sigma(y)$ thus is an empirically accessible datum. By
construction, furthermore, the EDT graph contains the orthology graph as a
subgraph.

The LDT and EDT relations can be complemented with a ``prior
  divergence time'' relation (PDT). Together, the EDT, LDT and PDT
  relations then define a 3-partition $\graphs$ of the edge set of a
  complete graph with the genes as vertices.  Since the EDT relation has
  some connection with orthology and the LDT relation with xenology, it
  seems intuitive that the PDT relation might be connected with paralogy.
  However, for none of the three relations this connection is strict in the
  sense that it would enforce a particular type of evolutionary event at
  the corresponding last common ancestor. Figure~\ref{fig:simple-scenarios}
  shows examples of evolutionary scenarios with genes in EDT relation (top
  row), LDT relation (middle row) and PDT relation (bottom row) with the
  corresponding last common ancestor being any of the event types
  speciation, HGT, and duplication.  The EDT, LDT, and PDT relations are
  therefore distinct from the orthology, xenology, and paralogy relations
  considered in \cite{Hellmuth:16a}.  Nevertheless, the relative timing
  information from the last common ancestors of pairs of extant genes can
  be used to construct the topologies of the underlying gene and species
  tree as well as a reconciliation between them. The reconciliation then
  determines the orthology, xenology, and paralogy relations.  The
  reconciliation, however, is in general not uniquely determined by the
  3-partition $\graphs$.
    
We show here that a collection of \emph{informative} and \emph{forbidden
triples} defined by $\graphs$ are the key criteria to determine whether or
not $\graphs$ derives from a scenario $\scen$. While both LDT and PDT
graphs are cographs, this is not always the case for the EDT graph. We
shall see, however, that it is a cograph if both $T$ and $S$ are binary
(fully resolved) trees. In Section ``Explanation of $\graphs$ by Relaxed
Scenarios'' we derive a quartic time algorithm for the recognition of
edge-tripartitions that derive from a corresponding scenario. This
construction is then used to give a triple-based characterization.  We then
show that the existence of an explaining scenario is sufficient to
guarantee that $\graphs$ can also be explained by scenarios with several
additional desirable properties. Importantly, these restricted scenarios
have properties that are often assumed for valid reconciliations of $T$ and
$S$ in the literature.  For instance, it is possible to choose the
  scenarios such that each event (inner node of $T$) has at least one
  purely vertical descendant; this is the case for all scenarios in
  Fig.~\ref{fig:simple-scenarios}.  In Section ``Orthology and
Quasi-Orthology'', EDT graphs are connected with several competing notions
of ``orthology'' proposed by different authors
\cite{Fitch:70,Gray:83,Fitch:00,Darby:17}.

\begin{figure}[htbp]
  \centering
  \includegraphics[width=0.25\textwidth]{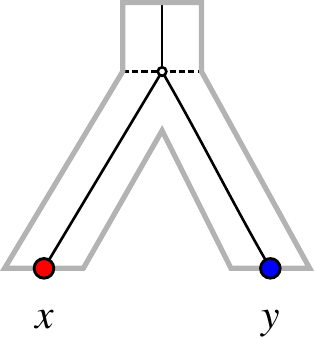}
  \hfill
  \includegraphics[width=0.25\textwidth]{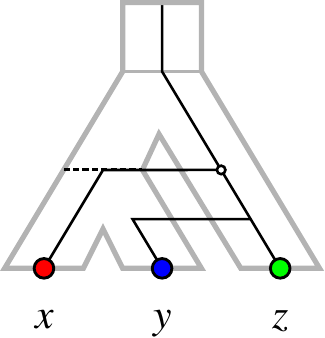}
  \hfill
  \includegraphics[width=0.25\textwidth]{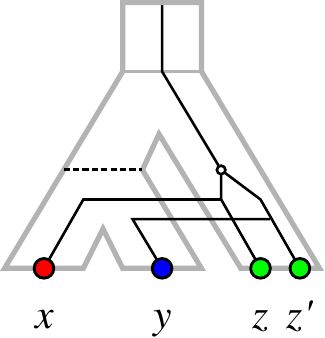}
  \smallskip\hrule\smallskip
  \includegraphics[width=0.25\textwidth]{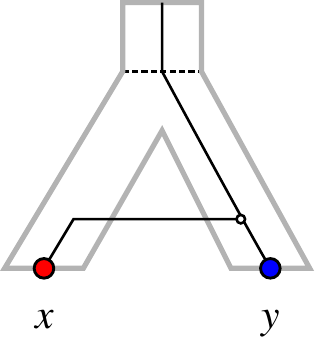}
  \hfill
  \includegraphics[width=0.25\textwidth]{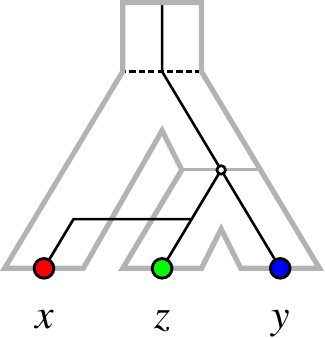}
  \hfill
  \includegraphics[width=0.25\textwidth]{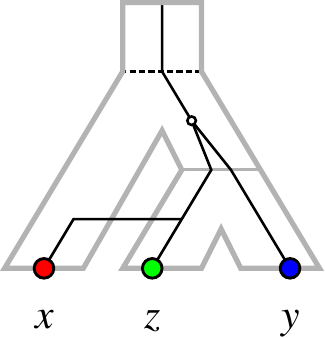}
  \smallskip\hrule\smallskip
  \includegraphics[width=0.25\textwidth]{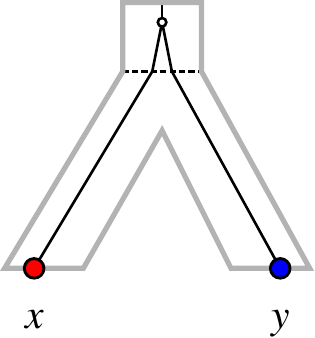}
  \hfill
  \includegraphics[width=0.25\textwidth]{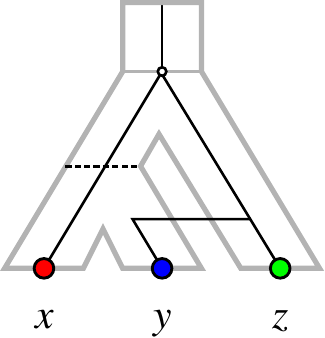}
  \hfill
  \includegraphics[width=0.25\textwidth]{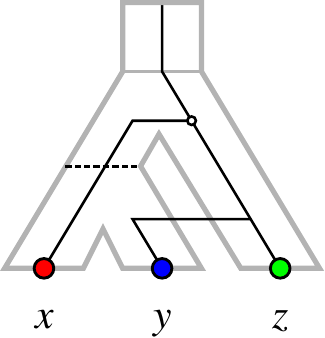}
  \caption{Examples of evolutionary scenarios depicted as gene trees
      (black inline trees) embedded into species trees (gray outline
      trees).  In all cases, the ancestral gene $\lca_T(x,y)$ of $x$ and
      $y$ is highlighted as white circle while the corresponding species
      $\lca_S(\sigma(x),\sigma(y))$ is highlighted as dashed line.
      \emph{Top row:} scenario with $x$ and $y$ in EDT relation, i.e., the
      ancestral gene $\lca_T(x,y)$ diverged concurrently with the
      corresponding species $\lca_S(\sigma(x),\sigma(y))$.  The
      evolutionary event at $\lca_T(x,y)$ is either a speciation (left),
      a  horizontal gene transfer (center), or a duplication (right).
      \emph{Middle row:} scenario with $x$ and $y$ in
      LDT relation, i.e., the ancestral gene $\lca_T(x,y)$ diverged after
      the corresponding species $\lca_S(\sigma(x),\sigma(y))$.  The
      evolutionary event at $\lca_T(x,y)$ is either a horizontal gene
      transfer (left), a speciation (center), or a duplication (right).
      \emph{Bottom row:} scenario with $x$ and $y$ in PDT relation, i.e.,
      the ancestral gene $\lca_T(x,y)$ diverged before the corresponding
      species $\lca_S(\sigma(x),\sigma(y))$.  The evolutionary event at
      $\lca_T(x,y)$ is either a duplication (left), a speciation (center),
      or a horizontal gene transfer (right).}
  \label{fig:simple-scenarios}
\end{figure}

\section{Notation}

\paragraph{Graphs} 
We consider undirected simple graphs $G=(V,E)$ with vertex set
$V(G)\coloneqq V$ and edge set $E(G)\coloneqq E$.  We write $G\subseteq H$
if $G=(V,E)$ is a subgraph of $H=(V', E')$, i.e., if $V\subseteq V'$ and
$E\subseteq E'$. The subgraph of $G$ that is induced by the subset $X
  \subseteq V$ will be denoted by $G[X]$. A connected component $C$ of 
  $G$ is an inclusion-maximal subset $C\subseteq V$ such that $G[C]$ is 
  connected.
The complement of a graph
$G=(V,E)$ is the graph $\overline{G}=(V,\overline{E})$ with vertex set $V$
and an edge $xy\in\overline{E}$ for $x\ne y$ precisely if $xy\notin E$.
We denote by $K_n$ the graph on $n$ vertices in which every possible
  edge is present, hereafter called a \emph{complete graph}.  A graph
property $\Pi$ is a subset of the set of all graphs. A graph property $\Pi$
is \emph{closed under complementation} if $G\in\Pi$ implies
$\overline{G}\in\Pi$.
  
\paragraph{Rooted Trees} 
Trees are connected and acyclic graphs.  All trees in this contribution
have a distinguished vertex $\rho$, called the \emph{root} of the tree. For
two vertices $x,y\in V(T)$, we write $y \preceq_{T} x$ if $x$ lies on the
unique path from the root to $y$, in which case $x$ is called an
\emph{ancestor} of $y$, and $y$ is called a \emph{descendant} of $x$.  If,
in addition, $x$ and $y$ are adjacent in $T$, then $x$ is the \emph{parent}
of $y$ (denoted by $\parent_T(y)$), and $y$ is a \emph{child} of $x$. The
set of children of $x$ is denoted by $\child_T(x)$.  We write edges $e=xy$
indicating that $y\preceq_T x$.  It will be convenient to extend the
relation $\preceq_{T}$ to the union $V(T)\cup E(T)$ as follows: For a
vertex $x\in V(T)$ and an edge $e=uv\in E(T)$, we set $x \preceq_T e$ if
and only if $x\preceq_T v$; and $e \preceq_T x$ if and only if $u\preceq_T
x$.  In addition, for edges $e=uv$ and $f=ab$ in $T$, we put $e\preceq_T f$
if and only if $v \preceq_T b$ (note that under this definition, $uv
  \preceq_T uv$). For $x,y\in V(T)\cup E(T)$, we may also write $x
\succeq_{T} y$ instead of $y \preceq_{T} x$. We use $y \prec_T x$ for $y
\preceq_{T} x$ and $x \neq y$. Moreover, we say that $x$ and $y$ are
\emph{comparable} if $y\preceq_{T} x$ or $x\preceq_{T} y$ holds and,
otherwise, $x$ and $y$ are \emph{incomparable}.  Note that $\preceq_{T}$ is
a partial order with a unique maximal element $\rho$. The \emph{leaves}
$L=L(T)\subseteq V(T)$ of $T$ are precisely the $\preceq_{T}$-minimal
elements.

From here on, we assume that the root $\rho$ as well as every non-leaf
vertex of a tree have always at least two children. Moreover, we write
$T(u)$ for the \emph{subtree of $T$ rooted at $u$}, i.e, the tree that is
induced by $u$ and all its descendants.

For a set of leaves $A\subseteq L$, we write $\lca_T(A)$ for the
\emph{last common ancestor} of $A$, i.e., the unique $\preceq_T$-minimal
vertex in $V(T)$ such that $x\preceq \lca_T(A)$ for all $x\in A$. For
simplicity, we write $\lca_T(x,y)$ instead of $\lca_T(\{x,y\})$. The
\emph{restriction of $T$} to a subset $L'\subseteq L$, in symbols $T_{\vert
  L'}$, is obtained from the minimal subtree of $T$ that connects all
leaves in $L'$ by suppressing all vertices with degree two except possibly
the root $\rho_{T_{\vert L'}}$. We often write $T_{\vert x_1\dots x_k}$
instead of $T_{\vert \{x_1,\dots,x_k\}}$.  A tree $T$ \emph{displays}
a tree $T'$ with $L(T')\subseteq L(T)$ if $T'$ is isomorphic to $T_{\vert
  L(T')}$.

\paragraph{Planted Trees} 
In order to accommodate evolutionary events pre-dating
$\rho\coloneqq\lca(L)$, we consider \emph{planted trees}, i.e., we assume
an additional planted root $0_T$ with degree $1$ that is the parent of the
``root'' $\rho$. The \emph{inner vertices} of $T$ are $V^0(T) \coloneqq
V(T) \setminus (L(T) \cup \{0_T\})$.  In particular, a planted tree $T$
always displays the rooted tree $T_{\vert L(T)}$ obtained by removing $0_T$
and its incident edge $0_T\rho$.  
\begin{remark}
  Unless explicitly stated otherwise, the trees that appear in this
  contribution are planted phylogenetic trees, i.e., $0_T$ is the only
  vertex with exactly one child. All other vertices are either leaves or
  have at least two children.
\end{remark}

\paragraph{Triples and Fan Triples}
A (rooted) triple is a binary rooted tree on three vertices.  We denote by
$xy\vert z$ the rooted triple $t$ with leaf set $\{x,y,z\}$ and
$\lca_t(x,y) \prec_T \lca_t (x,z) = \lca_t(y,z)$. A tree $T$ displays
  $xy\vert z$ if $\lca_T(x,y) \prec_T \lca_T (x,z) = \lca_T(y,z)$.  A
\emph{fan triple} $x\vert y\vert z$ on leaves $x,y,z$ is the tree
$(x,y,z)$.  A tree $T$ displays the fan triple $x\vert y\vert
z$ if $\lca_T(x,y) = \lca_T(x,z) = \lca_T(y,z)$.

As usual, we say that a set $\R$ of triples is \emph{consistent} if there
is a tree $T$ that displays all of the triples in $\R$. If $(\R,\F)$ is a
pair of two triple sets, we say that $(\R,\F)$ is consistent if there is a
tree $T$ that displays all of the triples in $\R$ but none of the triples
in $\F$. In this case, we say that $T$ \emph{agrees with} $(\R,\F)$.  We
will frequently make use of the following simple observation that
collects the structures of the subtree $T_{\vert L'\cup L''}$ on $\vert
L'\cup L''\vert =4$ leaves implied by two subtrees $T_{\vert L'}$ and
$T_{\vert L''}$ on three leaves (triples) sharing $\vert L'\cap L''\vert=2$
common leaves.  The statements are closely related to the so-called
``inference rules'' for rooted triples, see in particular
\cite{Dekker86,Bryant:95}. We leave the elementary proofs to the interested
reader.  We use Newick notation for rooted trees, i.e., inner vertices
correspond to matching parentheses, leaves are given by their labels, and
commas are used to separate sibling to increase readability. For example,
the triple $ab\vert c$ is equivalently represented as $((a,b),c)$.

\begin{fact}
  \label{obs:triples}
  Let $T$ be a tree and $a,b,c,d\in L(T)$ be pairwise distinct leaves. 
  Suppose $T$ displays $ab\vert c$.
  \begin{enumerate}[nolistsep,noitemsep]
  \item[(i)] If $T$ displays $cd\vert a$, then 
    $T_{\vert abcd}=((a,b),(c,d))$. 
  \item[(ii)] If $T$ displays $ac\vert d$,
    then $T_{\vert abcd} = (((a,b),c),d)$.
  \item[(iii)] If $T$ displays $ad\vert c$, then
    $T$ displays $bd\vert c$
    and $T_{\vert abcd}$ is one of the trees
    $(((a,d),b),c)$,
    $(((b,d),a),c)$,
    $(((a,b),d),c)$,
    or $((a,b,d),c)$.
  \item[(iv)] If $T$ displays $ab\vert d$, then $T_{\vert abcd}$ is
   one of the trees $(((a,b),c),d)$,
    $(((a,b),d),c)$, $((a,b),c,d)$, or $((a,b),(c,d))$.
  \item[(v)] If $T_{\vert bcd}=(b,c,d)$, then $T_{\vert abcd}=((a,b),c,d)$.
  \end{enumerate}
  \medskip  
  \noindent
  Suppose that $T$ does not display any of the triples on
  $\{a,b,c\}$, i.e., $T_{\vert abc}=(a,b,c)$. 	
  \begin{enumerate}[nolistsep,noitemsep]

  \item[(vi)] If $T_{\vert bcd}=(b,c,d)$, then $T_{\vert abcd}=(a,b,c,d)$
    or $T_{\vert abcd}=((a,d),b,c)$.
  \end{enumerate}
\end{fact}
We will make use of Obs.~\ref{obs:triples} throughout the subsequent proofs
without explicit reference.

\paragraph{Cographs}
The \emph{join} of two graphs $G=(V,E)$ and $H=(W,F)$ with disjoint vertex
sets $V\cap W=\emptyset$ is the graph $G\gjoin H$ with vertex set
$V\cupdot W$ and edge set $E\cupdot F\cupdot \{xy\mid x\in V,y\in
W\}$. Similarly, their \emph{disjoint union} $G\cupdot H$ has vertex set
$V\cupdot W$ and edge set $E\cupdot F$.  \emph{Cographs} are recursively
defined as the graphs that either are $K_1$s or can be obtained from the
join or disjoint union of two cographs.  Cographs have been studied
extensively. We summarize some basic results in the next proposition.
\begin{proposition}{\cite{Corneil:81}}
  Given an undirected graph $G$, the following statements are equivalent:
  \begin{enumerate}[noitemsep,nolistsep]
    \item $G$ is a cograph.
    \item $G$ is explained by a \emph{cotree} $(T,t)$, i.e., a rooted tree $T$
      with $L(T)=V(G)$ and $t\colon V^0(T)\to\{0,1\}$ such that $xy\in E(G)$
      precisely if $t(\lca_T(x,y))=1$.
    \item The complement graph $\overline G$ of $G$ is a cograph.
    \item $G$ does not contain a $P_4$, i.e., a path on four vertices, as an
    induced subgraph.
    \item Every induced subgraph $H$ of $G$ is a cograph.
  \end{enumerate}
  \label{prop:cograph}
\end{proposition}

\section{Equal Divergence Time Graphs}

\subsection{Evolutionary Scenarios}

The vertices in phylogenetic trees designate evolutionary events such as 
speciations, gene duplications, or horizontal gene transfers. Conceptually, any 
such event $x$ is associated with a specific point in time $\tT(x)$.
\begin{definition}
  \label{def:timemap}
Let $T$ be a rooted or planted tree. Then $\tT:V(T)\to\mathbb{R}$ is
a time map for $T$ if $x\prec_T y$ implies $\tT(x)<\tT(y)$. The tuple
  $(T,\tT)$ is called \emph{dated tree}.
\end{definition}
Definition~\ref{def:timemap} ensures that the ancestor relation $x\prec_T
y$ and the timing of the vertices are not in conflict. It also
  pertains to arbitrary rooted trees since these can be seen as
  restrictions of planted trees to $V\setminus\{0_T\}$. Note that
  for an edge $uv$ of $T$, the convention that $uv$ implies $v\prec_T u$,
  also implies $\tT(v)<\tT(u)$.
Below we will make use of the fact that time maps are
easily constructed for rooted trees: 

\begin{lemma}
  \textnormal{\cite[Lemma~1]{Schaller:21f}} Given a tree $T$ (planted
    or not), a time map $\tT$ for $T$ satisfying $\tT(x)=\tau_0(x)$ with
    arbitrary choices of $\tau_0(x)$ for all $x\in L(T)$ can be constructed
    in linear time.
  \label{lem:arbitrary-tT}
\end{lemma}

It is usually difficult and often impossible to obtain reliable, accurate
``time stamps'' $\tT(x)$ for evolutionary relevant events
\cite{Rutschmann:06,Sauquet:13}. Although the time map $\tT$ turns out to
be a convenient formal tool, we will never need to make use of the absolute
values of $\tT(x)$. Instead, we will only need \emph{relative} timing
information, i.e., it will be sufficient to know whether an event
pre-dates, post-dates, or is concurrent with another one. This information
is often much easier to extract \cite{Ford:09,Szollosi:22}.  For the
  sake of concreteness, one may imagine that $\tau_0(x) = 0$ for all $x \in
  L(T)$, although this is not a requirement.
 
\begin{definition}
  \label{def:rs}
  A \emph{relaxed scenario} $\scen=(T,S,\sigma,\mu,\tT,\tS)$ consists of a
  dated gene tree $(T,\tT)$, a dated species tree $(S,\tS)$, a leaf
  coloring $\sigma\colon L(T)\to M$ with $M\subseteq L(S)$, and a
  \emph{reconciliation map} $\mu \colon V(T)\to V(S)\cup E(S)$ such that
  \begin{itemize}
    \item[\AX{S0}] $\mu(x)=0_S$ if and only if $x=0_T$.
    \item[\AX{S1}] $\mu(x)\in L(S)$ if and only if $x\in L(T)$ and, in
      particular, $\mu(x)=\sigma(x)$ in this case.
    \item[\AX{S2}]     If $\mu(x)\in  V(S)$, then
    $\tS(\mu(x))=\tT(x)$.
    \item[\AX{S3}] 
      If $\mu(x)=uv\in E(S)$, then  $\tS(v)<\tT(x)<\tS(u)$.
  \end{itemize}
\end{definition}
The axioms \AX{S2} and \AX{S3} specify \emph{time consistency}. Note
  that we impose no (direct) restrictions on ancestrality relationships,
  hence the \emph{relaxed} nature of our scenarios.  In particular, for
  vertices $x, y \in V(T)$, it is possible that $x$ is a descendant of $y$,
  but that $\mu(x)$ is \emph{not} a descendant of $\mu(y)$.  This may occur
  if $\mu(x)$ and $\mu(y)$ are incomparable because of the presence of
  horizontal gene transfers on the path from $y$ to $x$.  This contrasts
  with traditional reconciliation models that only support gene
  duplications and forbid this type of map.  By minimizing the amount of
  constraints imposed on the model, we aim to characterize the broadest
  class of divergence time patterns that could be explained in some way.
  Conversely, this means that divergence times that cannot be explained by
  our scenarios can be deemed erroneous with more confidence, as they
  cannot even meet a relaxed set of requirements.  In the later sections,
  however, we focus on more restrictive scenarios.  As we shall see,
  relaxed scenarios allow ``unobservable'' transfers, for which the
  ancestral gene in the origin species has no direct extant descendants (in
  the sense that they were not transmitted by any transfer).  We will study
  \emph{restricted scenarios} in which such unobservable transfers are
  forbidden, and then later on we look at scenarios in which transfers are
  entirely forbidden. The scenarios considered in \cite{Tofigh:11} as
  well as the H-trees \cite{Gorecki:10} admit the assignment of unique
  event type (duplication, speciation, etc.) to a vertex $x$ in the gene
  tree $T$ depending on its reconciliation and the reconciliation of its
  children. This is not the case in relaxed scenarios. Here a vertex in $T$
  may simultaneously represent multiple event types. For example a
  ``speciation'' vertex with $\mu(x)\in V(S)$ may still have multiple
  direct descendants in the same lineage, hence sharing properties of of a
  duplication.  We first consider a few simple properties of
reconciliation maps. In fact, these are well-known properties for more
restrictive definitions of reconciliation.

\begin{lemma}
  Let $\scen=(T,S,\sigma,\mu,\tT,\tS)$ be a relaxed scenario. If $v,w\in
  V(T)$ such that $v\preceq_T w$ and $\mu(v)=\mu(w)\in V^0(S)$, then $v=w$.
  \label{lem:v=w}
\end{lemma}
\begin{proof}
  Set $U\coloneqq \mu(v)=\mu(w)\in V^0(S)$. Then $\tT(w)=\tT(v)=\tS(U)$.
  However, if $v\preceq_T w$ and $v\ne w$, i.e., $v\prec_T w$, then
  $\tT(v)<\tT(w)$ by Def.~\ref{def:timemap}; a contradiction.
\end{proof}

\begin{lemma}
  \label{lem:mu-timing}
  If $\scen=(T,S,\sigma,\mu,\tT,\tS)$ is a relaxed scenario then
  $x\preceq_T y$ implies $\mu(x)\not\succ_S \mu(y)$ for all $x,y\in V(T)$.
\end{lemma}
\begin{proof}
  If $x=y$, then there is nothing to show. Otherwise, $x\prec_T y$ and
  Def.\ \ref{def:timemap} implies that $\tT(x)<\tT(y)$.
  If $\mu(x)\in V(S)$ set $u\coloneqq \mu(x)$,
  otherwise let $u$ be the lower delimiting vertex of the edge
  $\mu(x)\in E(S)$. Similarly, set $v\coloneqq \mu(y)$ if $\mu(y)\in
  V(S)$,
  otherwise choose $v$ as the upper delimiting vertex of the edge
  $\mu(y)\in E(S)$.
  By time consistency, we have $\tS(u)\le\tT(x)$ and $\tT(y)\le\tS(v)$.
  Together with $\tT(x)<\tT(y)$, this yields $\tS(u)<\tS(v)$.
  Now assume, for contradiction, that $\mu(x)\succ_S\mu(y)$.
  One easily verifies that this implies $v\preceq_S u$ and thus $\tS(v)\le
  \tS(u)$; a contradiction.
\end{proof}

\begin{definition}
  The \emph{HGT-labeling} of a relaxed scenario $\scen$ is the map
  $\lambda:E(T)\to\{0,1\}$ such that $\lambda(uv)=1$ if and only if
  $\mu(u)$ and $\mu(v)$ are incomparable in $S$.
\end{definition}
We call an edge $e\in E(T)$ with $\lambda(e)=1$ an \emph{HGT edge}.

\begin{definition}
  For a relaxed scenario $\scen=(T,S,\sigma,\mu,\tT,\tS)$, we define the
  \emph{equal-divergence-time} (\emph{EDT}) graph $(\Gg(\scen),\sigma)$,
  the \emph{later-divergence-time} (\emph{LDT}) graph $(\Gu(\scen),\sigma)$
  and the \emph{prior-divergence-time} (\emph{PDT}) graph
  $(\Ga(\scen),\sigma)$ as follows: all graphs have as vertex set $L(T)$
  and are equipped with vertex coloring $\sigma\colon L(T)\to L(S)$.
  However, they differ in their edge sets defined as
  \begin{equation}
    \begin{split}
    E(\Gg(\scen)) &\coloneqq
          \left\{xy \mid x\ne y \textrm{ and }
          \tT(\lca_T(x,y))=\tS(\lca_S(\sigma(x),\sigma(y))\right\},\\
    E(\Gu(\scen)) &\coloneqq
          \left\{xy \mid x\ne y \textrm{ and }
          \tT(\lca_T(x,y)) < \tS(\lca_S(\sigma(x),\sigma(y))\right\},\\
    E(\Ga(\scen)) &\coloneqq
          \left\{xy \mid x\ne y \textrm{ and }
          \tT(\lca_T(x,y)) > \tS(\lca_S(\sigma(x),\sigma(y))\right\}.\\
    \end{split}
  \end{equation}
  Moreover, we write
  $\graphs(\scen)=(\Gu(\scen),\Gg(\scen),\Ga(\scen),\sigma)$.
\end{definition}
A vertex-colored graph $(G,\sigma)$ is an equal-divergence-time (EDT)
graph, if there is a relaxed scenario $\scen=(T,S,\sigma,\mu,\tT,\tS)$ such
that $G = \Gg(\scen)$. In this case, we say that $\scen$ \emph{explains}
$(G,\sigma)$.  By construction, the edge sets of $\Gg(\scen)$,
$\Gu(\scen)$, and $\Ga(\scen)$ are pairwise disjoint and their union is the
edge set of the complete graph on $L(T)$.  This motivates the definition of
the following tuple of vertex-colored graphs.
\begin{definition}
  A \emph{(colored) graph $3$-partition}, denoted by $\graphs = (\Gu, \Gg,
  \Ga, \sigma)$, is an ordered tuple of three edge-disjoint graphs on the
  same vertex set $L$ and with coloring $\sigma\colon L\to M$ such that
  $E(\Gu)\cupdot E(\Gg) \cupdot E(\Ga) = \binom{L}{2}$ (i.e. every
    unordered pair of $L$ is an edge of exactly one of the three
    graphs).\\ 
    We say that $\graphs$ is \emph{explained} by a scenario
  $\scen$ if $\Gu=\Gu(\scen)$, $\Gg=\Gg(\scen)$, and $\Ga=\Ga(\scen)$.
\end{definition}
An example for a graph $3$-partition and a relaxed scenario that explains
it is shown in Figure~\ref{fig:graphs-example}.
\begin{figure}[t]
  \centering
  \includegraphics[width=0.55\textwidth]{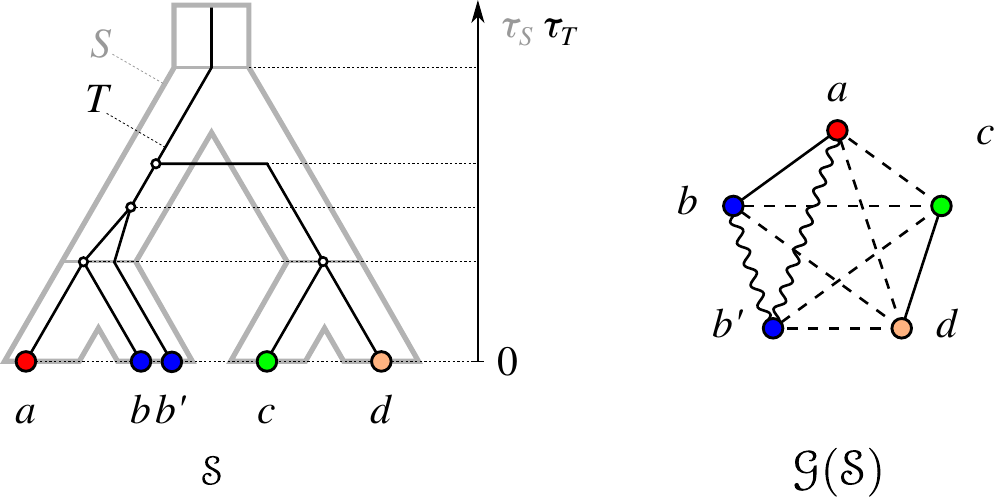}
  \caption{Left: A relaxed scenario $\scen=(T,S,\sigma,\mu,\tT,\tS)$. The
    maps $\mu$ and $\sigma$ are shown implicitly by the embedding of $T$
    into $S$ and the colors of the leaves of $T$, respectively. If a vertex
    $x$ is drawn higher than a vertex $y$, this means that $\tau(y)
    <\tau'(x)$, $\tau,\tau'\in\{\tT,\tS\}$. In the remainder of the paper,
    we will omit drawing the time axis explicitly. Right: The graph
    $3$-partition $\graphs(\scen)$ that is explained by
    $\scen$. Throughout, the edges of the LDT graph $\Gu(\scen)$, EDT graph
    $\Gg(\scen)$, and PDT graph $\Ga(\scen)$ will always be drawn as dashed,
    solid straight, and wavy lines, respectively.}
  \label{fig:graphs-example}
\end{figure}

The \emph{restriction} $\graphs_{\vert L'}$ of a graph $3$-partition $\graphs =
(\Gu, \Gg, \Ga, \sigma)$ to a subset $L'\subseteq L$ of vertices is given
by $(\Gu[L'], \Gg[L'], \Ga[L'], \sigma_{\vert L'})$.

\begin{lemma}
  \label{lem:gt-samecol}
  Let $\scen=(T,S,\sigma,\mu,\tT,\tS)$ be a relaxed scenario.  For all
  distinct vertices $x,y\in L(T)$ with $\sigma(x)=\sigma(y)$, it holds
  $xy\in E(\Ga(\scen))$.
\end{lemma}
\begin{proof}
  Since $x\ne y$, $u\coloneqq\lca_T(x,y)$ is not a leaf.  In particular,
  therefore, we have $\tT(x),\tT(y) < \tT(u)$ by the definition of time
  maps.  Moreover, we have $\tT(x)=\tS(\sigma(x))$ by the definition of
  scenarios.  If $\sigma(x)=\sigma(y)$, then
  $\lca_S(\sigma(x),\sigma(y))=\sigma(x)$ is a leaf and thus
  $\tS(\lca_S(\sigma(x),\sigma(y)))=\tS(\sigma(x))=\tT(x)<\tT(u)$.  Hence,
  $xy\in E(\Ga(\scen))$.
\end{proof}

The edge set of $\Gg(\scen)$, $\Gu(\scen)$, and $\Ga(\scen)$ are
disjoint. Lemma~\ref{lem:gt-samecol} therefore implies
\begin{corollary}
  \label{cor:equal-colors}
  Let $\scen=(T,S,\sigma,\mu,\tT,\tS)$ be a relaxed scenario.  If $xy\in
  E(\Gg(\scen))$ or $xy\in E(\Gu(\scen))$, then $\sigma(x)\ne\sigma(y)$,
  i.e., $\Gg(\scen)$ and $\Gu(\scen)$ are always \emph{properly colored}.
\end{corollary}

Hence, neither the class of EDT graphs nor the class of LDT graphs  is
closed under complementation because the complements of $\Gg(\scen)$ and
$\Gu(\scen)$ may contain edges between vertices with same color.

\subsection{Scenarios without HGT edges}

In order to connect our discussion to the ample literature on
DL-scenarios mentioned in the introduction, we briefly consider the
case of HGT-free scenarios.

\begin{lemma}
  If $\scen=(T,S,\sigma,\mu,\tT,\tS)$ is a relaxed scenario without
  HGT-edges, then $x\preceq_T y$ implies $\mu(x)\preceq_S \mu(y)$ for all
  $x,y\in V(T)$.
  \label{lem:no-hgt-compare}
\end{lemma}
\begin{proof}
  Suppose $\scen$ does not contain HGT-edges, i.e., $\mu(x)$ and $\mu(y)$
  are comparable in $S$ for all edges $yx\in E(T)$.  Two vertices $x,y\in
  V(T)$ with $x\preceq_T y$ are either equal, implying $\mu(x)=
  \mu(y)$, or they lie on a directed path $v_1\coloneqq y, v_2, \dots
  v_k\coloneqq x$ with $k\ge 2$.  If $yx\in E(T)$, then $x\prec_{T} y$
  implies $\mu(x)\preceq_S \mu(y)$ due to Lemma~\ref{lem:mu-timing}.  The
  vertices along a path in $T$ therefore satisfy $\mu(x)\preceq_S\dots
  \preceq_S \mu(v_2)\preceq_S \mu(y)$.  By transitivity of $\preceq_S$, we
  conclude that $x\prec_T y$ implies $\mu(x)\preceq_S\mu(y)$.
\end{proof}

\begin{lemma}
  If $\scen$ is a relaxed scenario without HGT-edges, then any pair of
  distinct leaves $x,y\in L(T)$ satisfies
  $\lca_S(\sigma(x),\sigma(y))\preceq_S \mu(\lca_T(x,y))$ and
  $\tS(\lca_S(\sigma(x),\sigma(y))) \leq \tT(\lca_T(x,y))$.  In particular,
  we have $\lca_S(\sigma(x),\sigma(y)) = \mu(\lca_T(x,y))$ if and only if
  $\tS(\lca_S(\sigma(x),\sigma(y))) = \tT(\lca_T(x,y))$, i.e., $xy\in
  E(\Gg)$.
  \label{lem:order}
\end{lemma}
\begin{proof}
  Consider an arbitrary pair of distinct vertices $x,y$ and
  $u\coloneqq\lca_T(x,y)\in V(T)$.  Then $x,y\preceq_T u$ and by Lemma
  \ref{lem:no-hgt-compare} we have $\mu(x)\preceq_S \mu(u)$ and
  $\mu(y)\preceq_S \mu(u)$. Since $x$ and $y$ are leaves, we have
  $\sigma(x)=\mu(x)$ and $\sigma(y)=\mu(y)$. The definition of the ancestor
  order and the last common ancestor now imply
  $\lca_S(\sigma(x),\sigma(y))\preceq_S \mu(u)$.  If
  $\lca_S(\sigma(x),\sigma(y)) = \mu(u)$, then time consistency
  implies $\tS(\lca_S(\sigma(x),\sigma(y))) = \tT(u)$.  Conversely,
  suppose $\lca_S(\sigma(x),\sigma(y))\prec_S \mu(u)$.  If
  $\mu(u)$ is a vertex $v$ of $S$, then we have
  $\tS(\lca_S(\sigma(x),\sigma(y))) < \tS(v)=\tT(u)$.  If
  $\mu(u)$ is an edge $vw$ of $S$ (with $w\prec_S v$), then we have
  $\tS(\lca_S(\sigma(x),\sigma(y))) \le\tS(w) <\tT(u) < \tS(v)$. In
  either case we therefore obtain $\tS(\lca_S(\sigma(x),\sigma(y)))
  <\tT(u)$.
\end{proof}

As an immediate consequence of Lemma~\ref{lem:order}, we recover
\cite[Cor.~6]{Schaller:21f}:
\begin{corollary}
  \label{cor:Gu-edgeless}
  If $\scen$ is a relaxed scenario without HGT-edges, then $\Gu(\scen)$ has
  no edges.
\end{corollary}

\subsection{Informative Triples}

If a graph 3-partition $\graphs=(\Gu, \Gg, \Ga, \sigma)$ is explained
  by some relaxed scenario ${\scen=(T,S,\sigma,\mu,\tT,\tS)}$, several
  structural constraints on $T$ and $S$ can be deduced directly from
  $\graphs$.  In particular, we show in this section that many subgraphs of
  $\graphs$ on three vertices enforce rooted triples that are either
  required or forbidden in $T$ or $S$.

\begin{lemma}
  \label{lem:species-triple-nohgt}
  Let $\scen=(T,S,\sigma,\mu,\tT,\tS)$ be a relaxed scenario without
  HGT-edges, suppose $\sigma(x)$, $\sigma(y)$, and $\sigma(z)$ are pairwise
  distinct, the triple $xy\vert z$ is displayed by $T$, and
  $\mu(\lca_T(x,z))=\lca_S(\sigma(x),\sigma(z))$. Then $S$ displays
  $\sigma(x)\sigma(y)\vert \sigma(z)$.
\end{lemma}
\begin{proof}
  By assumption $\lca_T(x,y)\prec_T \lca_T(x,z)=\lca_T(y,z)$.
  Lemma~\ref{lem:no-hgt-compare} implies
  $\mu(\lca_T(x,y))\preceq_S\mu(\lca_T(x,z))$ and Lemma~\ref{lem:v=w}
  implies $\mu(\lca_T(x,y))\ne\mu(\lca_T(x,z))$ and thus
  $\mu(\lca_T(x,y))\prec_T\mu(\lca_T(x,z))$. Moreover, by
  Lemma~\ref{lem:order} we have $\lca_S(\sigma(x),\sigma(y))\preceq_S
  \mu(\lca_T(x,y))$. We therefore conclude
  $\lca_S(\sigma(x),\sigma(y))\preceq_S \mu(\lca_T(x,y))\prec_S
  \mu(\lca_T(x,z)) = \lca_S(\sigma(x),\sigma(z))$.  Therefore, $S$ displays
  $\sigma(x)\sigma(y)\vert \sigma(z)$.
\end{proof}
Lemma~\ref{lem:species-triple-nohgt} defines the ``informative species
triples'' \cite{HernandezRosales:12a,Hellmuth:17,lafond2020reconstruction} that play a key role for
the characterization of feasible reconciliation maps in a slightly
different setting.

We recall two results that link triples in $T$ with the LDT graph:
\begin{lemma}
  \label{lem:Gu-T-triples}
  \textnormal{\cite[Lemma~7]{Schaller:21f}} Let
  $\scen=(T,S,\sigma,\mu,\tT,\tS)$ be a relaxed scenario with pairwise
  distinct leaves $x,y,z\in L(T)$. If $xy \in E(\Gu(\scen))$ and $xz, yz
  \notin E(\Gu(\scen))$, then $T$ displays $xy\vert z$.
\end{lemma}
\begin{lemma}
  \label{lem:Gu-S-triples}
  \textnormal{\cite[Lemma~6]{Schaller:21f}} Let
  $\scen=(T,S,\sigma,\mu,\tT,\tS)$ be a relaxed scenario with leaves
  $x,y,z\in L(T)$ and pairwise distinct colors $X\coloneqq\sigma(x)$,
  $Y\coloneqq\sigma(y)$, and $Z\coloneqq\sigma(z)$.  If $xz, yz \in
  E(\Gu(\scen))$ and $xy \notin E(\Gu(\scen))$, then $S$ displays $XY\vert Z$.
\end{lemma}

For example, Lemma~\ref{lem:Gu-S-triples} applies to $b, c, d$ in
  Figure~\ref{fig:graphs-example}: $bc, bd \in E(\Gu(\scen))$, $cd \notin
  E(\Gu(\scen))$, and $\sigma(c) \sigma(d) \vert \sigma(b)$ is a triple of
  the species tree.  We next show a statement similar to
Lemma~\ref{lem:Gu-T-triples} for the corresponding PDT $\Ga(\scen)$:
\begin{lemma}
  \label{lem:Ga-T-triples}
  Let $\scen=(T,S,\sigma,\mu,\tT,\tS)$ be a relaxed scenario with pairwise
  distinct leaves $x,y,z\in L(T)$. If $xz, yz \in E(\Ga(\scen))$ and $xy
  \notin E(\Ga(\scen))$, then $T$ displays $xy\vert z$.
\end{lemma}
\begin{proof}
  Suppose $xz, yz \in E(\Ga(\scen))$ and $xy \notin E(\Ga(\scen))$.  Put
  $X\coloneqq\sigma(x)$, $Y\coloneqq\sigma(y)$, and $Z\coloneqq\sigma(z)$
  and observe that $X\ne Y$ by Cor.~\ref{cor:equal-colors}.  Assume for
  contradiction that $xy\vert z$ is not displayed by $T$. Hence, the tree $T$
  displays either $xz\vert y$ or $yz\vert x$ or
  $\lca_T(x,y)=\lca_T(x,z)=\lca_T(y,z)$.  One easily verifies that, in all
  three cases, it holds $\lca_T(x,y)\succeq_{T}\lca_T(x,z)$ and
  $\lca_T(x,y)\succeq_{T}\lca_T(y,z)$.  This together with the assumption
  that $xz, yz \in E(\Ga(\scen))$ and $xy \notin E(\Ga(\scen))$ and time
  consistency implies
  \begin{align*}
 	 &\tS(\lca_S(X,Y))\ge \tT(\lca_T(x,y)) \ge \tT(\lca_T(x,z)) >
  	 \tS(\lca_S(X,Z)) \text{ and} \\
  	 & \tS(\lca_S(X,Y))\ge \tT(\lca_T(x,y)) \ge \tT(\lca_T(y,z)) >
  \tS(\lca_S(Y,Z)).
  \end{align*}
  In particular, this implies that $Y\ne Z$ and $X\ne Z$, resp., and thus
  $X$, $Y$, and $Z$ are pairwise distinct.  Since $\lca_S(X,Y)$ and
  $\lca_S(X,Z)$ are both ancestors of $X$, they are comparable in
  $S$. Together with $\tS(\lca_S(X,Y)) > \tS(\lca_S(X,Z))$ and the
  definition of time maps, this implies $\lca_S(X,Y) \succ_S \lca_S(X,Z)$.
  Thus, $S$ displays the triple $XZ\vert Y$. By similar arguments, we obtain
  that $S$ also displays the triple $YZ\vert X$; a contradiction.  Hence, $T$
  must display $xy\vert z$.
\end{proof}

Again using Figure~\ref{fig:graphs-example} as an example, one can
  check that $T$ must display $ab|b'$ because of
  Lemma~\ref{lem:Ga-T-triples}.  Let us now consider the EDT graph:
\begin{lemma}
  \label{lem:Ge-T-triples}
  Let $\scen=(T,S,\sigma,\mu,\tT,\tS)$ be a relaxed scenario with pairwise
  distinct leaves $x,y,z\in L(T)$ and suppose that $xz, yz \in
  E(\Gg(\scen))$.  If $xy \notin E(\Gg(\scen))$, then $T$ displays neither
  $xz\vert y$ nor $yz\vert x$.  In particular, if $xy \in E(\Gu(\scen))$,
  then $T$ displays $xy\vert z$.
\end{lemma}
\begin{proof}
  Suppose that $xz, yz \in E(\Gg(\scen))$ and $xy \notin E(\Gg(\scen))$.
  Recall that $\Gg(\scen)$, $\Gu(\scen)$, and $\Ga(\scen)$ are pairwise
  edge-disjoint. Put $X\coloneqq\sigma(x)$, $Y\coloneqq\sigma(y)$, and
  $Z\coloneqq\sigma(z)$ and observe that $X\ne Z$ and $Y\ne Z$ by
  Cor.~\ref{cor:equal-colors}.  If $xy\in E(\Gu(\scen))$, then Lemma
  \ref{lem:Gu-T-triples} implies that $T$ displays $xy\vert z$ and thus,
  none of $xz\vert y$ or $yz\vert x$.  Now suppose $xy\in E(\Ga(\scen))$
  and assume, for contradiction that $T$ displays $xz\vert y$ and thus
  $\lca_T(x,z)\prec_T \lca_T(x,y)=\lca_T(y,z)$.  By assumption and time
  consistency, this implies $\tS(\lca_S(X,Z))=\tT(\lca_T(x,z)) <
  \tT(\lca_T(y,z)) = \tS(\lca_S(Y,Z))$.  The latter implies that $X\ne Y$
  and thus $X$, $Y$, and $Z$ are pairwise distinct.  Since $\lca_S(X,Z)$
  and $\lca_S(Y,Z)$ are both ancestors of $Z$, they are comparable in
  $S$. Together with $\tS(\lca_S(X,Z)) < \tS(\lca_S(Y,Z))$ and the
  definition of time maps, this implies $\lca_S(X,Z) \prec_S \lca_S(Y,Z)$.
  Thus, $S$ displays the triple $XZ\vert Y$.  Therefore, we have
  $\lca_S(X,Y)=\lca_S(Y,Z)$.  In summary, we obtain $\tS(\lca_S(X,Y)) =
  \tS(\lca_S(Y,Z)) = \tT(\lca_T(y,z)) = \tT(\lca_T(x,y))$; a contradiction
  to $xy\in E(\Ga(\scen))$.  Hence, $T$ does not display $xz\vert y$.  For
  similar reasons, $T$ does not display $yz\vert x$, which concludes the
  proof.
\end{proof}

The case that $xz, yz \in E(\Gg(\scen))$, $xy \in E(\Ga(\scen))$ and
$xy\vert z$ is not displayed by $T$ is not covered by Lemma
\ref{lem:Ge-T-triples}. To see that this situation is possible, consider
the trees $S=((X,Y),Z)$ and $T=(x,y,z)$ with $\sigma(x)=X$, $\sigma(y)=Y$
and $\sigma(z)=Z$.  Now choose $\mu$ such that $\mu(\rho_T) = \rho_S$.  One
easily verifies that $xz, yz \in E(\Gg(\scen))$ and $xy \in E(\Ga(\scen))$
while $T$ by construction does not displayed $xy\vert z$.

\begin{lemma}
  \label{lem:Ge-S-triples}
  Let $\scen=(T,S,\sigma,\mu,\tT,\tS)$ be a relaxed scenario with leaves
  $x,y,z\in L(T)$ and pairwise distinct colors $X\coloneqq\sigma(x)$,
  $Y\coloneqq\sigma(y)$, and $Z\coloneqq\sigma(z)$. Suppose that
  $xz, yz \in E(\Gg(\scen))$. If $xy \notin E(\Gg(\scen))$, then $S$
  displays neither $XZ\vert Y$ nor $YZ\vert X$.  If, in particular,
  $xy \in E(\Ga(\scen))$ then $S$ displays $XY\vert Z$.
\end{lemma}
\begin{proof}
  Suppose that $xz, yz \in E(\Gg(\scen))$ and $xy \notin E(\Gg(\scen))$.
  By Lemma~\ref{lem:Ge-T-triples}, $T$ does not display $xz\vert y$ or
  $yz\vert x$.  Suppose for contradiction that $S$ displays $XZ\vert Y$,
  i.e., $\lca_S(X,Z)\prec_S \lca_S(Y,Z)$.  This together with the
  assumption that $xz, yz \in E(\Gg(\scen))$ and time consistency implies
  $\tT(\lca_T(x,z)) = \tS(\lca_S(X,Z)) < \tS(\lca_S(Y,Z)) =
  \tT(\lca_T(y,z))$.  Since $\lca_T(x,z)$ and $\lca_T(y,z)$ are both
  ancestors of $z$, they must be comparable. This together with
  $\tT(\lca_T(x,z)) < \tT(\lca_T(y,z))$ yields $\lca_T(x,z)\prec_T
  \lca_T(y,z)$ and thus $T$ displays $xz\vert y$; a contradiction.
  Therefore, $S$ does not display $XZ\vert Y$. For similar reasons,
  $YZ\vert X$ is not displayed.

  Now assume in addition that $xy \in E(\Ga(\scen))$.  Since $T$ does not
  display $xz\vert y$ and $\lca_T(x,y)$ and $\lca_T(x,z)$ are both
  ancestors of $x$ and thus comparable, we have
  $\lca(x,y)\preceq_T\lca_T(x,z)$.  Now this together with time
  consistency, $xy \in E(\Ga(\scen))$, and $xz \in E(\Gg(\scen))$ yields
  $\tS(\lca_S(X,Y)) < \tT(\lca_T(x,y)) \le \tT(\lca_T(x,z)) =
  \tS(\lca_S(X,Z))$.  Since $\lca_S(X,Y)$ and $\lca_S(X,Z)$ are both
  ancestors of $X$, they are comparable in $S$. Together with
  $\tS(\lca_S(X,Y)) < \tS(\lca_S(X,Z))$ and the definition of time maps,
  this implies $\lca_S(X,Y) \prec_S \lca_S(X,Z)$.  Thus, $S$ displays the
  triple $XY\vert Z$.
\end{proof}

Finally, we consider the species triples implied by the PDT graph.
The following result in particular generalizes the last statement
  in Lemma~\ref{lem:Ge-S-triples} above.
\begin{lemma}
  Let $\scen=(T,S,\sigma,\mu,\tT,\tS)$ be a relaxed scenario with leaves
  $x,y,z\in L(T)$ and pairwise distinct colors $X\coloneqq\sigma(x)$,
  $Y\coloneqq\sigma(y)$, and $Z\coloneqq\sigma(z)$.  If $xy \in
  E(\Ga(\scen))$ and $xz, yz \notin E(\Ga(\scen))$, then $S$ displays
  $XY\vert Z$.
  \label{lem:Ga-T-triples-mixed}
\end{lemma}
\begin{proof}
  Recall that by construction $\Gu(\scen)$, $\Gg(\scen)$, and $\Ga(\scen)$
  are edge-disjoint.  If $xz, yz \in E(\Gu(\scen))$ or $xz, yz \in
  E(\Gg(\scen))$, the statement follows immediately from
  Lemma~\ref{lem:Gu-S-triples} and~\ref{lem:Ge-S-triples}, respectively.
  Now consider the case that $xz \in E(\Gu(\scen))$ and $yz \in
  E(\Gg(\scen))$.  Hence, we have $\tT(\lca_T(x,y)) > \tS(\lca_S(X,Y))$ and
  $\tT(\lca_T(y,z)) = \tS(\lca_S(Y,Z))$.  Moreover, $T$ displays $xz\vert y$ by
  Lemma~\ref{lem:Gu-T-triples} and thus $\lca_T(x,y)=\lca_T(y,z)$.  To
  summarize, we have $\tS(\lca_S(Y,Z)) = \tT(\lca_T(y,z)) =
  \tT(\lca_T(x,y)) > \tS(\lca_S(X,Y))$.  Since $\lca_S(X,Y)$ and
  $\lca_S(Y,Z)$ are both ancestors of $Y$, they are comparable in
  $S$. Together with $\tS(\lca_S(Y,Z)) > \tS(\lca_S(X,Y))$ and the
  definition of time maps, this implies $\lca_S(X,Y) \prec_S \lca_S(Y,Z)$.
  Thus, $S$ displays the triple $XY\vert Z$.  One proceeds similarly if $yz \in
  E(\Gu(\scen))$ and $xz \in E(\Gg(\scen))$.
\end{proof}

See $a, b', c$ in Figure~\ref{fig:graphs-example}, which enforce
  $\sigma(a) \sigma(b') \vert \sigma(c)$ in the species tree by
  Lemma~\ref{lem:Ga-T-triples-mixed}.  With the facts that we have
  gathered, we can now define our set of required and forbidden triples.

\begin{definition}
  \label{def:inf-forb-triples}
  Let $\graphs=(\Gu, \Gg, \Ga, \sigma)$ be a tuple of three graphs on the same
  vertex set $L$ and with vertex coloring $\sigma$.\\[5pt]
  The set of \emph{informative triples on $L$}, denoted by $\R_T(\graphs)$,
  contains a triple $xy\vert z$ if $x,y,z\in L$ and one of the following
  conditions holds
  \begin{enumerate}[noitemsep]
    \item[(a)] $xy \in E(\Gu)$ and $xz, yz\notin E(\Gu)$,
    \item[(b)] $xz, yz \in E(\Ga)$ and $xy \notin E(\Ga)$.
  \end{enumerate}
  The set of \emph{forbidden triples on $L$}, denoted by $\F_T(\graphs)$,
  contains a triple $xz\vert y$ (and by symmetry also $yz\vert x$) if $x,y,z\in L$
  and $xz, yz \in E(\Gg)$ and $xy \notin E(\Gg)$.\\[5pt]
  The set of \emph{informative triples on $\sigma(L)$}, denoted by
  $\R_S(\graphs)$, contains a triple $XY\vert Z$ if there are $x,y,z\in L$ with
  pairwise distinct colors $X\coloneqq\sigma(x)$, $Y\coloneqq\sigma(y)$,
  and $Z\coloneqq\sigma(z)$ and one of the following conditions holds
  \begin{enumerate}[noitemsep]
    \item[(a')] $xz, yz \in E(\Gu)$ and $xy \notin E(\Gu)$,
    \item[(b')] $xy \in E(\Ga)$ and $xz, yz \notin E(\Ga)$.
  \end{enumerate}
  The set of \emph{forbidden triples on $L(S)$}, denoted by
  $\F_S(\graphs)$, contains a triple $XZ\vert Y$ (and by symmetry also
  $YZ\vert X$) if there are $x,y,z\in L$ with pairwise distinct colors
  $X\coloneqq\sigma(x)$, $Y\coloneqq\sigma(y)$, $Z\coloneqq \sigma(z)$, and
  $xz, yz \in E(\Gg)$ and $xy \notin E(\Gg)$.
\end{definition}

The notation $\R_T$, $\F_T$, $\R_S$, and $\F_S$ in
Definition~\ref{def:inf-forb-triples} is motivated by
Proposition~\ref{prop:triples} below, which shows that the triples on $L$
and $L(S)$, resp., provide information of the gene tree $T$ and species
tree $S$ explaining $\graphs$, provided such trees exists.  Summarizing
Lemmas~\ref{lem:Gu-T-triples} to~\ref{lem:Ga-T-triples-mixed}, we obtain:
\begin{proposition}
  \label{prop:triples}
  Let $\scen=(T,S,\sigma,\mu,\tT,\tS)$ be a relaxed scenario and
  $\graphs=(\Gu(\scen), \Gg(\scen), \Ga(\scen), \sigma)$.  Then $T$ agrees
  with $(\R_T(\graphs), \F_T(\graphs))$ and $S$ agrees with
  $(\R_S(\graphs), \F_S(\graphs))$.
\end{proposition}

\section{The Cograph Structure}
\label{sect:cograph}

Cographs naturally appear as graph structures associated with
vertex-labeled trees and more generally in the context of binary relations
associated with reconciliations of gene trees and species trees. For
example, orthology graphs in scenarios without horizontal gene transfer are
cographs \cite{Hellmuth:13a}. As we shall see below, both $\Gu(\scen)$
  and $\Ga(\scen)$ are cographs for all relaxed scenarios $\scen$. In
  contrast, $\Gg(\scen)$ is a cograph only under some additional
  constraints.  It is, however, always a so-called perfect graph.
\begin{lemma}
  \label{lem:Gu-cograph}
  \textnormal{\cite[Lemma~8]{Schaller:21f}} Let
  $\scen=(T,S,\sigma,\mu,\tT,\tS)$ be a relaxed scenario. Then $\Gu(\scen)$
  is a cograph.
\end{lemma}
It may not come as a surprise, therefore, that an analogous result holds for
$\Ga(\scen)$:
\begin{lemma}
  \label{lem:Ga-cograph}
  Let $\scen=(T,S,\sigma,\mu,\tT,\tS)$ be a relaxed scenario. Then $\Ga(\scen)$
  is a cograph.
\end{lemma}
\begin{proof}
  Set $A\coloneqq \sigma(a)$, $B\coloneqq \sigma(b)$, $C\coloneqq
  \sigma(c)$, and $D\coloneqq \sigma(d)$.  Suppose for contradiction that
  $\Ga(\scen)$ is not a cograph, i.e., it contains an induced $P_4$
  $a-b-c-d$.  By Prop.~\ref{prop:triples}, $T$ displays the informative
  triples $ac\vert b$ and $bd\vert c$.  Hence, 
  $T_{\vert abcd}=((a,c),(b,d))$ and, therefore, $\lca_T(a,d) = \lca_T(b,c)$.
  Moreover, by Cor.~\ref{cor:equal-colors}, we know that $A\ne C$, $A\ne
  D$, and $B \ne D$.
  Therefore, we have to consider the cases
  (i) $\vert \{A,B,C,D\}\vert =4$;
  (ii) $\vert \{A,B,C,D\}\vert =2$;
  (iii) $\vert \{A,B,C,D\}\vert =3$ and $A=B$, $C$, and $D$ are pairwise
  distinct;
  (iv) $\vert \{A,B,C,D\}\vert =3$ and $A$, $B$, and $C=D$ are pairwise
  distinct; and
  (v) $\vert \{A,B,C,D\}\vert =3$ and $A$, $B=C$, and $D$ are pairwise
  distinct.
  
  In Case~(i), $A$, $B$, $C$, and $D$ are pairwise distinct.  By
  Prop.~\ref{prop:triples}, $S$ displays the informative triples $AB\vert
  D$ and $CD\vert A$ (see Definition~\ref{def:inf-forb-triples}.b').
  Thus, $S_{\vert ABCD} = ((A,B),(C,D))$ and we have $\lca_S(B,C) =
  \lca_S(A,D)$.  In Case~(ii), we must have $A=B$ and $C=D$. Thus, we again
  obtain $\lca_S(B,C)=\lca_S(A,D)$.

  In Case~(iii), Prop.~\ref{prop:triples} implies that $S$ displays the
  informative triple $CD\vert A (=CD\vert B)$. Thus, we have $\lca_S(B,C) =
  \lca_S(A,D)$.  In Case~(iv), Prop.~\ref{prop:triples} implies that $S$
  displays the informative triple $AB\vert D (=AB\vert C)$. Thus, we have
  $\lca_S(B,C) = \lca_S(A,D)$.  In Case~(v), Prop.~\ref{prop:triples}
  implies that $S$ displays the informative triples $AB\vert D$ and
  $CD\vert A (=BD\vert A)$. Since $S$ cannot displays both of these
  triples, Case~(v) can be immediately excluded.

  In Cases~(i)--(iv), we have $\lca_T(a,d) = \lca_T(b,c)$ and $\lca_S(B,C) =
  \lca_S(A,D)$. Together with $bc\in E(\Ga(\scen))$, it follows
  $\tT(\lca_T(a,d)) = \tT(\lca_T(b,c)) > \tS(\lca_S(B,C)) = \tS(\lca_S(A,D))$;
  a contradiction to $ad\notin E(\Ga(\scen))$.

  In summary, $\Ga(\scen)$ does not contain an induced $P_4$ and thus it is
  a cograph.
\end{proof}

Lemmas~\ref{lem:Gu-cograph} and \ref{lem:Ga-cograph} naturally suggest to
ask whether an analogous result holds for $\Gg(\scen)$, i.e., whether
the EDT graph is always a cograph. If this is the case,
$\{\Gu(\scen),\Gg(\scen),\Ga(\scen)\}$ form a ``cograph 3-partition'' in
the sense of \cite{Hellmuth:15q,hellmuth2018tree}.  As illustrated in
Fig.\ \ref{fig:2-colP4}, this is not the case in general. Therefore, we
investigate in the following conditions under which $\Gg(\scen)$ may or may
not be a cograph and their implications for the underlying tree structure.

\begin{lemma}
  \label{lem:2-colP4}
  Let $\scen=(T,S,\sigma,\mu,\tT,\tS)$ be a relaxed scenario.  If
  $(\Gg(\scen),\sigma)$ contains an induced $P_4$ $a-b-c-d$ on two colors,
  then $T_{\vert {abcd}} = ((a,d),b,c)$.
\end{lemma}
\begin{proof}
  By assumption and by Cor.~\ref{cor:equal-colors}, $A\coloneqq\sigma(a) =
  \sigma(c)$, $B\coloneqq\sigma(b) = \sigma(d)$, and $A\ne B$. Therefore,
  and since $ab, bc, cd\in E(\Gg(\scen))$, we have $\tT(\lca_T(a,b)) =
  \tT(\lca_T(b,c))=\tT(\lca_T(c,d))=\tS(\lca_S(A,B))$.
  Def.~\ref{def:timemap} together with $\tT(\lca_T(a,b)) =
  \tT(\lca_T(b,c))$ implies that we can have neither $\lca_T(a,b) \prec_T
  \lca_T(b,c)$ nor $\lca_T(b,c) \prec_T \lca_T(a,b)$.  Since $\lca_T(a,b)$
  and $\lca_T(b,c)$ are both ancestors of $b$ and thus comparable in $T$,
  we conclude $\lca_T(a,b) =\lca_T(b,c)$.  Similarly, we obtain
  $\lca_T(b,c) =\lca_T(c,d)$.  Moreover, since $ad\notin E(\Gg(\scen))$, we
  have $\tT(\lca_T(a,d))\ne\tS(\lca_S(A,B))=\tT(\lca_T(a,b))$ and thus
  $\lca_T(a,d) \ne \lca_T(a,b)$, which implies that $T$ displays one of the
  triples $t_1 = ab\vert d$ or $t'_1=ad\vert b$.  By similar arguments, $T$
  displays one of the triples $t_2 = cd\vert a$ or $t'_2=ad\vert c$.  We
  next examine the possible combination of these triples.
 
  If $T$ displays $t_1$ and $t_2$, then $T_{\vert abcd} = ((a,b),(c,d))$, in
  which case $\lca_T(a,b) \neq \lca_T(b,c)$; a contradiction.  If $T$
  displays $t_1$ and $t'_2$, then $T_{\vert abcd} = (((a,b),d),c)$.  Again
  $\lca_T(a,b) \neq \lca_T(b,c)$; again a contradiction.  If $T$ displays
  $t'_1$ and $t_2$, then $T_{\vert abcd} = (((c,d),a),b)$. Hence $\lca_T(a,b)
  \neq \lca_T(c,d)$; a contradiction.  If $T$ displays $t'_1$ and $t'_2$,
  then $T_{\vert abcd}$ is either of the form $(((a,d),c),b)$,
  $(((a,d),b),c)$, $((a,d),b,c)$, or $((a,d),(b,c))$.  For the first two
  cases, we obtain $\lca_T(a,b)\neq \lca_T(c,d)$, while for the latter case
  we obtain $\lca_T(b,c)\neq \lca_T(c,d)$. Thus we reach a contradiction in all 
  three cases, leaving $T_{\vert abcd} = ((a,d),b,c)$ as the only possibility.
\end{proof}

\begin{figure}[t]
  \centering
  \includegraphics[width=0.4\textwidth]{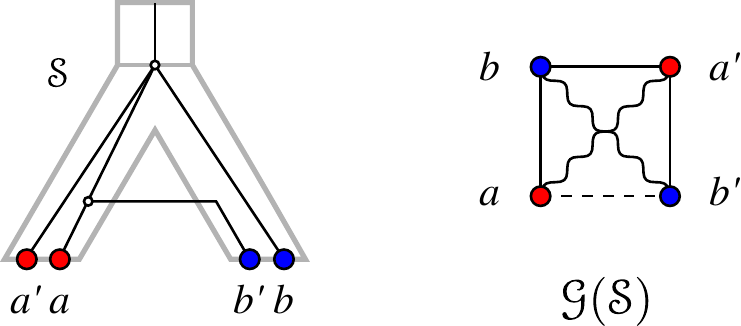}
  \caption{$(\Gg(\scen),\sigma)$ can contain a 2-colored $P_4 = a-b-a'-b'$.
    However, due to Cor.~\ref{cor:2colP4-binary}, $T$ cannot be a binary
    tree in this case.}
  \label{fig:2-colP4}
\end{figure}

Note that the tree $T_{\vert {abcd}} = ((a,d),b,c)$ in
Lemma~\ref{lem:2-colP4} is displayed by $T$ but not binary. Hence, we
obtain
\begin{corollary}
  Let $\scen=(T,S,\sigma,\mu,\tT,\tS)$ be a relaxed scenario. If $(\Gg(\scen),\sigma)$ contains a
  2-colored $P_4$, then $T$ is not a binary tree.
  \label{cor:2colP4-binary}
\end{corollary}

\begin{lemma}
  Let $\scen=(T,S,\sigma,\mu,\tT,\tS)$ be a relaxed scenario.  If $(\Gg(\scen),\sigma)$ contains an
  induced $P_4$ $a-b-c-a'$ on three distinct colors with
  $A=\sigma(a)=\sigma(a')$, $B=\sigma(b)$, and $C=\sigma(c)$, then
  $S_{\vert ABC} = (A,B,C)$. In particular, $S$ is not a binary tree.  Moreover,
  we have $T_{\vert {abca'}} = ((a,c),(b,a'))$.
  \label{lem:3-colP4-1}
\end{lemma}
\begin{proof}
  By assumption $P_3 = a-b-c$ is an induced path.
  Lemma~\ref{lem:Ge-S-triples} thus imply that $S$ does not display $AB\vert C$
  and $BC\vert A$.  Similarly, the induced $P_3 = b-c-a'$ implies that $S$ does
  not display $BC\vert A$ and $AC\vert B$. This leaves $S_{\vert ABC} = (A,B,C)$ as the
  only possibility.  By Lemma~\ref{lem:Ge-S-triples}, we immediately see
  that $ac, ba'\in \Gu(\scen)$ since otherwise $S$ would display $AC\vert B$ or
  $AB\vert C$.  This, together with Lemma~\ref{lem:Gu-T-triples} and $ab,bc,ca'
  \notin \Gu(\scen)$, implies that $T$ displays $ac\vert b$ and $ba'\vert c$ and, 
  therefore, $T_{\vert {abca'}} = ((a,c),(b,a'))$.
\end{proof}

\begin{figure}[t]
  \centering
  \includegraphics[width=0.8\textwidth]{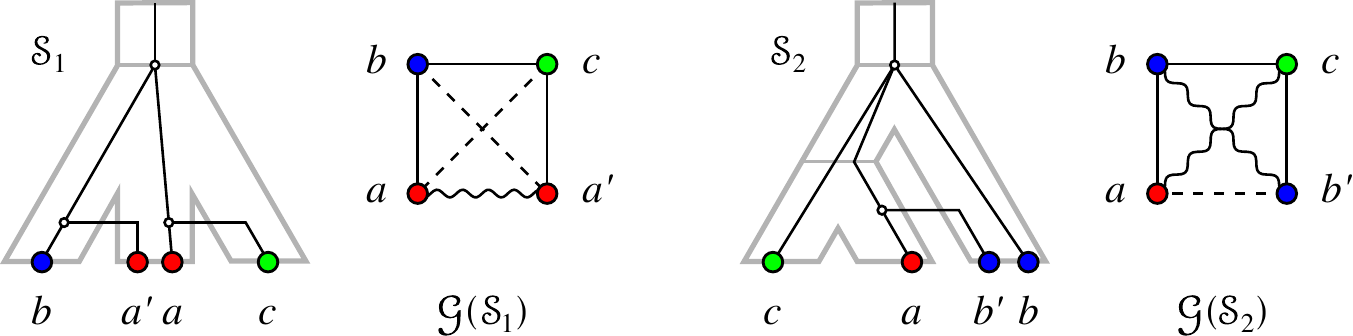}
  \caption{Left: $(\Gg(\scen_1),\sigma)$ contains an induced path $P_4 =
    a-b-c-a'$ on three colors with $\sigma(a)= \sigma(a')$ as in
    Lemma~\ref{lem:3-colP4-1}. Right: $(\Gg(\scen_2),\sigma')$ contains an
    induced path $P_4 = a-b-c-b'$ on three colors with $\sigma'(b)=
    \sigma'(b')$ as in Lemma~\ref{lem:3-colP4-2}.}
  \label{fig:EDT-P4-3col}
\end{figure}

\begin{lemma}
  Let $\scen=(T,S,\sigma,\mu,\tT,\tS)$ be a relaxed scenario.  If $E(\Gg(\scen))$ contains an
  induced $P_4$ $a-b-c-b'$ on three distinct colors with $A=\sigma(a)$,
  $B=\sigma(b)=\sigma(b')$ and $C=\sigma(c)$, then $S_{\vert ABC} = ((A,C),B)$
  and $T_{\vert abcb'}=((a,b'),b,c)$.
  \label{lem:3-colP4-2}
\end{lemma}
\begin{proof}
  Suppose that $\Gg(\scen)$ contains an induced $P_4$ $a-b-c-b'$ on three
  distinct colors $A=\sigma(a)$, $B=\sigma(b)=\sigma(b')$, and
  $C=\sigma(c)=C$. By Lemma~\ref{lem:Ge-T-triples}, $T$ displays
  neither $bc\vert b'$ nor $b'c\vert b$. Hence, we have to consider two cases: (1)
  $T_{\vert bcb'}=(b,c,b')$, or (2) $T_{\vert bcb'}=bb'\vert c$. By similar arguments,
  we have either (I) $T_{\vert abc}=(a,b,c)$ or
  (II) $T_{\vert abc}=ac\vert b$. We proceed by combining these alternatives:

  Case (1,I) yields (i) $T_{\vert abcb'}=(a,b,c,b')$ or (ii)
  $T_{\vert abcb'}=((a,b'),b,c)$, Case (1,II) yields $T_{\vert abcb'}=((a,c),b,b')$,
  Case (2,I) yields $T_{\vert abcb'}=((b,b'),a,c)$, and Case (2,II) yields
  $T_{\vert abcb'} = ((b,b'),(a,c))$. In all cases except Case (1,I,ii), we have
  $\lca_T(a,b)=\lca_T(a,b')$ and $ab\in E(\Gg(\scen))$ thus
  implies $\tS(\lca_S(A,B)) = \tT(\lca_T(a,b))=\tT(\lca_T(a,b'))$
  and $ab'\in E(\Gg(\scen))$; a contradiction. This leaves Case
  (1,I,ii), $T_{\vert abcb'}=((a,b'),b,c)$, as the only possibility.
  Lemma~\ref{lem:Ge-S-triples} together with $ab,bc\in E(\Gg(\scen))$ and
  $ac\notin E(\Gg(\scen))$ implies that either $S_{\vert ABC} = (A,B,C)$ or
  $S_{\vert ABC} = AC\vert B$. In the first case, we have $\lca_S(A,C)=\lca_S(B,C)$.
  Together with $T_{\vert abcb'}=((a,b'),b,c)$ (and thus
  $\lca_T(b,c)=\lca_T(a,c)$) and $bc\in E(\Gg(\scen))$, we obtain
  $\tS(\lca_S(A,C)) = \tS(\lca_S(B,C)) = \tT(\lca_T(b,c)) =
  \tT(\lca_T(a,c))$.  Therefore, we must have $ac\in E(\Gg)$; a
  contradiction. In summary, therefore, we have $S_{\vert ABC} = ((A,C),B)$ and
  $T_{\vert abcb'}=((a,b'),b,c)$.
\end{proof}
Fig.~\ref{fig:EDT-P4-3col} shows two examples of scenarios that realize EDT
graphs containing $P_4$s on three colors as described in
Lemma~\ref{lem:3-colP4-1} and Lemma~\ref{lem:3-colP4-2}, respectively.

Instead of considering the three graphs $\Gu$, $\Gg$, and $\Ga$ in
isolation, we can alternatively think of a graph $3$-partition
$\graphs=\{\Gu,\Gg,\Ga,\sigma\}$ as a complete graph $K_n$ whose edges are
colored with three different colors depending on whether they are contained
in $E(\Gu)$, $E(\Gg)$, or $E(\Ga)$. This links our results to the
literature on edge-colored graphs.  Complete edge-colored permutation
graphs are characterized \cite{Hartmann:22a} as the edge-partitions of
$K_n$ such that (i) each color class induces a permutation graph in the
usual sense \cite{Bose:98}, and (ii) the edge coloring is a Gallai
coloring, i.e., it contains no ``rainbow triangle'' with three distinct
colors.  While every cograph is also a permutation graph \cite{Bose:98},
rainbow triangles may appear in the edge-coloring defined by
$\{\Gu,\Gg,\Ga\}$ that is explained by a relaxed scenario. In fact, induced
$P_4$s in $\Gg$ are always associated with rainbow triangles.
\begin{lemma}
  \label{lem:rainbow-triangles}
  Let $\scen=(T,S,\sigma,\mu,\tT,\tS)$ be a relaxed scenario.  If
  $\Gg(\scen)$ contains an induced $P_4 = a-b-c-d$, then either $ad\in
  E(\Gu(\scen))$ and $ac,bd \in E(\Ga(\scen))$ or $ad\in E(\Ga(\scen))$ and
  $ac,bd \in E(\Gu(\scen))$. In either case, both $\{a,b,d\}$ and
  $\{a,c,d\}$ are rainbow triangles.
\end{lemma}
\begin{proof}
  Suppose $\Gg\coloneqq\Gg(\scen)$ contains an induced $P_4 = a-b-c-d$ and,
  therefore, $ac,ad,bd \notin E(\Gg)$.  Since $\Gg$, $\Gu\coloneqq
  \Gu(\scen)$ and $\Ga\coloneqq \Ga(\scen)$ are edge-disjoint, and $\Gu$
  and $\Ga$ are cographs (cf.\ Lemmas~\ref{lem:Gu-cograph}
  and~\ref{lem:Ga-cograph}), the cases $ac,ad,bd \in E(\Gu)$ and $ac,ad,bd
  \in E(\Ga)$ are not possible because otherwise $b-d-a-c$ is an induced
  $P_4$.  Moreover, $ab\vert c, bc\vert a \in \F_T(\graphs(\scen))$ as well
  as $cd\vert b,bc\vert d \in \F_T(\graphs(\scen))$ and thus $T$ displays
  neither of these two triples by Prop.~\ref{prop:triples}.  We consider
  two cases:
  
  If $ad\in E(\Gu)$ then at most one of the edges $ac$ and $bd$ can be
  contained in $\Gu$.  Suppose, for contradiction, that $ac \in E(\Ga)$ and
  $bd \in E(\Gu)$. Then $ad, bd\in E(\Gu)$ and $ac,bc,cd\notin E(\Gu)$.
  Prop.~\ref{prop:triples} implies that $T$ displays the informative
  triples $ad\vert c$ and $bd\vert c$.  Hence, $T$ also displays $ab\vert
  c$; a contradiction to $ab\vert c \in \F_T(\graphs(\scen))$.  By similar
  arguments, $ac \in E(\Gu)$ and $bd \in E(\Ga)$ implies that $T$ displays
  $cd\vert b$; a contradiction to $cd\vert b \in \F_T(\graphs(\scen))$.
  This leaves $ac,bd \in E(\Ga)$ as the only possible case.
  
  If $ad\in E(\Ga)$ then at most one of the edges $ac$ and $bd$ can be
  contained in $\Ga$.  Suppose, for contradiction, that $ac \in E(\Ga)$ and
  $bd \in E(\Gu)$. Then $bd\in E(\Gu)$ and $ab, ad\notin E(\Gu)$.
  Prop.~\ref{prop:triples} implies that $T$ displays $bd\vert a$. Moreover,
  $ac,ad\in E(\Ga)$ and $cd\notin E(\Ga)$ imply that $T$ displays $cd\vert
  a$.  Thus, $T$ displays $bc\vert a$; a contradiction.  By similar
  arguments, $ac \in E(\Gu)$ and $bd \in E(\Ga)$ implies that $T$ displays
  $bc\vert d$; a contradiction to $bc\vert d \in \F_T(\graphs(\scen))$.
  Again, we are left with $ac,bd \in E(\Gu)$ as the only possibility.
  
  In summary, we have $ad\in E(\Gu)$ and $ac,bd \in E(\Ga)$ or $ad\in
  E(\Ga)$ and $ac,bd \in E(\Gu)$, and thus both $\{a,b,d\}$ and $\{a,c,d\}$
  form a rainbow triangle in the edge coloring defined by $\graphs(\scen)$.
\end{proof}
As an immediate consequence, we obtain
\begin{corollary}
  \label{cor:no-rainbow->cograph}
  If the edge-coloring defined by $\graphs(\scen)$ does not contain a
  rainbow triangle, then $\Gg(\scen)$ is a cograph.
\end{corollary}
The converse of Cor.~\ref{cor:no-rainbow->cograph}, however, is not true in
general. A counterexample is given in Fig.~\ref{fig:rdt-triangle}.

\begin{figure}[t]
  \centering
  \includegraphics[width=0.4\textwidth]{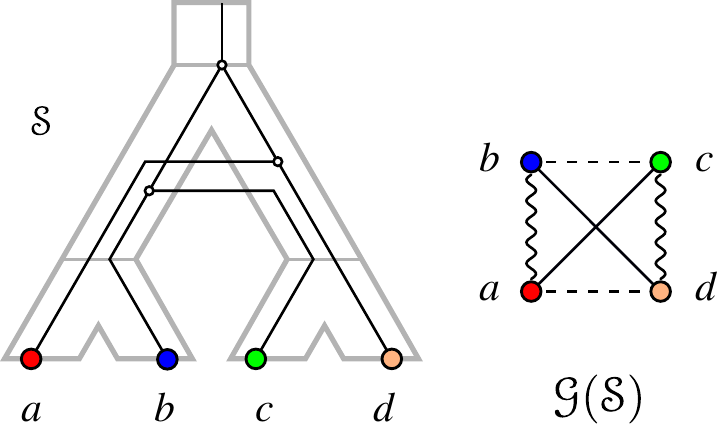}
  \caption{Example of a relaxed scenario $\scen$ and corresponding graph
    3-partition $\graphs(\scen)$ with $\graphs(\scen)$ containing rainbow
    triangles and $\Gg(\scen)$ being a cograph.}
  \label{fig:rdt-triangle}
\end{figure}

\begin{lemma}
  Let $\scen=(T,S,\sigma,\mu,\tT,\tS)$ be a relaxed scenario.  Suppose that $(\Gg(\scen),\sigma)$
  contains an induced $P_4 = a-b-c-d$ on four distinct colors
  $\sigma(a)=A$, $\sigma(b)=B$, $\sigma(c)=C$, and $\sigma(d)=D$. Then,
  exactly one of the following alternatives holds:
  \begin{enumerate}[noitemsep,nolistsep]
  \item[(i)] $ad\in E(\Gu(\scen))$, $ac,bd \in E(\Ga(\scen))$,
    $S_{\vert ABCD} = 
    ((A,C),(B,D))$ and $T_{\vert abcd} = ((a,d),b,c)$ or 
  \item[(ii)] $ad\in E(\Ga(\scen))$, $ac,bd \in E(\Gu(\scen))$, $S_{\vert ABCD} 
    = ((A,D),B,C)$ and $T_{\vert abcd} = ((a,c),(b,d))$.
  \end{enumerate}
  \label{lem:4col-P4}
\end{lemma}
\begin{proof}
  Set $\graphs\coloneqq\graphs(\scen)$, $\Gu\coloneqq \Gu(\scen)$,
  $\Gg\coloneqq\Gg(\scen)$, and $\Ga\coloneqq\Ga(\scen)$.  By
  Lemma~\ref{lem:rainbow-triangles}, we have exactly one of the following
  two alternatives (i') $ad\in E(\Gu)$ and $ac,bd \in E(\Ga)$ or (ii')
  $ad\in E(\Ga)$, $ac,bd \in E(\Gu)$.
  
  \textit{Case~(i'):} Since $ac,bd \in E(\Ga)$ and $ab,bc,cd \notin
  E(\Ga)$, $S$ displays the informative triples $AC\vert B, BD\vert C \in
  \R_S(\graphs)$ by Prop~\ref{prop:triples}. Hence, $S_{\vert ABCD} =
  ((A,C),(B,D))$.  Furthermore, by
  Prop.~\ref{prop:triples}, $T$ displays $ad\vert b, ad\vert c\in \R_T(\graphs)$ and
  none of $ab\vert c, bc\vert a, bc\vert d, cd\vert b \in \F_T(\graphs)$.  If $T$ displays
  $ac\vert b$, then this together with $T$ displaying $ad\vert b$ 
   implies that $T$ also displays $cd\vert b$; a
  contradiction. Thus, it holds $T_{\vert abc}=(a,b,c)$.  Together with the
  fact that $T$ displays $ad\vert b$, this implies $T_{\vert abcd} = ((a,d),b,c)$.
  In summary, Case~(i) is satisfied.
  
  \textit{Case~(ii'):} Since $ac,bd \in E(\Gu)$ and $ab,bc,cd \notin
  E(\Ga)$, $T$ displays the informative triples $ac\vert b, bd\vert c \in
  \R_T(\graphs)$ by Prop.~\ref{prop:triples}. Hence, $T_{\vert abcd} =
  ((a,c),(b,d))$.  Furthermore, by
  Prop~\ref{prop:triples}, $S$ displays $AD\vert B, AD\vert C\in \R_S(\graphs)$ and
  none of $AB\vert C, BC\vert A, BC\vert D, CD\vert B \in \F_S(\graphs)$.  Re-using analogous
  arguments as for $T$ in Case~(i'), we conclude that $S_{\vert ABCD} =
  ((A,D),B,C)$.  In summary, Case~(ii) is satisfied.
\end{proof}

\begin{figure}[t]
  \centering
  \includegraphics[width=0.9\textwidth]{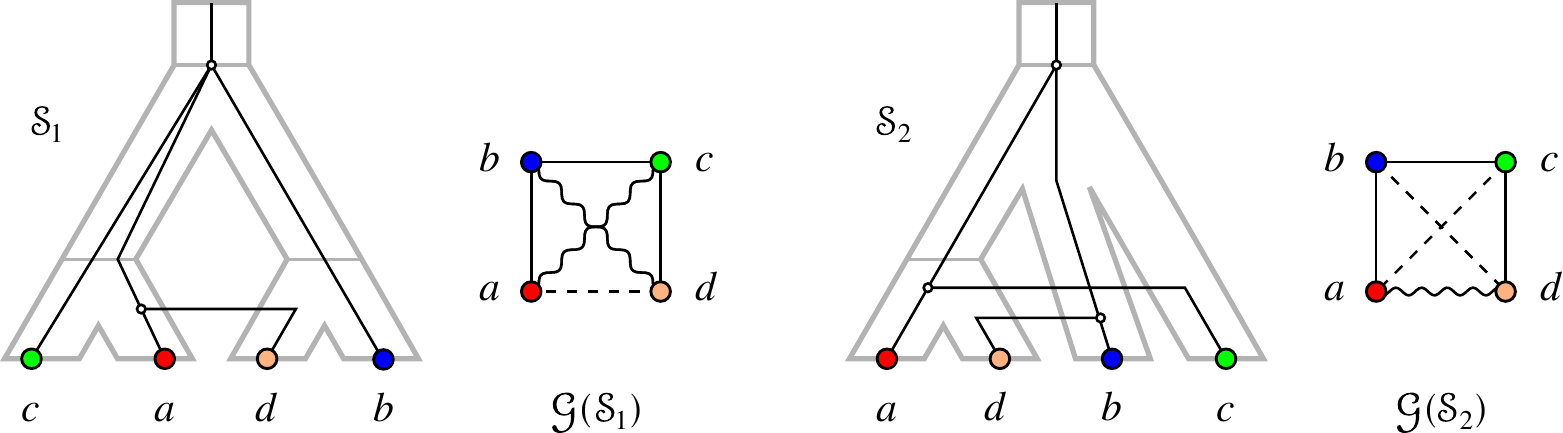}
  \caption{$(\Gg(\scen_1),\sigma)=(\Gg(\scen_2),\sigma)$ contains an
    induced path $P_4 = a-b-c-d$ on four colors as in
    Lemma~\ref{lem:4col-P4}.}
  \label{fig:EDT-P4-4col}
\end{figure}

Cor.~\ref{cor:equal-colors} implies that two adjacent vertices in
$\Gg(\scen)$ cannot have the same color. The $2-$, $3-$ and $4$-colored
$P_4$s considered in Lemmas~\ref{lem:2-colP4}, \ref{lem:3-colP4-1},
\ref{lem:3-colP4-2}, and \ref{lem:4col-P4} therefore cover all possible
colorings of an induced $P_4$ in $(\Gg(\scen),\sigma)$. Moreover, in each
case, the existence of a $P_4$ in $(\Gg(\scen),\sigma)$ implies that at
least one of $S$ and $T$ is non-binary.  We summarize this discussion and
Lemmas~\ref{lem:Gu-cograph} and~\ref{lem:Ga-cograph} in the following
\begin{theorem}\label{thm:binary->Gg-cograph}
  Let $\scen=(T,S,\sigma,\mu,\tT,\tS)$ be a relaxed scenario.  Then
  $\Gu(\scen)$ and $\Ga(\scen)$ are cographs. If both $S$ and $T$ are
  binary trees, then $\Gg(\scen)$ is also a cograph.
\end{theorem}

In the case of HGT-free scenarios, the condition that $S$ and $T$ are
binary is no longer necessary:
\begin{lemma}
  \label{lem:HGT-free-Gg-cograph}
  Let $\scen$ be a relaxed scenario without HGT-edges. Then $\Gg(\scen)$ is
  a cograph.
\end{lemma}
\begin{proof}
  By Cor.~\ref{cor:Gu-edgeless}, $\Gu(\scen)$ is edge-less. Therefore,
  $\Gg(\scen)$ is the complement of the cograph $\Ga(\scen)$
  (cf.\ Lemma~\ref{lem:Ga-cograph}) and thus, by Prop.\ \ref{prop:cograph},
  also a cograph.
\end{proof}

The similarities of $\graphs$ and edge-colored permutations graphs noted
above naturally lead to the question whether $\Gg$ is a permutation
graph. The example in Figure~\ref{fig:EDT-P5} shows that this is not the
case, however: The cycle on six vertices, $C_6$, is not a permutation graph
\cite{Gallai:67}.

\begin{figure}[t]
  \centering
  \includegraphics[width=0.9\textwidth]{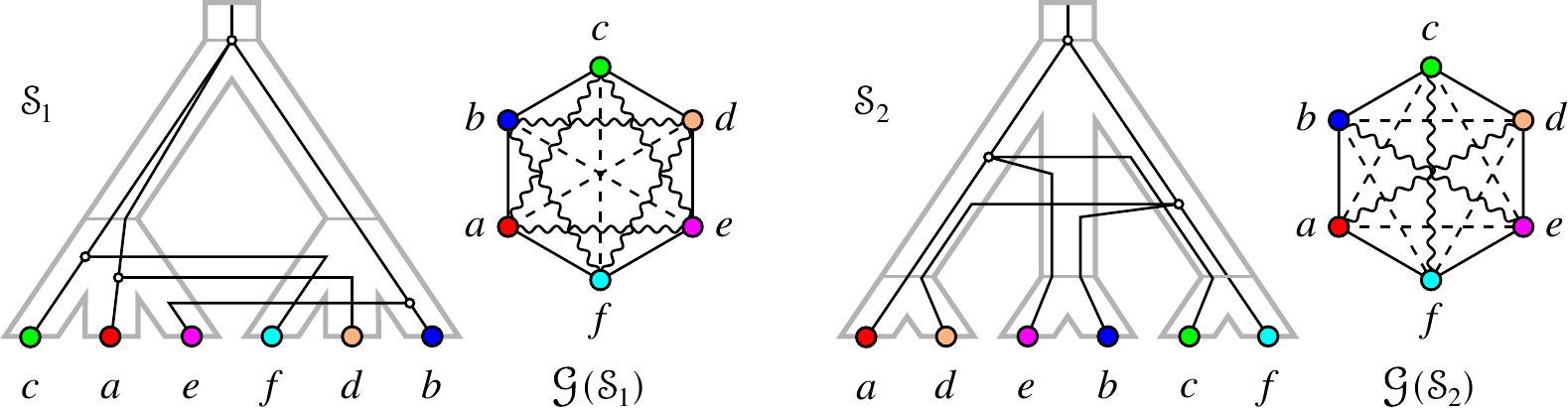}
  \caption{The EDT graph may contain an induced $C_6$, i.e, a cycle on six
    vertices. In this case, the EDT graph also contains induced $P_5$s.}
  \label{fig:EDT-P5}
\end{figure}

\begin{lemma}
  \label{lem:no-P6}
  If $\scen$ is a relaxed scenario, then $\Gg(\scen)$ does not contain an
  induced $P_6$.
\end{lemma}
\begin{proof}
  Set $\Gu\coloneqq \Gu(\scen)$, $\Gg\coloneqq\Gg(\scen)$, and
  $\Ga\coloneqq\Ga(\scen)$.  Suppose, for contradiction, that $\Gg$
  contains an induced $P_6=a-b-c-d-e-f$ (where the colors of these six
  vertices are not necessarily all distinct).  Since $a-b-c-d$ is an
  induced $P_4$ in $\Gg$ in this case, Lemma~\ref{lem:rainbow-triangles}
  implies that either (i) $ad\in E(\Gu)$ and $ac,bd \in E(\Ga)$ or (ii)
  $ad\in E(\Ga)$ and $ac,bd\in E(\Gu)$. Consider Case~(i).  Since $b-c-d-e$
  is an induced $P_4$ in $\Gg$ and $bd\in E(\Ga)$,
  Lemma~\ref{lem:rainbow-triangles} implies $be\in E(\Gu)$ and $ce \in
  E(\Ga)$. Repeating this argument for the induced $P_4$ $c-d-e-f$ in $\Gg$
  now yields $cf\in E(\Gu)$ and $df \in E(\Ga)$. Consider the pair $af$. If
  $af\in E(\Gu)$, then $\Gu$ contains the induced $P_4$ $d-a-f-c$, a
    contradiction to Lemma~\ref{lem:Gu-cograph}. Similarly, if $af\in
    E(\Ga)$, then $\Gu$ contains the induced $P_4$ $d-f-a-c$, a
    contradiction to Lemma~\ref{lem:Ga-cograph}.  Thus, only $af\in
  E(\Gg)$ remains, which contradicts that $a-b-c-d-e-f$ is an induced $P_6$
  in $\Gg$.  Case~(ii) is not possible for analogous reasons.  Hence,
  $\Gg(\scen)$ cannot contain an induced $P_6$.
\end{proof}
$P_6$-free graphs have been characterized in \cite{Liu:07,vantHof:10}.
Since any induced $P_k$ with $k\ge 6$ also contains an induced $P_6$,
Lemma~\ref{lem:no-P6} implies that the longest possible induced path in an
EDT graph has $5$ vertices. Figure~\ref{fig:EDT-P5} shows that this
situation can indeed be realized. In particular, the $P_5$s in these
examples are part of induced cycles on six vertices. Using
Lemma~\ref{lem:rainbow-triangles} and the arguments in the proof of
Lemma~\ref{lem:no-P6}, we can conclude that $\graphs(\scen_1)$ and
$\graphs(\scen_2)$, as shown in Fig.~\ref{fig:EDT-P5}, are the only two
configurations for an induced $C_6$ that can appear in an EDT graph.

A graph is \emph{odd-hole free} it it does not contain an induced cycle of
odd length greater than three \cite{Conforti:04}.
\begin{proposition} 
  \label{prop:odd-hole-free}
  If $\scen$ is a relaxed scenario, then $\Gg(\scen)$ does not contain an
  induced $C_5$ and induced cycles $C_\ell$ on $\ell \geq 7$ vertices.  In
  particular, EDT graphs are odd-hole free.
\end{proposition}
\begin{proof}
  Set $\Gu\coloneqq \Gu(\scen)$, $\Gg\coloneqq\Gg(\scen)$, and
  $\Ga\coloneqq\Ga(\scen)$.  Suppose, for contradiction, that $\Gg$
  contains an induced $C_5$ on vertices $a,b,c,d,e$ with $ab,bc,cd,de,ea\in
  E(\Gg)$. Thus, $a-b-c-d$ is an induced $P_4$ in $\Gg$ and
  Lemma~\ref{lem:rainbow-triangles} implies that either (i) $ad\in E(\Gu)$
  and $ac,bd \in E(\Ga)$ or (ii) $ad\in E(\Ga)$ and $ac,bd\in E(\Gu)$.  In
  Case~(i), we have $ad\in E(\Gu)$ and $ac,bd \in E(\Ga)$.  Since $b-c-d-e$
  is an induced $P_4$ in $\Gg$ and $bd\in E(\Ga)$,
  Lemma~\ref{lem:rainbow-triangles} implies $be\in E(\Gu)$ and $ce \in
  E(\Ga)$. Repeating this argument for the induced $P_4$ $c-d-e-a$ in $\Gg$
  now yields $ac\in E(\Gu)$; a contradiction.  Case~(ii) is not possible
  for analogous reasons.  Hence, $\Gg(\scen)$ cannot contain an induced
  $C_5$. Moreover, by Lemma \ref{lem:no-P6}, $\Gg(\scen)$ does not contain
  induced $P_6$s. Since every induced $C_{\ell}$ with $\ell\ge 7$ contains
  an induced $P_6$, such induced cycles cannot be part of an EDT graph. In
  particular, this implies that EDT graphs are odd-hole free.
\end{proof}

Prop.~\ref{prop:odd-hole-free} implies that not every $P_6$-free graph
$(G,\sigma)$ is an EDT graph, even if we restrict ourselves to
properly-colored graphs.  In particular, the cycle on $5$ vertices with
pairwise distinct colors is a properly colored $P_6$-free graph that is not
an EDT graph.  Moreover, the example in Fig.~\ref{fig:EDT-P5} shows that an
EDT graph may contain induced $C_6$s, i.e., they are in general not
even-hole free. Moreover, EDT graphs may contain induced $C_4$s. To see
this, consider the trees $T=((a_1,a_2),(b_1.b_2))$, $S=(A,B)$ and assume
that $\sigma(a_i)=A$ and $\sigma(b_i)=B$, $1\leq i\leq 2$. Now put
$\mu(\rho_T)=\rho_S$ and $\mu(\lca_T(a_1,a_2))=\rho_SA$ and
$\mu(\lca_T(b_1,b_2))=\rho_SB$.  It is now an easy exercise to verify that
$a_1, b_1, a_2, b_2$ form an induced $C_4$ in $\Gg$.

A graph $G$ is \emph{perfect}, if the chromatic number of every induced
subgraph equals the order of the largest clique of that subgraph
\cite{Berge:61}. A \emph{Berge graph} is a graph that contains no odd-hole
and no odd-antihole (the complement of an odd-hole) \cite{Chudnovsky:05}.
The strong perfect graph theorem \cite{Chudnovsky:06} asserts that a graph
is perfect iff it is a Berge graph. 

\begin{proposition}
  \label{prop:EDT-perfect-graph}
  If $\scen$ is a relaxed scenario, then $\Gg(\scen)$ is a perfect graph.
\end{proposition}
\begin{proof}
  By Prop.~\ref{prop:odd-hole-free}, $\Gg(\scen)$ is odd-hole free. By the
  strong perfect graph theorem, it suffices, therefore, to show that
  $\Gg(\scen)$ does not contain an odd-antihole.  Assume, for
  contradiction, that $\Gg(\scen)$ contains an odd-antihole $K$. Its
  complement $\overline K$ is, thus, an odd cycle that is entirely composed
  of edges of $\Gu(\scen)$ and $\Ga(\scen)$. Since $\overline K$ is a cycle
  of odd length $\ge5$, the edges along this cycle cannot be alternatingly
  taken from $\Gu(\scen)$ and $\Ga(\scen)$. In other words, in $\overline
  K$ there are at least two incident edges $ab,bc$ that are either both
  contained in $\Gu(\scen)$ or $\Ga(\scen)$. In addition, $\overline K$
  must contain an edge $cd$ and thus, $cd\notin E(\Gg(\scen))$. This,
  however, implies that $\Gg(\scen)$ contains an induced $P_4$
  $c-a-d-b$. By Lemma~\ref{lem:rainbow-triangles}, $\{c, a, b\}$
    should induce a rainbow triangle, which is a contradiction since $ab$
    and $bc$ are both either in the graph $\Gu(\scen)$ or $\Ga(\scen)$.
\end{proof}
Since perfect graphs are closed under complementation we obtain
\begin{corollary}
  If $\scen$ is a relaxed scenario, then $\Gu(\scen)\cup \Ga(\scen)$ is a
  perfect graph.
\end{corollary}	
The converse of Prop.~\ref{prop:EDT-perfect-graph} does not hold as shown
by the examples in Fig.~\ref{fig:no-EDT}, even under the restriction to
properly-colored graphs. Suppose the graph $(G,\sigma)$ in
Fig.~\ref{fig:no-EDT}(A) is explained by a relaxed scenario $\scen$. Put
$A\coloneqq \sigma(a)=\sigma(a')$, $B\coloneqq \sigma(b)=\sigma(b')$,
$C\coloneqq \sigma(c)=\sigma(c')$, and $D\coloneqq \sigma(d)=\sigma(d')$.
By Lemma~\ref{lem:4col-P4}, the induced $P_4 = a-b-c-d$ implies that
$S_{\vert ABCD} = ((A,C),(B,D))$ or $S_{\vert ABCD} = ((A,D),B,C)$, and the induced
$P_4=c'-a'-d'-b'$ implies that $S_{\vert ABCD} = ((A,B),(C,D))$ or $S_{\vert ABCD} =
((B,C),A,D)$; a contradiction. Clearly, $G$ contains no odd hole and no odd
antihole and, thus, it is a perfect graph.  Moreover, it is not
sufficient to require that $(G',\sigma)$ is a properly colored cograph.
To see this, suppose that the cograph $(G',\sigma')$ in
Fig.~\ref{fig:no-EDT}(B) is explained by a relaxed scenario $\scen$.  All
possible assignments for the edges $ac$ and $ad$ are shown on the
right-hand side, i.e., we have $ac\in E(\Ga(\scen))$, $ad\in
E(\Ga(\scen))$, or $ac,ad\in E(\Gu(\scen))$ yielding the informative
triples (for the species tree $S$) $AC\vert B$, $AD\vert B$, and $CD\vert A$,
respectively. However, all of these three triples are forbidden triples for
$S$ as a consequence of the three smaller connected components of
$(G',\sigma')$; a contradiction.
\begin{figure}[t]
  \centering
  \includegraphics[width=0.85\textwidth]{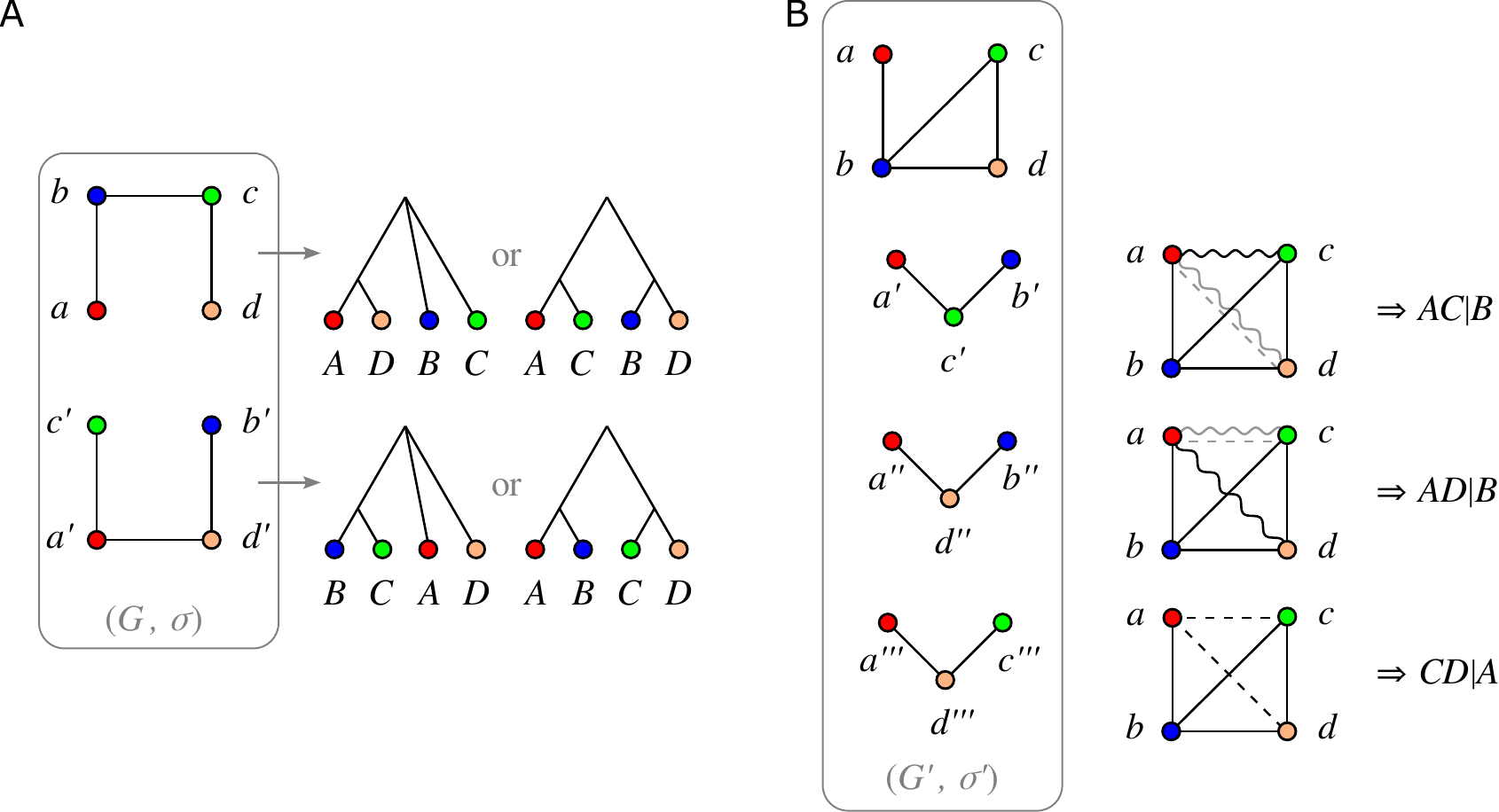}
  \caption{(A) A properly-colored perfect graph $(G, \sigma)$ on $8$
    vertices that is not an EDT graph. Next to the graph, the possible
    topologies of the species tree that are implied by the induced $P_4$
    according to Lemma~\ref{lem:4col-P4} are shown.  (B) A properly-colored
    cograph $(G',\sigma')$ that is not an EDT graph.  All possible
    assignments for the edges $ac$ and $ad$ are shown on the right-hand
    side together with the informative triples that they imply for the
    species tree according to Prop.~\ref{prop:triples}. The assignment of
    the gray edges do not affect the respective triple. }
  \label{fig:no-EDT}
\end{figure}

\section{Explanation of $\boldsymbol{\graphs}$ by Relaxed Scenarios}
\label{sect:main}

In \cite{Schaller:21f}, we derived an algorithmic approach that recognizes
LDT graphs and constructs a relaxed scenario $\scen$ for $(\Gu,\sigma)$ in
the positive case.  Here, we adapt the algorithmic idea to the case that,
instead of $(\Gu,\sigma)$, the graph 3-partition
$\graphs = (\Gu, \Gg, \Ga, \sigma)$ is given, see
Algorithm~\ref{alg:recognition}, which is illustrated in
  Figure~\ref{fig:algo-1-example}. As we shall see, the additional
information can be leveraged to separate the construction of $S$ and $T$ in
such a way that a suitable species tree can be computed first using a
well-known approach. This then considerably simplifies the construction of
a corresponding gene tree $T$.  More precisely, we construct the gene tree
and its reconciliation with $S$ in a top-down fashion via a recursive
decomposition of $L$ into subsets that is guided by $\graphs$ and $S$.  We
first introduce three auxiliary graphs that we will use for this purpose.

\begin{algorithm}[t]
  \small

  \caption{Construction of a relaxed scenario $\scen$ for a graph
    $3$-partition $\graphs = (\Gu, \Gg, \Ga, \sigma)$ with coloring
    $\sigma\colon L\to M$. \newline
    Input is a graph $3$-partition
    $\graphs = (\Gu, \Gg, \Ga, \sigma)$ with
    coloring $\sigma\colon L\to M$ such that $\Gu$ and $\Gg$ are properly
    colored, $\Gu$ and $\Ga$ are cographs, and $(\R_S(\graphs), \F_S(\graphs))$
    is consistent.\newline
    Output is a relaxed scenario $\scen=(T,S,\sigma,\mu,\tT,\tS)$ explaining
    $\graphs$.}    
  \label{alg:recognition}

\begin{algorithmic}[1]
  \State  $S\leftarrow$ tree on $M$ with planted root
    $0_S$ and agreeing with $(\R_S(\graphs), \F_S(\graphs))$\label{line:S}
  \State  $\tS\leftarrow$ time map for $S$ satisfying $\tS(x)=0$ for all
    $x\in L(S)$\label{line:tS}
  \State $\epsilon \leftarrow \frac{1}{3} \min\{\tS(y)-\tS(x) \mid yx\in
           E(S)\}$\label{line:epsilon}
  \State initialize empty maps $\mu, \tT$\;

  \Procedure{BuildGeneTree}{$L',u_{S}$} 
    \State create vertex $\rho'$\label{line:create-rho}
    \State $\tT(\rho')\leftarrow \tS(u_{S}) + \epsilon$ and
        $\mu(\rho')\leftarrow \parent_S(u_S) u_S$\label{line:mu-tT-inner1}
    \If{$u_S$ is a leaf\label{line:species-leaf}}
      \For{\textbf{each} $x\in L'$\label{line:loop-u_S-leaf}}
         \State connect $x$ as a child of $\rho'$\label{line:attach-leaf}
         \State $\tT(x)\leftarrow 0$ and $\mu(x)\leftarrow
                   \sigma(x)$\label{line:mu-tT-leaves}
      \EndFor 
    \Else \label{line:else}
      \State compute $H_1$, $H_2$, and $H_3$ for $L'$ and $u_S$
          \label{line:aux-graphs}\;
      \For{\textbf{each} connected component $C_i$ of $H_1$} \label{line:forH1}
        \State create vertex $u_i$ as a child of $\rho'$
               \label{line:create-u_i}
        \State $\tT(u_i)\leftarrow \tS(u_{S})$ and $\mu(u_i)\leftarrow u_S$
               \label{line:mu-tT-inner2}
        \For{\textbf{each} connected component $C_j$ of $H_2$ such
             that $C_j\subseteq C_i$}\label{line:iterate-C_j}
           \State create vertex $v_j$ as a child of $u_i$
                  \label{line:create-v_j}
           \State choose $v^*_S\in \child_S(u_{S})$ such that $\sigma(C_j)\cap
                  L(S(v^*_S))\ne\emptyset$\label{line:choose-v-S}
           \State $\tT(v_j)\leftarrow \tS(u_{S}) - \epsilon$ and
               $\mu(v_j)\leftarrow u_S v^*_S$\label{line:mu-tT-inner3}

          \For{\textbf{each} connected component $C_k$ of $H_3$ such that
            $C_k\subseteq C_j$} \label{line:forH3}
            \State identify $v_S\in \child_S(u_{S})$ such that
               $\sigma(C_k)\subseteq L(S(v_S))$
               \label{line:choose-v-S-for-class}
            \State connect \textsc{BuildGeneTree}$(C_k, v_S)$ as a child of
               $v_j$\label{line:recursive-call}
          \EndFor 
        \EndFor 
      \EndFor 
    \EndIf 
    \State \textbf{return} $\rho'$
  \EndProcedure

  \State $T' \leftarrow$ tree with root \textsc{BuildGeneTree}$(L,\rho_S)$ \label{line:callBGT}
  \State $T\leftarrow T'$ with (i) a planted root $0_T$ added, and (ii) all
     inner degree-2 vertices (except $0_T$) suppressed \label{line:Tphylo}
  \State $\tT(0_T)\leftarrow \tS(0_S)$ and $\mu(0_T)\leftarrow 0_S$
    \label{line:mu-tT-planted-root}
  \State \textbf{return}
    $(T,S,\sigma,\mu_{\vert V(T)},\tau_{T\vert V(T)},\tS)$\
\end{algorithmic}
\end{algorithm}

\begin{figure}[htbp]
  \centering
  \includegraphics[width=.95\textwidth]{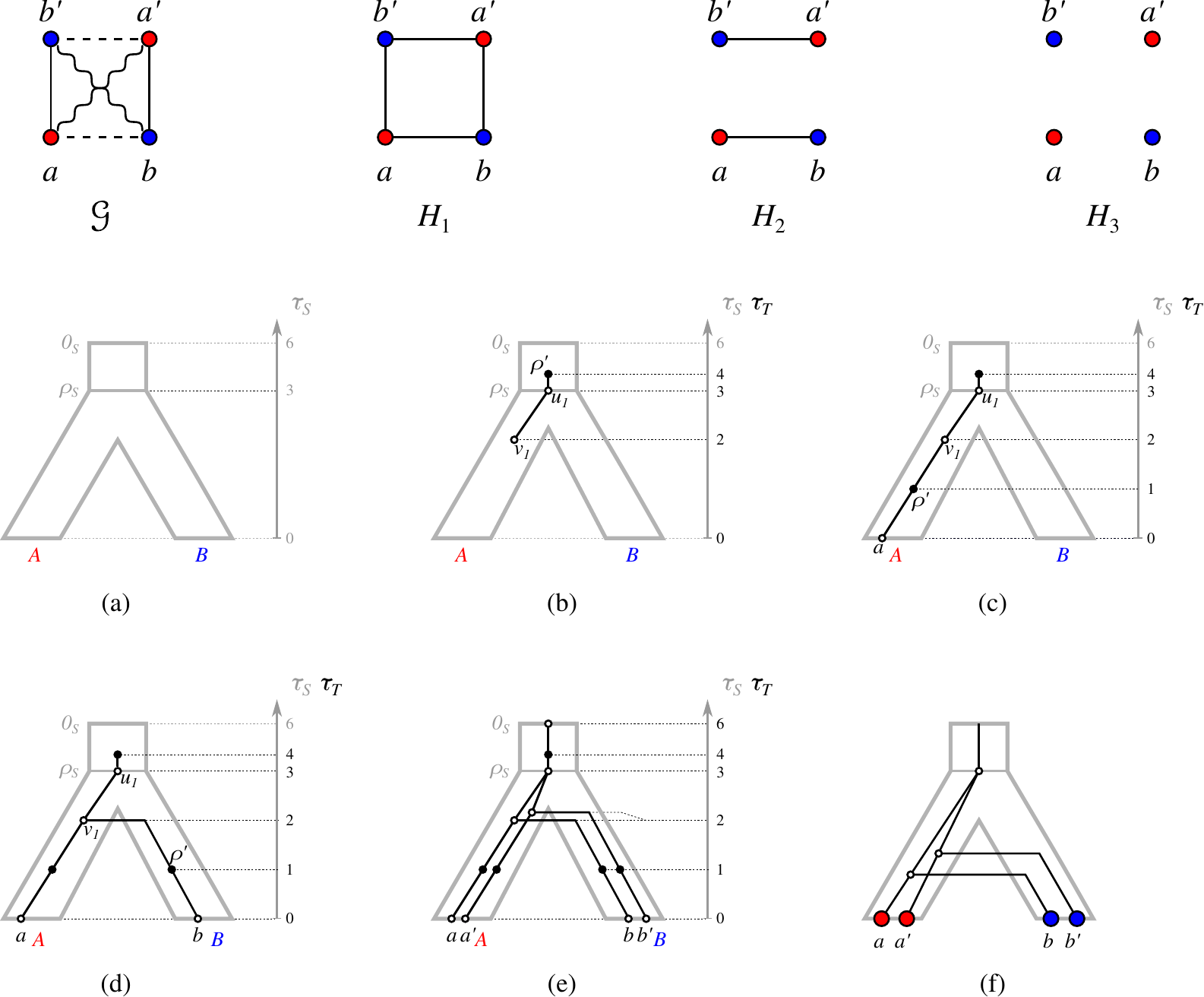}
  \caption{Illustration of Algorithm~\ref{alg:recognition} with a
      valid input $\graphs = (\Gu, \Gg, \Ga, \sigma)$. We have $\sigma(a) =
      \sigma(a')\eqqcolon A$ and $\sigma(b) = \sigma(b')\eqqcolon
      B$. Line~\ref{line:S} constructs species tree $S$ that agrees with
      $(\R_S(\graphs), \F_S(\graphs))$.  Here, $\R_S(\graphs) =
      \F_S(\graphs) = \emptyset$ and $S$ is unique. In Line~\ref{line:tS},
      a time map $\tS$ for $S$ such that $\tS(x)=0$ for all $x\in L(S)$ is
      initialized. We choose $\tS(0_S)=6$ and $\tS(\rho_S)=3$, see panel
      (a). Hence, $\epsilon=1$ (Line~\ref{line:epsilon}). \newline
      Line~\ref{line:callBGT} then calls
      \textsc{BuildGeneTree}$(\{a,a',b,b'\},\rho_S)$ for the first time,
      hence $u_S = \rho_S$. In Line~\ref{line:create-rho} a vertex $\rho'$
      is created. Its time map is set to $\tT(\rho') = \tS(\rho_{S}) +
      \epsilon = 3+1=4$ and the reconciliation is set to $\mu(\rho') = 0_S
      \rho_S$ in Line~\ref{line:mu-tT-inner1}. Since $u_S=\rho_S$ is not a
      leaf, we proceed with computing $H_1$, $H_2$, and $H_3$ for $L=
      \{a,a',b,b'\}$ and $\rho_S$ in Line~\ref{line:aux-graphs},
      illustrated in the top row.  Since $H_1$ has only one connected
      component $C$, the for-loop in Line~\ref{line:forH1} runs only
      once. In Line~\ref{line:create-u_i}, we thus create a single vertex
      $u_1$ as a child of $\rho'$. We then consider the two connected
      components $C_1$ and $C_2$ of $H_2$ as both satisfy $C_j\subseteq C$,
      $j\in \{1,2\}$. Here, we start with considering the component $C_1$
      that is induced by the vertices $a$ and $b$ and create a vertex $v_1$
      as a child of $u_1$ in Line~\ref{line:create-v_j}. We choose $v^*_S =
      A$ in Line \ref{line:choose-v-S} (note that we also could have chosen
      $v^*_S = B$) and set $\tT(v_1) = \tS(\rho_{S}) - \epsilon = 2$ and
      $\mu(v_1) = \rho_S A$ in Line~\ref{line:mu-tT-inner3}. These steps
      are illustrated in panel (b). \newline
      Line~\ref{line:forH3} then considers the connected components $C_k$
      of $H_3$ that satisfy $C_k\subseteq C_1 = \{a,b\}$; both of these
      connected components are the single vertex graphs induced by $a$ and
      $b$, respectively.  Starting with $C'=\{a\}$, 
      Line~\ref{line:choose-v-S-for-class} identifies $v_S\in\child_S(\rho_{S})$
      such that $\sigma(C')=\{A\}\subseteq L(S(v_S))$, i.e., $v_S = A$ and
      calls \textsc{BuildGeneTree}$(\{a\}, A)$; the subtree returned by
      this call is attached as a child of $v_1$ in 
      Line~\ref{line:recursive-call}.  Hence, we are now back in 
      Line~\ref{line:create-rho} where $u_S = A$.  In 
      Line~\ref{line:create-rho}, a further (new) vertex $\rho'$ is
      created. Line~\ref{line:mu-tT-inner1} computes
      $\tT(\rho') = \tS(A) + \epsilon = 0+1=1$ and $\mu(\rho') = \rho_S A$.
      Now $u_S=A$ is a leaf of $S$, hence we proceed in Line
      \ref{line:species-leaf} and connect each $x\in L' = \{a\}$ as a child
      of $\rho'$ in Line~\ref{line:attach-leaf}.  In 
      Line~\ref{line:mu-tT-leaves}, we put $\tT(a) = 0$ and
      $\mu(a)= \sigma(a)=A$.  These steps are illustrated in panel
      (c). \newline
      Then \textsc{BuildGeneTree}$(\{b\}, B)$ is executed and we obtain the
      ``partial'' gene tree and reconciliation shown in panel (d).  The
      algorithm proceeds on component $C_2$ of $H_2$, which is induced by the
      vertices $a'$ and $b'$ and creates a vertex $v_2$ as a child of $u_1$
      in Line~\ref{line:create-v_j}. Again, we chose $v^*_S = A$ in 
      Line~\ref{line:choose-v-S}. By similar arguments as in the previous part,
      we obtain the ``partial'' gene tree and reconciliation shown in panel
      (e). The tree $T'$ returned in Line~\ref{line:callBGT} is the gene
      tree shown in panel (e) except for the planted root $0_T$, which is
      added in Line~\ref{line:Tphylo}.  In addition, all resulting inner
      degree-2 vertices (highlighted as black circuits) are suppressed in
      Line~\ref{line:Tphylo}. The resulting gene tree (without specified
      time map) and the resulting relaxed scenario is shown in
      panel (f). Note, if we choose $v^*_S = B$ in 
      Line~\ref{line:choose-v-S} when proceeding on the connected component 
      $C_2$ of $H_2$ induced by $a'$ and $b'$, we would obtain the restricted 
      scenario $\scen_2$ as shown in Figure~\ref{fig:algo-S6-violated}.
  }
  \label{fig:algo-1-example}
\end{figure}

\begin{definition}
  Let $\graphs = (\Gu, \Gg, \Ga, \sigma)$ be a graph $3$-partition on vertex
  set $L$ with coloring $\sigma\colon L\to M$ and $S$ be a tree on $M$.\\
  For $L'\subseteq L$ and $u\in V^0(S)$ such that
  $\sigma(L')\subseteq L(S(u))$, we define the auxiliary graphs on $L'$:
  \begin{flalign*}
    H_1 \coloneqq\ &   (L', E(\Gu[L']) \cup E(\Gg[L']))&\\
    H_2 \coloneqq\ & (L', E(\Gu[L']) \cup \{xy \in E(\Gg[L']) \mid
    \sigma(x), \sigma(y)\prec_S v \text{ for some } v\in\child_S(u)\} )&\\
    H_3 \coloneqq\ & (L', \{xy \mid x \text{ and } y \text{ are in the same
      connected component of } H_2 \text{ and }  &\\ 
    	& \sigma(x), \sigma(y)\preceq_S 
    	v \text{ for some }
    v\in\child_S(u)\}) &
  \end{flalign*} 
\end{definition}

By construction, $H_2$ is a subgraph of $H_1$. In particular, therefore,
every connected component of $H_2$ is entirely included in some connected
component of $H_1$.  In turn, one easily verifies that the connected
components of $H_3$ are complete graphs. Moreover, $H_3$ contains all edges
of $H_2\cap \Gg$ while there might be edges of $\Gu[L']$ that are not
contained in $H_3$. This implies that every connected
component of $H_3$ is entirely included in some connected
component of $H_2$.

We use the inclusion relation of the connected components to construct the
local topology of $T$ in a recursive manner, see
Figure~\ref{fig:generalized-algo} for an illustration of the following
description.  In each step, i.e., for some $L'\subseteq L$ and $u_S\in
V(S)$, we first construct a ``local root'' $\rho'$
(cf.\ Algorithm~\ref{alg:recognition}, Line~\ref{line:create-rho}).  If
$u_S$ is a leaf of $S$ (the base case of the recursion), we directly attach
the elements of $L'$ as children of $\rho'$
(Lines~\ref{line:species-leaf}-\ref{line:mu-tT-leaves}).  On the other
hand, if $u_S$ is an inner vertex, we create a new child of $\rho'$ for
each connected component of $H_1$ in Line \ref{line:create-u_i}. For a
specific connected component $C_i$ of $H_1$ (corresponding to child $u_i$
of $\rho'$), we then add a new child $v_j$ of $u_i$ for each connected
component $C_j$ of $H_2$ such that $C_j\subseteq C_i$ in Line
\ref{line:create-v_j}.  We proceed similarly for the connected components
$C_k$ of $H_3$, which necessarily are subsets of a specific connected
component $C_j$ of $H_2$.  The vertex corresponding to $C_k$ is the ``local
root'' created in a recursive call operating on $C_k$ as new subset of $L$
and $v_S\in\child_S(u_S)$ as new vertex of $S$, which is chosen such that
$\sigma(C_k)\subseteq L(S(u_S))$ in Line
\ref{line:choose-v-S-for-class}. If $C_j=C_i$ or $C_k=C_j$, then the
corresponding vertices $v_i$ and $v_j$, respectively, have a single child.
As a consequence, the resulting tree $T'$ is in general not
phylogenetic. The final gene tree $T$ is then obtained by suppressing all
vertices with a single child (Line \ref{line:Tphylo}).

By definition, two vertices $x$ and $y$ are in the same connected component
$C_k$ of the auxiliary graph $H_3$ only if $\sigma(x)$ and $\sigma(y)$ are
descendants of the same child $v_S$ of the species tree vertex $u_S$.  In
particular, we therefore can always find $v_S\in \child_S(u_{S})$ such that
$\sigma(C_k)\subseteq L(S(v_S))$ in Line~\ref{line:choose-v-S-for-class} of
Algorithm~\ref{alg:recognition}.  This guarantees that all colors appearing on
the vertices in $L'$ are descendants of the species tree vertex $u_S$ in
each recursion step:
\begin{fact}
  \label{obs:color-subset}
  In every recursion step of Algorithm~\ref{alg:recognition}, it holds
  $\sigma(L')\subseteq L(S(u_s))$.  In particular, the auxiliary graphs
  $H_1$, $H_2$, and $H_3$ are always well-defined.
\end{fact}

\begin{figure}[t]
  \centering
  \includegraphics[width=0.65\textwidth]{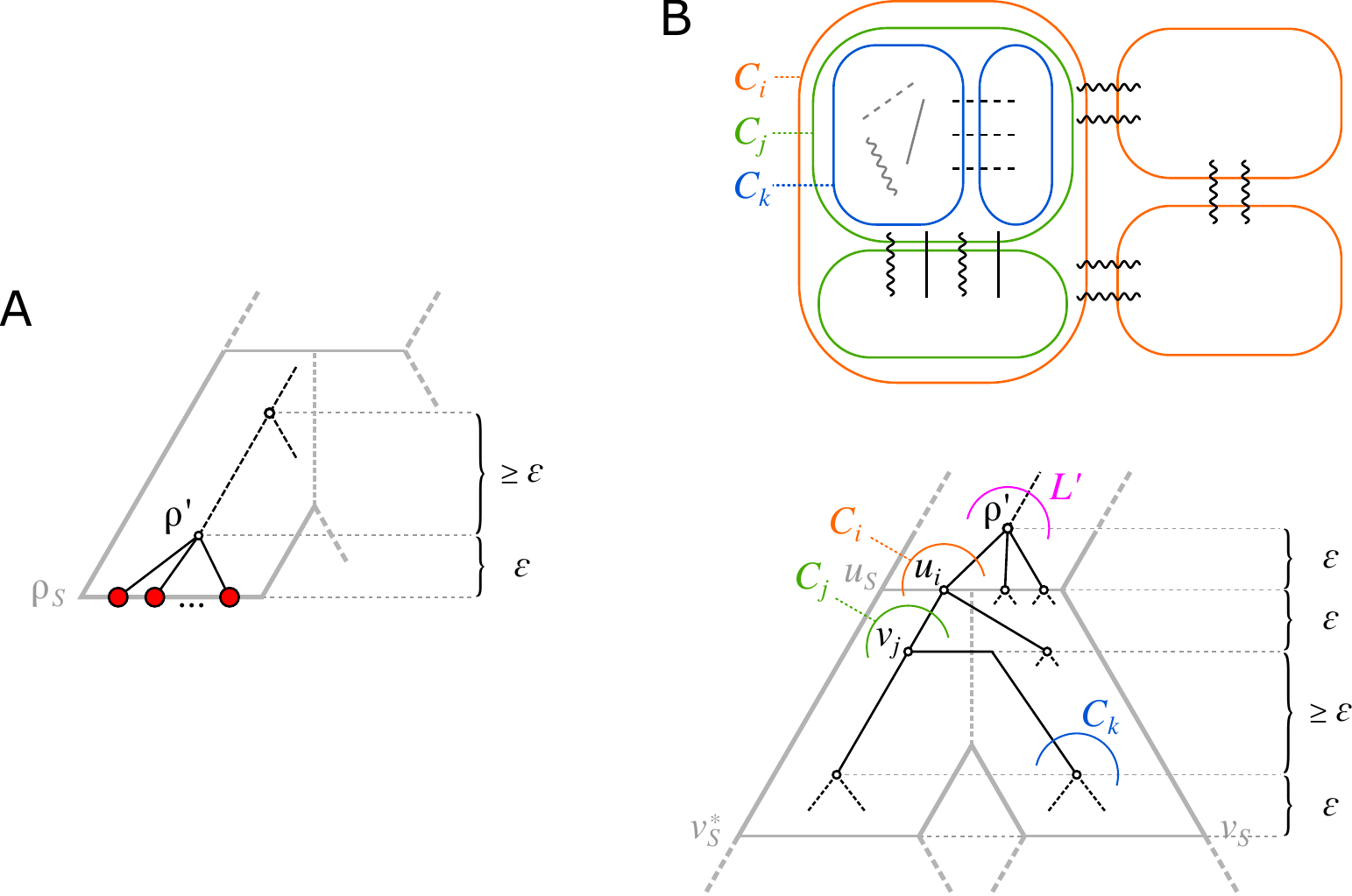}
  \caption{Illustration of a recursion step in \texttt{BuildGeneTree}
    (Algorithm~\ref{alg:recognition}). A: The current vertex in the
    species tree, $u_S$, is a leaf. B: The current vertex in the
    species tree, $u_S$, is an inner vertex. The connected components of
    $H_1$, $H_2$, and $H_3$ are represented by the orange, green, and blue
    boxes, respectively. For simplicity, only those connected components of
    $H_2$ and $H_3$ are shown that are included in $C_i$ and $C_j$,
    respectively. Two vertices $x$ and $y$ must form an edge in (a) $\Ga$
    (wavy lines) if they are in distinct components of $H_1$, (b) $\Ga$ or
    $\Gg$ (solid straight lines) if they are in the same component of $H_1$ but
    distinct components of $H_2$, and (c) $\Gu$ (dashed lines) if they are
    in the same component of $H_2$ but distinct components of $H_3$. Below,
    the construction of the reconciliation map and the time map is
    illustrated. The half circles indicate that $L'=L(T(\rho'))$,
    $C_i=L(T(u_i))$, etc.\ if the respective vertex is not suppressed.}
  \label{fig:generalized-algo}
\end{figure}

The recursion in Algorithm~\ref{alg:recognition} can be thought of as a
tree with the root being the top-level call of \texttt{BuildGeneTree} on
$L$ and $\rho_S$ and leaves being the calls in which $u_S$ is a leaf of
$S$.  Note that, for some recursion steps $R$ on $L'$ and $u_S$, all of its
``descendant recursion steps'' have input $L''$ and $u'_S$ satisfying
$L''\subseteq L'$ and $u'_S\prec_{S} u_S$.  Therefore, and because all
leaves that are descendants of $\rho'$ (created in $R'$) must have been
attached in some descendant recursion step of $R$, we have
$L(T'(\rho'))\subseteq L'$. In turn, all elements $x\in L'$ are either
directly attached to $\rho'$ if $u_S$ is a leaf, or will eventually be
passed down to a recursion step on a leaf $l\prec_{S} u_S$ because each
$x\in L'$ is in some connected component $C_k$ of $H_3$, $C_k$ is entirely
included in a connected component $C_j$ of $H_2$, and $C_j$ is entirely
included in a connected component $C_i$ of $H_1$. In this ``leaf recursion
step'', $x$ is therefore attached to some descendant of $\rho'$, implying
$L'\subseteq L(T'(\rho'))$. Therefore, we have $L'=L(T'(\rho'))$.  We can
apply very similar arguments to see that $L(T'(u_i))=C_i$ and
$L(T'(v_j))=C_j$ hold for each connected component $C_i$ of $H_1$ and $C_j$
of $H_2$ with corresponding vertices $u_i$ and $v_j$ created in
Lines~\ref{line:create-u_i} and~\ref{line:create-v_j}, respectively.
Clearly, contraction of the redundant vertices to obtain the final tree $T$
does not change these relationship. We summarize these considerations as
follows:
\begin{fact}
  \label{obs:leaf-sets}
  Let $T$ be a tree returned by Algorithm~\ref{alg:recognition} and $u\in 
  V^0(T)$ be an inner vertex created in a recursion step on $L'$ and $u_S$.	
  \begin{enumerate}[noitemsep,nolistsep]
  \item If $u$ is a vertex $\rho'$ created in Line~\ref{line:create-rho}, 
    then $L(T(u))=L'$.
  \item If $u$ is a vertex $u_i$ created in Line~\ref{line:create-u_i}, 
    then $L(T(u))=C_i$ where $C_i$ is the connected component of $H_1$ 
    corresponding to $u_i$.
  \item If $u$ is a vertex $v_j$ created in Line~\ref{line:create-v_j}, 
    then $L(T(u))=C_j$ where $C_j$ is the connected component of $H_2$ 
    corresponding to $v_j$.
  \end{enumerate}
  In particular, every $x\in L(T(u))$ satisfies $x\in L'$.
\end{fact}

Algorithm~\ref{alg:recognition} is a generalization of the
algorithm presented in \cite{Schaller:21f} for the construction of a
relaxed scenario $\scen$ for a given LDT graph $(G,\sigma)$. A key property
of the algorithm is that the restriction of $\scen$ to $S(u_S)$, i.e., the
incomplete scenarios obtained for given $u_S$ satisfies the time
consistency constraints \AX{S2} and \AX{S3}.  The construction of $\scen$
in Algorithm~\ref{alg:recognition} differs from the procedure described in
\cite{Schaller:21f} only by including in $V(T) \setminus (L(T) \cup
\{0_T\})$ the additional vertices $u_i$ created in
Line~\ref{line:create-u_i}.  These satisfy $\mu(u_i)=u_S$. In the following
line, we set $\tT(u_i)\leftarrow \tS(u_S)$. Hence, constraint \AX{S2}
remains satisfied and \AX{S3} is void because $\mu(u_i)\in V(S)$. One
easily checks, furthermore, that the reconciliation map $\mu$ constructed
in Algorithm~\ref{alg:recognition} satisfies \AX{S0}
(Line~\ref{line:mu-tT-planted-root}) and \AX{S1}
(Line~\ref{line:mu-tT-leaves}).

\begin{definition}
$\graphs = (\Gu, \Gg, \Ga, \sigma)$ is a \emph{valid input} for
Algorithm~\ref{alg:recognition} if $(\Gu, \Gg, \Ga, \sigma)$ is a
3-partition, $\Gu$ and $\Gg$ are properly colored, $\Gu$ and $\Ga$ are
cographs, and $(\R_S(\graphs), \F_S(\graphs))$ is consistent.
\end{definition}
\begin{lemma}
  \label{lem:algo-rs}
  Given a valid input $\graphs = (\Gu, \Gg, \Ga, \sigma)$ with vertex set
  $L$, Algorithm~\ref{alg:recognition} returns a relaxed scenario
  $\scen=(T,S,\sigma,\mu,\tT,\tS)$ such that $L(T)=L$.
\end{lemma}
\begin{proof}
  In order to keep this contribution self-contained, a detailed proof
  of Lemma~\ref{lem:algo-rs}, which largely parallels the material in
  \cite{Schaller:21f}, is given in Appendix~\ref{appx:additional-proofs}.
\end{proof}

We continue with a number of intermediate results that we will need to
establish the correctness of Algorithm~\ref{alg:recognition}.

\begin{lemma}
  \label{lem:exit-Ga-edge}
  Let $\graphs = (\Gu,\Gg,\Ga,\sigma)$ with vertex set $L$ be a valid input
  for Algorithm~\ref{alg:recognition}.  Consider a recursion step on
  $L'\subseteq L$ and $u_S\in V^0(S)$ of
  Algorithm~\ref{alg:recognition}. Then there are no $x,y\in L'$ in the
  same connected component of $H_1$ such that $xy\in E(\Ga)$ and
  $\lca_S(\sigma(x),\sigma(y))=u_S$.
\end{lemma}
\begin{proof}
  Assume for contradiction that, for some $L'$ and $u_S\in V^0(S)$
  appearing in the recursion, there is a connected component $C_i$ of $H_1$
  with vertices $x,y\in C_i$ and colors $X\coloneqq \sigma(x)$ and
  $Y\coloneqq\sigma(y)$ such that $xy\in E(\Ga)$ and $\lca_S(X,Y)=u_S$.  By
  assumption, $u_S$ is an interior vertex and thus $X\ne Y$.  Since the
  input $\Ga$ is a cograph, the induced subgraph $\Ga[L']$ and its
  complement, which by construction equals $H_1=\Gu[L'] \cup \Gg[L']$, are
  also cographs (cf.\ Prop.\ \ref{prop:cograph}).

  Consider a shortest path $P$ in $H_1$ connecting $x$ and $y$, which
  exists since $x,y\in C_i$. Since $\Ga[L']$ and $H_1$ are edge-disjoint
  and $xy\in E(\Ga[L'])$, $P$ contains at least $3$ vertices.  Since $H_1$
  is a cograph and thus does not contain induced $P_4$s, $P$ contains at
  most $3$ vertices. Hence, $P$ is of the form $x-z-y$ and we have $xy\in
  E(\Ga)$ and $xz, yz\notin E(\Ga)$.  Therefore, and since $\Gu$ and $\Gg$
  are properly colored, we have $Z\coloneqq\sigma(z)\notin \{X,Y\}$, and
  thus $X,Y,Z$ are pairwise distinct colors.  By Prop.\ \ref{prop:triples},
  $XY\vert Z\in \R_S(\graphs)$.  Taken together, the latter arguments and the
  construction of $S$ in Line~\ref{line:S} imply that $S$ displays the
  informative triple $XY\vert Z$.  Since $x,y,z\in L'$, we have $X,Y,Z\in
  L(S(u_s))$ by Obs.~\ref{obs:color-subset}. In particular, therefore,
  $Z\preceq_S u_S$. Thus $XY\vert Z$ implies that $\lca_S(X,Y) \prec_S u_S$; a
  contradiction.
\end{proof}

\begin{lemma}
  \label{lem:exit-not-Gu-edge}
  Let $\graphs = (\Gu,\Gg,\Ga,\sigma)$ with vertex set $L$ be a valid input
  for Algorithm~\ref{alg:recognition}.  Consider a recursion step on
  $L'\subseteq L$ and $u_S\in V^0(S)$ of
  Algorithm~\ref{alg:recognition}. Then, for all $x,y \in L'$ that are
  contained in the same connected component of $H_2$ but in distinct
  connected components of $H_3$, it holds $xy\in E(\Gu)$.
\end{lemma}
\begin{proof}
  Suppose that, for some $L'$ and $u_S\in V^0(S)$ appearing in the
  recursion, there is a connected component $C_j$ of $H_2$ with $x,y\in
  C_j$ such that $x$ and $y$ are in distinct connected components of $H_3$.
  In addition, suppose for contradiction that $xy\notin E(\Gu)$.  We may
  assume w.l.o.g.\ that $x$ and $y$ have minimal distance in $H_2$, i.e.,
  there are no vertices $x',y'\in C_j$ such that $x'$ and $y'$ are in
  distinct connected components of $H_3$, $x'y'\notin E(\Gu)$, and in
  addition the distance of $x'$ and $y'$ in $H_2$ is smaller than that of
  $x$ and $y$.  Set $X\coloneqq \sigma(x)$ and $Y\coloneqq \sigma(y)$ and
  let $C_x$ and $C_y$ be the connected components of $H_3$ that contain $x$
  and $y$, respectively. By Obs.~\ref{obs:color-subset}, we have
  $\sigma(L')\subseteq L(S(u_S))$. This and the fact that $x$ and $y$ are
  in distinct connected components of $H_3$ but in the same connected
  component $C_j$ of $H_2$ implies that $X\preceq_S v_X$ and $Y \preceq_S
  v_Y$ for two distinct children $v_X,v_Y \in \child_S(u_S)$.  In
  particular, we have $X\ne Y$ and $\lca_S(X,Y)=u_S$. Moreover, by
  construction, every connected component of $H_2$ is contained in a
  connected component of $H_1$ and thus, $x$ and $y$ are in the same
  connected component of $H_1$.  The latter two arguments together with
  Lemma~\ref{lem:exit-Ga-edge} imply $xy\notin E(\Ga)$.  In summary, we
  therefore have $xy\in E(\Gg)$.

  Consider a shortest path $P$ connecting $x$ and $y$ in $H_2$, which
  exists since $x,y\in C_j$.  By construction, $xy\in E(\Gg)$ and
  $\lca_S(X,Y)=u_S$ imply that $xy\notin E(H_2)$ and thus $P$ contains at
  least $3$ vertices.  Let $z\in C_j\setminus \{x,y\}$ be the neighbor of
  $x$ in $P$.  We consider the two possibilities (a) $xz\in E(\Gg)$ and (b)
  $xz\in E(\Gu)$.  Note that $X\ne \sigma(z) \eqqcolon Z$ holds in both
  cases since $\Gg$ and $\Gu$ are properly colored.

  In Case~(a), we must have $Z\preceq_S v_X$ since $xz$ is an edge in
  $H_2$.  This implies that $Z\ne Y$ (and thus $X,Y,Z$ are pairwise
  distinct) and $\lca_S(Y,Z)=u_S$.  Based on the latter arguments, $S$ must
  display the triple $XZ\vert Y$.  Suppose that $yz\notin E(\Gg)$. Together
  with $xy,xz \in E(\Gg)$, we have $XZ\vert Y\in \F_S$ and thus, by
  construction of $S$ in Line~\ref{line:S}, $S$ cannot display $XZ\vert Y$;
  a contradiction.  Hence, $yz\in E(\Gg)$ must hold.  Since $Z\preceq_S
  v_X$ and $Y\preceq_S v_Y$, we have $yz\notin E(H_3)$. Note that connected
  components in $H_3$ are complete graphs. Hence, $yz\notin E(H_3)$ implies
  that $y$ and $z$ are in distinct connected components of $H_3$.  However,
  the distance of $y$ and $z$ in $H_2$ is strictly smaller than that of $x$
  and $y$ (because $z$ is closer to $y$ than $x$ in the shortest path $P$);
  a contradiction to our choice of $x$ and $y$.  In summary, Case~(a)
  therefore cannot occur.

  In Case~(b) we have $xz\in E(\Gu)$.  If $yz\in E(\Gu)$, then $Y\ne Z$
  (because $\Gu$ is properly colored) and, by definition, $XY\vert Z\in
  \R_S$.  By construction in Line~\ref{line:S}, the species tree $S$
  displays $XY\vert Z$. Together with $X,Y,Z\in L(S(u_S))$ by
  Obs.~\ref{obs:color-subset}, this contradicts that $\lca_S(X,Y)=u_S$.
  Similarly, if $yz\in E(\Gg)$, then $S$ displays neither of the forbidden
  triples $XY\vert Z$ and $YZ\vert X$.  Hence, $S$ displays $XZ\vert Y$ or
  $S_{\vert XYZ}$ is the star tree on the three colors.  In both cases, we
  have $\lca_S(Y,Z)=\lca_S(X,Y)=u_S$.  In particular, therefore $y$ and $z$
  are in distinct connected components of $H_3$. As argued before, the
  distance of $y$ and $z$ is smaller than that of $x$ and $y$. Taken
  together the latter arguments again contradict our choice of $x$ and $y$,
  and thus $yz\in E(\Ga)$ is left as the only remaining choice.
  
  In summary, only case (b) $xz\in E(\Gu)$ is possible, which in particular
  implies $yz\in E(\Ga)$. Therefore, we have $yz\notin E(H_2)$ and thus the
  path $P$ contains at least $4$ vertices. Thus, consider the neighbor
  $w\in C_j\setminus \{x,y,z\}$ of $y$ in $P$ and set $W\coloneqq
  \sigma(w)$.  We can apply analogous arguments for $x,y,w$ as we have used
  for $x,y,z$ to exclude the case (a') $yw\in E(\Gg)$ and, in case (b')
  $yw\in E(\Gu)$, we obtain $xw\in E(\Ga)$ as the only possibility.
  
  Taking the latter arguments together, it remains to consider the case $xy
  \in E(\Gg)$, $xz, yw\in E(\Gu)$, and $xw, yz\in E(\Ga)$.  Since $\Gu$ and
  $\Ga$ are cographs, we have $zw\in E(\Gg)$ because otherwise $x-z-w-y$ or
  $x-w-z-y$ would be an induced $P_4$ in $\Gu$ and $\Ga$, respectively.

  Now, $x$ and $w$ must be in the same connected component of $H_3$, as
  otherwise $xw \notin E(G_<)$ and the fact that $x$ and $w$ are at a
  shorter distance than $x$ and $y$ in $H_2$ would contradict our choice of
  $x$ and $y$.  Likewise, $y$ and $z$ are in the same connected component
  of $H_3$ since $yz \notin E(G_<)$ and they are closer than $x$ and $y$ in
  $H_2$.  It follows that $w$ and $z$ are in distinct connected components
  of $H_3$, again yielding a contradiction since they are closer than $x$
  and $y$ in $H_2$ and $wz \notin E(G_<)$.  In summary, therefore, we have
  $xy\in E(\Gu)$.
\end{proof}

The following result is a consequence of Lemma~\ref{lem:exit-not-Gu-edge}
and will be helpful later on.
\begin{corollary}
  \label{cor:H2-cc}
  Let $\graphs = (\Gu, \Gg, \Ga, \sigma)$ with vertex set $L$ be a valid
  input for Algorithm~\ref{alg:recognition}.  Consider a recursion step on
  $L'\subseteq L$ and $u_S\in V^0(S)$ of
  Algorithm~\ref{alg:recognition}. If $xy \in E(H_1)\setminus E(H_2)$, then
  $x$ and $y$ are in distinct connected components of $H_2$.
\end{corollary}
\begin{proof}
  Suppose $xy \in E(H_1)\setminus E(H_2)$. By construction of the auxiliary
  graphs, this implies that $xy\in E(\Gg)$ and there is no $v\in
  \child_S(u_S)$ such that $\sigma(x),\sigma(y) \prec_{S} v$.  The latter
  in particular yields that $xy \notin E(H_3)$.  This, together with the
  fact that $H_3$ is a graph whose connected components are complete
  graphs, implies that $x$ and $y$ are in distinct connected components of
  $H_3$.  We can now use Lemma~\ref{lem:exit-not-Gu-edge} to conclude that
  $x$ and $y$ must also be in distinct connected components of $H_2$ as
  otherwise we would obtain $xy\in E(\Gu)$; a contradiction.
\end{proof}

We are now in the position to demonstrate that Algorithm~\ref{alg:recognition}
is correct.

\begin{lemma}
  \label{lem:algo-explains}
  Let $\graphs$ be a valid input for
  Algorithm~\ref{alg:recognition}.  Then, Algorithm~\ref{alg:recognition}
  returns a relaxed scenario $\scen=(T,S,\sigma,\mu,\tT,\tS)$ that explains
  $\graphs$.
\end{lemma}
\begin{proof}
  Let $\graphs = (\Gu, \Gg, \Ga, \sigma)$ be a valid input with vertex set $L$ 
  for  Algorithm~\ref{alg:recognition}. By Lemma~\ref{lem:algo-rs}, 
  Algorithm~\ref{alg:recognition} returns a relaxed scenario
  $\scen=(T,S,\sigma,\mu,\tT,\tS)$ such that $L(T)=L$.
  We continue with showing that $ \scen$ explains $\graphs$.
  
  Consider two distinct vertices $x,y\in L=L(T)$ and their last common
  ancestor $\lca_T(x,y)$.  Let $L'\subseteq L$ and $u_S\in V(S)$ be the
  input of the recursive call of \texttt{BuildGeneTree} in which
  $\lca_T(x,y)$ was created.  By Obs.~\ref{obs:color-subset}
  and~\ref{obs:leaf-sets}, we have $\sigma(L')\subseteq L(S(u_s))$ and
  $x,y\in L'$, respectively, and therefore $\lca_S(\sigma(x), \sigma(y))
  \preceq_S u_S$.  Moreover, time consistency yields $\tS(\lca_S(\sigma(x),
  \sigma(y))) \le \tS(u_S)$.  The vertex $\lca_T(x,y)$ has been created in
  exactly one of the following three locations in the algorithm: (a) in
  Line~\ref{line:create-rho}, (b)~in Line~\ref{line:create-u_i}, and (c)~in
  Line~\ref{line:create-v_j}.

  In Case~(a), $\lca_T(x,y)$ equals $\rho'$ in the recursion step of
  interest.  Suppose first that $u_S$ is a leaf of $S$ and thus
  $\sigma(x)=\sigma(y)=u_S$.  Hence, we have $xy\in E(\Ga(\scen))$ by
  Cor.~\ref{cor:equal-colors} and $xy\in E(\Ga)$, since $\Gu$ and $\Gg$ are
  properly colored.  Now suppose that $u_S$ is not a leaf.  Then
  $\lca_T(x,y)=\rho'$ implies that $x$ and $y$ lie in distinct connected
  components of the auxiliary graph $H_1$ and thus $xy\notin E(H_1)$.  By
  construction of this graph, the latter yields $xy\in E(\Ga)$.  Moreover,
  we have set $\tT(\rho')=\tS(u_S)+\epsilon > \tS(u_S)$.  Together with
  $\lca_T(x,y)=\rho'$, this implies $\tS(\lca_S(\sigma(x), \sigma(y))) \le
  \tS(u_S) < \tT(\lca_T(x,y))$ and thus $xy\in E(\Ga(\scen))$.

  In Case~(b), $u_S$ is an inner vertex of $S$ and $\lca_T(x,y)$ equals
  $u_i$.  We have set $\tT(\lca_T(x,y)) = \tT(u_i) = \tS(u_S)$.  By
  construction, moreover, $x$ and $y$ must be in the same connected
  component $C_i$ of $H_1$ but in distinct connected components of $H_2$.
  Hence, we have $xy\notin E(H_2)$ which implies $xy\notin E(\Gu)$ by the
  construction of $H_2$.  Suppose first $\lca_S(\sigma(x), \sigma(y)) =
  u_S$.  Then $xy\in E(\Gg)$ as otherwise it would hold $xy\in E(\Ga)$; a
  contradiction to Lemma~\ref{lem:exit-Ga-edge}.  Moreover, we have
  $\tT(\lca_T(x,y)) = \tS(u_S)=\tS(\lca_S(\sigma(x), \sigma(y)))$ and thus
  $xy\in E(\Gg(\scen))$.  Now suppose $\lca_S(\sigma(x), \sigma(y)) \prec_S
  u_S$ and thus, by time consistency, $\tS(\lca_S(\sigma(x), \sigma(y))) <
  \tS(u_S) = \tT(\lca_T(x,y))$. This yields $xy\in E(\Ga(\scen))$.
  Moreover, from $\lca_S(\sigma(x), \sigma(y)) \prec_S u_S$, we conclude
  that $\sigma(x),\sigma(y)\preceq_S w$ for some child $w\in
  \child_S(u_S)$.  Therefore, we must have $xy\in E(\Ga)$ since otherwise
  $xy\in E(\Gg)$ would imply that $xy\in E(H_2)$.

  In Case~(c), $u_S$ is an inner vertex of $S$ and $\lca_T(x,y)$ equals
  $v_j$.  We have set $\tT(\lca_T(x,y)) = \tT(v_j) = \tS(u_S) - \epsilon <
  \tS(u_S)$.  By construction, moreover, $x$ and $y$ must be in the same
  connected component $C_j$ of $H_2$ (and thus also in the same connected
  component $C_i$ of $H_1$) but in distinct connected components of $H_3$.
  This immediately implies (i) that $xy \in E(\Gu)$ by
  Lemma~\ref{lem:exit-not-Gu-edge} and (ii), by construction of $H_3$, that
  $\sigma(x)$ and $\sigma(y)$ lie below distinct children of $u_S$. In
  particular, therefore, we have $\lca_S(\sigma(x), \sigma(y)) = u_S$ and
  thus $\tS(\lca_S(\sigma(x), \sigma(y))) = \tS(u_S) > \tT(\lca_T(x,y))$.
  This implies $xy\in E(\Gu(\scen))$.

  In summary, we have shown that $xy\in E(\Gu)$ iff $xy\in E(\Gu(\scen))$,
  $xy\in E(\Gg)$ iff $xy\in E(\Gg(\scen))$, and $xy\in E(\Ga)$ iff $xy\in
  E(\Ga(\scen))$.  Since $x,y\in L$ where chosen arbitrarily and $L=L(T)$,
  this proves that the relaxed scenario $\scen$ returned by the algorithm
  indeed explains the input $\graphs$.
\end{proof}

As outlined in the proof of Lemma~\ref{lem:algo-explains}, edges
  $xy\in E(\Gg)$ are considered only in Case (b) and we have $\lca_T(x,y) =
  u_i$ and $\lca_S(\sigma(x), \sigma(y)) = u_S$.  In this case, we put
  $\mu(u_i) = u_S$ in Line~\ref{line:mu-tT-inner2} of
  Algorithm~\ref{alg:recognition}. The reconciliation map $\mu$ therefore
  has the following property:
\begin{fact}
  \label{obs:generic}
  Let $\scen$ be a scenario produced by Algorithm~\ref{alg:recognition}
  for a valid input $\graphs=(\Gu, \Gg, \Ga, \sigma)$. Then $xy\in
  E(\Gg)$ implies $\mu(\lca_T(x,y)) = \lca_S(\sigma(x),\sigma(y))$.
\end{fact}

  A main result of this section is the following characterization of
  graph $3$-partitions that derive from relaxed scenarios:
\begin{theorem}
  \label{thm:explainable-charac}
  A graph $3$-partition $\graphs=(\Gu, \Gg, \Ga, \sigma)$ can be explained
  by a relaxed scenario if and only if $\Gu$ and $\Gg$ are properly
  colored, $\Gu$ and $\Ga$ are cographs, and $(\R_S(\graphs),
  \F_S(\graphs))$ is consistent.
\end{theorem}
\begin{proof}
  Suppose first that $\graphs$ can be explained by a relaxed scenario.
  Then $\Gu$ and $\Gg$ are properly colored by Cor.~\ref{cor:equal-colors},
  $\Gu$ and $\Ga$ are cographs by Lemmas~\ref{lem:Gu-cograph}
  and~\ref{lem:Ga-cograph}, respectively, and $(\R_S(\graphs),
  \F_S(\graphs))$ is consistent by Prop.~\ref{prop:triples}.  Conversely,
  suppose $\Gu$ and $\Gg$ are properly colored, $\Gu$ and $\Ga$ are
  cographs, and $(\R_S(\graphs), \F_S(\graphs))$ is consistent.  In this
  case, $\graphs = (\Gu, \Gg, \Ga, \sigma)$ is a valid input for
  Algorithm~\ref{alg:recognition} and Lemma~\ref{lem:algo-explains} implies
  that Algorithm~\ref{alg:recognition} returns a relaxed scenario that
  explains $\graphs$.
\end{proof}

This result implies almost immediately that the property of being
explainable by a relaxed scenario is hereditary:
\begin{corollary}
  \label{cor:hereditary}
  A graph $3$-partition $\graphs=(\Gu, \Gg, \Ga, \sigma)$ with vertex set
  $L$ can be explained by a relaxed scenario if and only if $\graphs_{\vert
    L'}$ can be explained by a relaxed scenario for all subsets
  $L'\subseteq L$.
\end{corollary}
\begin{proof}
  The \textit{if}-part is clear as $\graphs=\graphs_{\vert L}$. Conversely,
  suppose that $\graphs=(\Gu, \Gg, \Ga, \sigma)$ is explained by a relaxed
  scenario $\scen=(T,S,\sigma,\mu,\tT,\tS)$ and let $L'\subseteq L$.  By
  Prop.~\ref{prop:triples}, therefore, $S$ agrees with $(\R_S(\graphs),
  \F_S(\graphs))$.  By Thm.~\ref{thm:explainable-charac}, $\Gu$ and $\Gg$
  are properly colored and $\Gu$ and $\Ga$ are cographs.  Now consider
  $\graphs_{\vert L'} = (\Gu[L'], \Gg[L'], \Ga[L'], \sigma_{\vert L'})$.
  Clearly, the induced subgraphs $\Gu[L']$ and $\Gg[L']$ are also properly
  colored.  By Prop.~\ref{prop:cograph}, $\Gu[L']$ and $\Ga[L']$ are also
  cographs.  By definition of the informative and forbidden triples in
  Def.~\ref{def:inf-forb-triples} and the induced subgraph relationships,
  we observe furthermore that $\R_S(\graphs_{\vert L'}) \subseteq
  \R_S(\graphs)$ and $\F_S(\graphs_{\vert L'}) \subseteq \F_S(\graphs)$.
  Hence, $S$ displays all triples in $\R_S(\graphs_{\vert L'})$ and none of
  the triples in $\F_S(\graphs_{\vert L'})$, which yields that
  $(\R_S(\graphs_{\vert L'}), \F_S(\graphs_{\vert L'}))$ is consistent.  We
  can now again apply Thm.~\ref{thm:explainable-charac} to conclude that
  $\graphs_{\vert L'}$ is explainable.
\end{proof}

Using the characterization in Thm.~\ref{thm:explainable-charac}, we can
decide in polynomial time whether a graph $3$-partition is explainable by a
relaxed scenario:
\begin{corollary}
  \label{cor:perf}
  It can be decided in $O(\vert L\vert^4 \log\vert L\vert)$ time whether a 
  graph $3$-partition $\graphs=(\Gu, \Gg, \Ga, \sigma)$ can be explained by a 
  relaxed scenario.
\end{corollary}
\begin{proof}
  It can be checked in $O(\vert L\vert^2)$ time whether $\Gu$ and $\Gg$ are
  properly colored. It can be decided in in $O(\vert L\vert + \vert E\vert
  )$ time whether a graph $G=(L,E)$ is a cograph \cite{Corneil:85}. In
  particular, it can also be verified in $O(\vert L\vert^2)$ time that
  $\Gu$ and $\Ga$ are cographs.  Extraction of $\R\coloneqq\R_S(\graphs)$
  and $\F\coloneqq\F_S(\graphs)$ according to
  Def.~\ref{def:inf-forb-triples} requires $O(\vert L\vert^3)$.  Let
  $M'\subseteq M$ be the subset of colors that appear on the leaves of the
  triples in $\R\cup \F$.  By construction, we have $\vert M'\vert \in
  O(\vert L\vert )$.  The algorithm \texttt{MTT}, which stands for
  \emph{mixed triplets problem restricted to trees} and was described in
  \cite{He:06}, constructs a tree on $M'$ that agrees with $(\R,\F)$, if
  one exists, in $O(\vert \R\vert \cdot \vert M'\vert + \vert \F\vert \cdot
  \vert M'\vert \log \vert M'\vert + \vert M'\vert^2 \log \vert M'\vert)$
  time.  This, together with $\vert \R\vert, \vert \F\vert \in O(\vert
  L\vert^3)$ and $\vert M'\vert \in O(\vert L\vert )$ implies that it can
  be decided in $O(\vert L\vert^4\log\vert L\vert)$ whether $(\R,\F)$
  is consistent.
\end{proof}

In particular, it can be decided in $O(\vert L\vert^4 \log\vert L\vert)$ 
whether $\graphs=(\Gu, \Gg, \Ga, \sigma)$ can be explained by a
relaxed scenario without explicit construction of such a scenario.  We will
show in the following that the construction of relaxed scenarios is bounded
by the same complexity. For simplicity, we will explicitly require that
$\sigma\colon L\to M$ is surjective, i.e., that $\sigma(L)=M$ holds. One
easily verifies, however, that the existence of ``unused colors'' in $M$
only increases the size of the species tree $S$ (in particular, the number
of leaves in $S$ that are attached to $\rho_S$) but does not affect the
existence of a relaxed scenario that explains $\graphs$.

\begin{lemma}
  \label{lem:algo-complexity}
  Algorithm~\ref{alg:recognition} can be implemented to run in
  $O(\vert L\vert^4\log\vert L\vert)$ time (for valid inputs
  $\graphs=(\Gu, \Gg, \Ga, \sigma)$ such that $\sigma$ is surjective).
\end{lemma}
\begin{proof}
  Let $\graphs=(\Gu, \Gg, \Ga, \sigma)$ with vertex set $L$ be a valid
  input and surjective coloring $\sigma\colon L\to M$ that is given as
  input for Algorithm~\ref{alg:recognition}.  By assumption, $\Gu$ and
  $\Gg$ are properly colored, $\Gu$ and $\Ga$ are cographs, and
  $(\R_S(\graphs), \F_S(\graphs))$ is consistent.  Extraction of
  $\R\coloneqq\R_S(\graphs)$ and $\F\coloneqq\F_S(\graphs)$ according to
  Def.~\ref{def:inf-forb-triples} requires $O(\vert L\vert^3)$ operations.
  As argued in the proof of Corollary~\ref{cor:perf}, a tree $S$ on $M$ that 
  agrees with $(\R,\F)$ can be constructed in $O(\vert L\vert^4\log\vert 
  L\vert)$ time using algorithm \texttt{MTT} \cite{He:06}.

  A suitable time map $\tS$ can be constructed in $O(\vert M\vert)=O(\vert
  L\vert )$ time by Lemma~\ref{lem:arbitrary-tT}.

  We can employ the LCA data structure described by 
  Bender et al.\ \cite{Bender:05},
  which pre-processes $S$ in $O(\vert M\vert )=O(\vert L\vert )$ time to
  allow $O(1)$-query of the last common ancestor of pairs of vertices in
  $S$ afterwards.  In addition, we want to access the vertex
  $w\in\child_{S}(u)$ satisfying $v\preceq_{S} w$ for two given vertices
  $u,v\in V(T)$ with $v\prec_{S} u$. To achieve this, we pre-process $S$ as
  follows: We first compute $\depth(v)$ for each $v\in V(T)$, i.e., the
  number of edges on the path from the root to $v$ in a top-down traversal
  of $S$ in $O(\vert L\vert )$ time.  The \emph{Level Ancestor (LA)
  Problem} asks for the ancestor $\LA(v,d)$ of a given vertex $v$ that has
  depth $d$, and has solutions with $O(\vert L\vert )$ pre-processing and
  $O(1)$ query time \cite{Berkman:94,Bender:04}.  Hence, we can obtain the
  desired vertex $w$ as $\LA(v,\depth(u)+1)$ in constant time.

  Since $\sigma(L')\subseteq L(S(u_S))$ always holds by
  Obs.~\ref{obs:color-subset}, every $x\in L$ appears at most once in a
  loop corresponding to Line~\ref{line:loop-u_S-leaf}. Hence, the total
  effort of handling the cases where $u_S$ is a leaf is bounded by $O(\vert
  L\vert )$.  Consider now one execution of \texttt{BuildGeneTree} (without
  the recursive calls) in which $u_S$ is not a leaf.  Construction of the
  auxiliary graphs $H_1$ and $H_2$ is done in $O(\vert L'\vert^2)$, where
  the condition $\sigma(x),\sigma(y) \prec_{S} v$ for some $v\in
  \child_S(u_S)$ in the construction of $H_2$ is equivalent to querying the
  LCA data structure in $O(1)$ time whether $\lca_S(\sigma(x),\sigma(y))\ne
  u_S$.  The connected components of $H_1$ can be obtained in $O(\vert
  L'\vert + \vert E(H_1)\vert )=O(\vert L'\vert^2)$ time using
  breadth-first search.  Since $H_2$ is a subgraph of $H_1$, we can, for
  each connected component $C_i$ of $H_1$, determine the connected
  components $C_j$ of $H_2$ with $C_j\subseteq C_i$ again using
  breadth-first search and only the vertices in $C_i$ as start vertices.
  The overall effort for this is again bounded by $O(\vert L'\vert + \vert
  E(H_1)\vert )=O(\vert L'\vert^2)$.  We can now, for each connected
  component $C_j$ of $H_2$, construct the connected components $C_k$ of
  $H_3$ with $C_k\subseteq C_j$ by (i) adding the edge $xy$ to $H_3$ if
  $\lca_S(\sigma(x),\sigma(y)) \ne u_S$ for all $x,y\in C_j$ and (ii)
  performing breadth-first search on $H_3$ using only the vertices in $C_j$
  as start vertices.  Again, the overall effort for these breadth-first
  searches is bounded by $O(\vert L'\vert^2)$.  The number of connected
  component of the three graph $H_1$, $H_2$, and $H_3$ is bounded by
  $O(\vert L'\vert)$.  For each connected component $C_j$ of $H_2$, we
  have to choose $v^*_S\in \child_S(u_{S})$ such that $\sigma(C_j)\cap
  L(S(v^*_S))\ne\emptyset$ in Line~\ref{line:choose-v-S}.  To this end, we
  pick $x\in C_j$ arbitrarily and query $v^*_S=\LA(\sigma(x),
  \depth(u_S)+1)$.  For each connected component $C_k$ of $H_3$, we can
  find $v_S\in \child_S(u_{S})$ such that $\sigma(C_k)\subseteq L(S(v_S))$
  in Line~\ref{line:choose-v-S-for-class} in the same way.  In summary, for
  each connected component of each graph, the effort of creating a new
  vertex (in case of $H_1$ and $H_2$), attaching the vertex to the tree
  ($H_1$, $H_2$, and $H_3$), choosing $v^*_S$ in Line~\ref{line:choose-v-S}
  ($H_2$), choosing $v_S$ in Line~\ref{line:choose-v-S-for-class} ($H_3$),
  and assigning the values for $\tT$ and $\mu$ for the newly created
  vertices are all constant-time operations.  The overall effort for one
  recursion step (excluding the recursive calls) is therefore bounded by
  $O(\vert L' \vert^2)$.

  To bound the total effort of \texttt{BuildGeneTree}, consider the
  recursion tree $R$ of the algorithm and let $d$ be its maximum depth
  (i.e.\ the maximum distance from $\rho_R$ to a leaf).  Notice that when a
  recursion receives $u_S \in V(S)$ as input, it passes a child of $u_S$ to
  any recursive call that it makes.  Since terminal calls occur on leaves
  of $S$, it follows that $d$ is at most the height of $S$, which is
  $O(\vert V(S)\vert) = O(\vert L\vert)$ under the assumption that $\sigma$
  is surjective.  For $r \in V(R)$, denote by $L'_r$ the set $L'$ received
  as input on the recursive call corresponding to $r$.  If $r$ is not a
  leaf of $R$, then notice that $\{L'_q : q \in \child_R(r)\}$ is a
  partition of $L'_r$ (without repeated subsets), since a recursive call is
  made precisely for each connected component of $H_3$.

  Let $\ell \in \{0,1,\ldots,d\}$.  We claim that for any two vertices $r,
  q \in V(R)$ at distance $\ell$ from $\rho_R$, $L'_r \cap L'_q =
  \emptyset$.  This can be seen by induction, with $\ell = 0$ as the
  trivial base case.  Consider $\ell > 0$.  If $r$ and $q$ have the same
  parent, then $L'_r \cap L'_q = \emptyset$ follows from the observation
  that recursions partition their input $L'$ to their child calls.  If $r$
  and $q$ have distinct parents in $R$, we know by induction that
  $L'_{par_R(r)} \cap L'_{par_R(q)} = \emptyset$.  Since recursions pass a
  subset of their input $L'$, $L'_r \cap L'_q = \emptyset$ holds as well.
  Thus our claim is true.  Now, for a given depth $\ell \in
  \{0,1,\ldots,d\}$, denote by $r_1, \ldots, r_k$ the set of vertices of
  $R$ at distance $\ell$ from $\rho_R$.  The total effort of these vertices
  is $O(\vert L'_{r_1}\vert^2 + \ldots + \vert L'_{r_k}\vert^2)$ and, since
  $\vert L'_{r_1}\vert + \ldots + \vert L'_{r_k}\vert \leq \vert L\vert$ by
  our claim, the total time spent at depth $\ell$ is $O(\vert L\vert^2)$.
  Because this holds for every depth from $0$ to $d \in O(\vert L\vert)$,
  the total time spent in \texttt{BuildGeneTree} is $O(\vert L\vert^3)$.

  It only remains to argue on the time spent constructing the final output
  tree $T$.  Note that in each recursion with corresponding vertex $r \in
  V(R)$, \texttt{BuildGeneTree} adds at most $2\vert L'_r\vert+1$ nodes to
  the constructed tree $T'$ (we always add $\rho'$ and, additionally, in
  non-terminal calls, we add one $u_i$ and one $v_j$ vertex for each of the
  $O(\vert L'_r\vert)$ connected components of $H_1$ and $H_2$,
  respectively, and in terminal calls we add $\vert L'_r\vert$ leaves).
  Since the vertices of $R$ at the same depth $\ell$ receive pairwise
  disjoint $L'_r$ sets, it follows that a total of at most $O(\vert
  L\vert)$ vertices are added to $T$ by the recursive calls at the same
  depth $\ell$.  Since $d \in O(\vert L\vert)$, the resulting tree $T'$ has
  at most $O(\vert L\vert^2)$ vertices.  To obtain the final gene tree $T$,
  we can traverse $T'$ and suppress all vertices with a single child by
  removing the vertex and reconnecting its child to its parent in $(O(\vert
  V(T')\vert)=O(\vert L\vert^2)$ total time.
  
  Hence, the overall time complexity of Algorithm~\ref{alg:recognition} is
  $O(\vert L\vert^4 \log\vert L\vert)$.
\end{proof}

\section{Explanation of $\boldsymbol{\graphs}$ by Restricted Scenarios}
\label{sect:restricted} 

Relaxed scenarios may contain combinations of HGT and deletion events that
render the HGT event ``unobservable'' from extant data, because the gene
family died out in the lineage from which that HGT originated.  It is
therefore of interest to consider more restrictive classes of scenarios
that exclude such ``unobservable'' events.  In this section, we show
  that if a relaxed scenario explains $\graphs$, then there is always some
  scenario without these ``unobservable'' events that also explains
  $\graphs$.  To this end, we introduce the notion of a ``witness'':
\begin{definition}
  Let $\scen=(T,S,\sigma,\mu,\tT,\tS)$ be a relaxed scenario. We say that
  $x\in L(T)$ is a \emph{witness} for $v\in V(T)$ if $x\preceq_T v$ and the
  path from $v$ to $x$ in $T$ does not contain an HGT-edge. The scenario
  $\scen$ is \emph{fully witnessed} if every $v\in V(T)$ has a witness.
\end{definition}
It is not difficult to verify that, in order for a relaxed scenario
$\scen=(T,S,\sigma,\mu,\tT,\tS)$ to be fully witnessed, it is necessary and
sufficient that every vertex $v\in V^0(T)$ has a child $w$ such that
$\mu(w)\preceq_{S}\mu(v)$. In essence, this matches condition (2b) assumed
in the work of Tofigh et al.\ \cite{Tofigh:11} and is also a direct
consequence of condition (O2) in \cite{Hellmuth:17,Nojgaard:18a}.

A vertex $x\in V(T)$ with $\mu(x)\in V(S)$ describes an evolutionary
event that coincides with a \emph{speciation}. This suggests to require
additional constraints on $\mu$ that exclude scenarios that do not have a
simple biological interpretation. In particular, it seems natural to
prevent HGT-edges from emanating from such a vertex. This amounts to the
assumption that speciations and HGT events are not allowed to be lumped
into the same event (cf.\ \cite{Hellmuth:17}). Another interesting 
constraint on a speciation $u$ is to require that they are witnessed by a pair of descendants $x$ and
$y$ in two of the lineages that are separated by the speciation, i.e.,
such that $u=\lca_T(x,y)$ and
$\mu(\lca_T(x,y))=\lca_S(\sigma(x),\sigma(y))$. This condition is
reminiscent, but weaker, than the Last Common Ancestor reconciliation
\cite{Guigo:96,Page:97}.

\begin{definition}
  A relaxed scenario $\scen=(T,S,\sigma,\mu,\tT,\tS)$ is a \emph{restricted
    scenario} if it satisfies the following three constraints: 
  \begin{itemize}
    \item[(S4)] $\scen$ is fully witnessed.
    \item[(S5)] If $\mu(u)\in V^0(S)$, then $\mu(v)\prec_{S} \mu(u)$ holds 
			    for all $v\in\child_T(u)$.
    \item[(S6)] If $\mu(u)\in V^0(S)$, then there exist at least two leaves
      $x,y \in L(T)$ such that $\lca_T(x,y)=u$, both $x$ and $y$ are
      witnesses for $u$, and $\mu(u) = \lca_S(\sigma(x),\sigma(y))$.
  \end{itemize}
\end{definition}
It is worth noting that conditions (S4), (S5), and (S6) are not
necessarily satisfied by the most commonly studied classes of
evolutionary scenarios. For example, the DTL scenarios considered in
\cite{Nojgaard:18a} do not need to satisfy (S5) if $S$ or $T$ is
non-binary. In the remainder of this section, we show that -- curiously enough 
-- any data $\graphs=(\Gu, \Gg, \Ga, \sigma)$ that can be explained by a
relaxed scenario can also be explained by a restricted scenario. We start by
showing that Algorithm~\ref{alg:recognition} already enforces some
additional constraints.

\begin{lemma}
  \label{lem:algo-fully-witnessed}
  Given a valid input $\graphs = (\Gu, \Gg, \Ga, \sigma)$, the scenario
  $\scen=(T,S,\sigma,\mu,\tT,\tS)$ returned by
  Algorithm~\ref{alg:recognition} satisfies (S4), i.e., it is fully
  witnessed.
\end{lemma}
\begin{proof}
  Consider the intermediate tree $T'$ constructed in
  Algorithm~\ref{alg:recognition} which is not necessarily phylogenetic.
  By a slight abuse of notation, we will simply write $\mu$ and $\tT$ also
  for restrictions to subsets of $V(T)$.  We start with showing that each
  inner vertex $u\in V^0(T')$ has a child $v\in V(T')$ such that
  $\mu(v)\preceq_{S} \mu(u)$ and, thus, that $uv$ is not an HGT edge.  Let 
  $L'\subseteq L$ and $u_S\in V(S)$ be the input of the recursive call of 
  \texttt{BuildGeneTree} in which $u\in V^0(T')$ was created in one of 
  Lines~\ref{line:create-rho}, \ref{line:create-u_i}, or~\ref{line:create-v_j}.

  Suppose first $u=\rho'$ was created in Line~\ref{line:create-rho} and
  thus $\mu(u)=\parent_S(u_S) u_S$.  If $u_S$ is a leaf, then we attached
  all of the elements $x\in L'$ as children of $u$ and set
  $\mu(x)=\sigma(x)$. Since $\sigma(L')\subseteq L(S(u_S))=\{u_S\}$ holds
  by Obs.~\ref{obs:color-subset}, we have $\mu(x)=\sigma(x)=u_S$. Therefore,
  and since $L'$ is non-empty, $u$ has a child $v$ such that $\mu(v)= u_S
  \preceq_{S} \parent_S(u_S) u_S = \mu(u)$.  If $u_S$ is not a leaf, then
  we have attached at least one vertex $u_i$ corresponding to a connected
  component $C_i$ of $H_1$ as a child of $u$ in the same recursion step. In
  particular, we have set $\mu(u_i)=u_S$ in Line~\ref{line:mu-tT-inner2},
  and thus, $\mu(u_i)= u_S \preceq_{S} \parent_S(u_S) u_S = \mu(u)$.

  Suppose $u=u_i$ was created in Line~\ref{line:create-u_i} and thus
  $\mu(u)=u_S$. In particular, $u=u_i$ corresponds to some connected component
  $C_i$ of $H_1$. Since $H_2\subseteq H_1$ there is at least one connected
  component $C_j$ of $H_2$ such that $C_j\subseteq C_i$ and thus we have
  attached at least one vertex $v_j$ as created in Line~\ref{line:create-v_j}
  as a child of $u$ and set $\mu(v_j)= u_S v^*_S$ for some $v^*_S\in
  \child_{S}(u_S)$. Hence, we have $\mu(v_j)= u_S v^*_S \preceq_{S} u_S =
  \mu(u)$.

  Suppose, finally, that $u=v_j$ was created in Line~\ref{line:create-v_j}.
  Hence, $v_j$ corresponds to some connected component $C_j$ of $H_2$ and
  we have set $\mu(v_j)= u_S v^*_S$ for some $v^*_S\in \child_{S}(u_S)$
  such that $\sigma(C_j) \cap L(S(v^*_S)) \ne \emptyset$.  The latter
  implies that there is $x\in C_j$ such that $\sigma(x)\in L(S(v^*_S))$.
  By construction of the auxiliary graphs, there is a connected component
  $C_k$ such that $x\in C_k$ and $C_k\subseteq C_j$. Moreover, we have
  chosen $v_S\in \child_{S}(u_S)$ in Line~\ref{line:choose-v-S-for-class}
  such that $\sigma(C_k)\subseteq L(S(v_S))$. This together with
  $\sigma(x)\in L(S(v^*_S))$ and $\sigma(x)\in \sigma(C_k)$ implies that
  $v^*_S=v_S$.  In particular, we have attached the vertex $\rho'$ as a
  child to $u=v_j$ that was created in Line~\ref{line:create-rho} of the
  the recursion step $\texttt{BuildGeneTree}(C_k, v^*_S)$ and that
  satisfies $\mu(\rho')= \parent_S(v^*_S) v^*_S = u_S v^*_S$.  Hence, we
  have $\mu(\rho')= u_S v^*_S \preceq_{S} u_S v^*_S = \mu(u)$.

  In summary, each inner vertex $u\in V^0(T')$ has a child $v\in V(T')$
  such that $\mu(v)\preceq_{S} \mu(u)$. Therefore and since $T'$ is finite,
  we can find a descendant leaf $x\in L(T')$ for each $u\in V^0(T')$ that
  can be reached from $u$ by non-HGT-edges.

  Now consider a vertex $v\in V^0(T) \setminus \{0_T\} \subseteq V^0(T')$.
  By the arguments above, we find a path $P'= (v\eqqcolon v'_1 - v'_2 -
  \dots - v'_{k'}\coloneqq x)$ in $T'$ from $v$ to some of its descendant
  leaves $x\in L(T')=L(T)$ that does not contain any HGT-edge, i.e., it
  holds $\mu(v'_{i+1}) \preceq \mu(v'_i)$ for all $1\le i < k'$.  Therefore
  and since $T$ is obtained from $T'$ by adding $0_T$ and suppression of
  all vertices with a single child, we have $x\prec_{T} v$ and, moreover,
  the path $P= (v\eqqcolon v_1 - v_2 - \dots - v_k\coloneqq x)$ connecting
  $v$ and $x$ in $T$ contains only vertices that are also contained in $P'$
  in the same order.  We therefore conclude that $\mu(v_{i+1}) \preceq
  \mu(v_i)$ holds for all $1\le i < k$, i.e., $P$ does not contain any
  HGT-edge.  Hence, there is a witness for each vertex $v\in V^0(T)
  \setminus \{0_T\}$ By definition, each leaf $x\in L(T)$ is a witness of
  itself.  Finally, consider $0_T$ (and its unique child $\rho_T$). By
  construction, it holds $\mu(0_T)=0_S$.  Therefore and since every element
  $z\in V(S)\cup E(S)$ satisfies $z\preceq_{S} 0_T$, we have that
  $\mu(\rho_T) \preceq_{S} \mu(0_T)$, and thus $0_T \rho_T$ is not an
  HGT-edge.  Hence, every witness of $\rho_T$ is also a witness of $0_T$,
  which concludes the proof.
\end{proof}

\begin{lemma}
  \label{lem:algo-S5}
  Given a valid input $\graphs = (\Gu, \Gg, \Ga, \sigma)$, the scenario
  $\scen=(T,S,\sigma,\mu,\tT,\tS)$ returned by Algorithm~\ref{alg:recognition}
  satisfies (S5), i.e., $\mu(u)\in V^0(S)$ implies that $\mu(v)\prec_{S} 
  \mu(u)$ for all $v\in\child_T(u)$.
\end{lemma}
\begin{proof}
  Suppose that $\mu(u) \in V^0(S) = V(S) \setminus (L(S) \cup \{0_S\})$ and
  let $v\in \child_T(u)$ be an arbitrary child of $u$.  Inspection of
  Algorithm~\ref{alg:recognition} shows that $u$ must have been created in
  Line~\ref{line:create-u_i} in some recursion step on $L'\subseteq L$ and
  $u_S\in V^0(S)$ and thus $\mu(u)=u_S$.  Consider the intermediate tree
  $T'$ constructed in the algorithm from which $T$ is obtained by adding
  the planted root $0_T$ and suppression of all inner vertices with a
  single child.  In particular, the path connecting $u$ and $v$ in $T'$
  passes through some child $v'$ of $u$ in $T'$ (where $v=v'$ is possible).
  By construction, we have set $\mu(v')=u_S v^*_S$ for some $v^*_S\in
  \child_{S} (u_S)$ in Line~\ref{line:mu-tT-inner3}.  Re-using the
  arguments in the proof of Lemma~\ref{lem:algo-fully-witnessed}, we find a
  path $P= (v' \eqqcolon v_1 - \dots - v_{k}\coloneqq x)$ in $T'$ from $v'$
  to some of its descendant leaves $x\in L(T')=L(T)$ that satisfies
  $\mu(v_{i+1}) \preceq_S \mu(v_i)$ for all $1\le i < k$.  If $v$ lies on
  the path $P$, then the latter and transitivity of $\preceq_S$ immediately
  implies $\mu(v)\preceq_{S} \mu(v') = u_S v^*_S \prec_{S} u_S = \mu(u)$.
  Suppose for contradiction that $v$ is not a vertex in $P$. Then there
  must be some vertex $v_i(\ne v)$ with $1 \le i < k$ that is the last
  common ancestor of $v$ and $x$ in $T'$. In this case, $v_i$ must have at
  least two children in $T'$ and thus it was not suppressed. Since $v_i$
  furthermore lies on the path connecting $u$ and $v$, this contradicts
  that $v\in \child_T(u)$.  Hence, the case that $v$ is not a vertex in $P$
  does not occur.  Therefore, we have $\mu(v) \prec_{S} \mu(u)$, which
  together with the fact that $v\in \child_T(u)$ was chosen arbitrarily,
  implies that $\scen$ satisfies (S5).
\end{proof}

The example in Fig.~\ref{fig:algo-S6-violated} shows that
Algorithm~\ref{alg:recognition} is in general not guaranteed to return a
restricted scenario since it may violate (S6).
\begin{figure}[t]
  \centering
  \includegraphics[width=.75\textwidth]{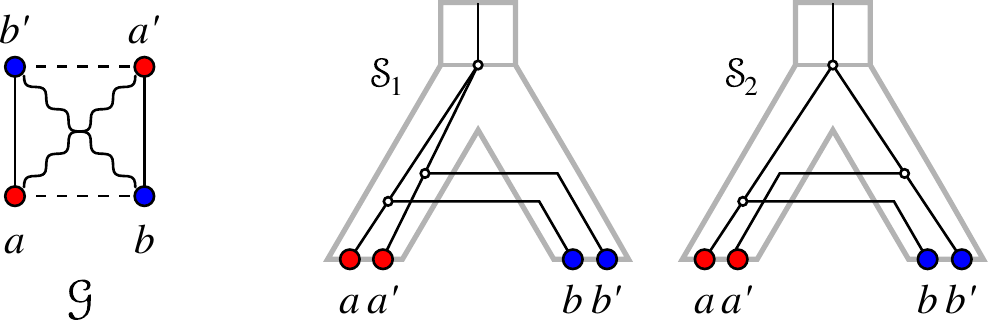}
  \caption{The graph 3-partition $\graphs = (\Gu, \Gg, \Ga, \sigma)$ used
      in Figure~\ref{fig:algo-1-example} as illustration of
      Algorithm~\ref{alg:recognition} is explained by different scenarios:
      Depending on the choice in Line~\ref{line:choose-v-S},
      Algorithm~\ref{alg:recognition} can return $\scen_1$
      as well as the restricted scenario $\scen_2$. 
      To ensure that \emph{always} a restricted scenario
      is returned we provide an alternative subroutine (summarized in 
      Algorithm~\ref{alg:restricted-scenario} below) that can be used in 
      Algorithm~\ref{alg:recognition}.}
\label{fig:algo-S6-violated}
\end{figure}
As we shall see in the following, however, we can construct such a scenario
for any valid input $\graphs = (\Gu, \Gg, \Ga, \sigma)$ by choosing the
vertex $v^*_S\in \child_{S}(u_S)$ in Line~\ref{line:choose-v-S} in a more
sophisticated manner.  More precisely, consider a connected component $C_i$
of $H_1$, for which we have created a corresponding vertex $u_i$ in
Line~\ref{line:create-u_i}).  If there is only one connected component
$C_j$ of $H_2$ such that $C_j\subseteq C_i$ (thus implying $C_j= C_i$),
then we proceed as in the original algorithm.  Otherwise, $C_i$ includes at
least two connected components of $H_2$.  In this case, there exists an
edge $xy\in E(H_1) \setminus E(H_2)$ with $x,y\in C_i$. From
Cor.~\ref{cor:H2-cc} and $H_2\subseteq H_1$ we obtain $x\in C_x \subseteq
C_i$ and $y\in C_y\subseteq C_i$ for two distinct connected components
$C_x$ and $C_y$ of $H_2$.  From the construction of the auxiliary graphs
$H_1$ and $H_2$ and $\sigma(L')\subseteq L(S(u_S))$, we know that $xy \in
E(\Gg)$.  Moreover, we have $\sigma(x)\preceq_S v_{\sigma(x)}$ and
$\sigma(y)\preceq_S v_{\sigma(y)}$ for distinct vertices
$v_{\sigma(x)},v_{\sigma(y)} \in \child_S(u_S)$ because otherwise $xy$
would be an edge in $H_2$.  Upon encountering $C_x$ and $C_y$ during
the iteration over connected components in Line~\ref{line:iterate-C_j},
we simply choose $v_{\sigma(x)}$ and $v_{\sigma(y)}$ in
Line~\ref{line:choose-v-S}, respectively.  Notice that this is in line with
the condition in Line~\ref{line:choose-v-S} because $\sigma(x)\in
\sigma(C_x)\cap L(S(v_{\sigma(x)}))$ and $\sigma(y)\in \sigma(C_y)\cap
L(S(v_{\sigma(y)}))$. For all other connected components, we simply choose
$v^*_S$ as in the original algorithm. These modifications of
Algorithm~\ref{alg:recognition} (which are restricted to the
\textbf{else}-block starting in Line~\ref{line:else}) are summarized in
Algorithm~\ref{alg:restricted-scenario}.

\begin{algorithm}[t]
  \small

  \caption{Alternative for the \textbf{else}-block starting in
    Line~\ref{line:else} of Algorithm~\ref{alg:recognition}. The
    modifications ensure that the returned scenario is restricted.}
  \label{alg:restricted-scenario}

  \begin{algorithmic}[0]
  \State compute $H_1$, $H_2$, and $H_3$ for $L'$ and $u_S$
  \For{\textbf{each} connected component $C_i$ of $H_1$}
    \State create vertex $u_i$ as a child of $\rho'$
    \State $\tT(u_i)\leftarrow \tS(u_{S})$ and $\mu(u_i)\leftarrow u_S$
    \State 
    $\mathcal{L} \leftarrow$ list of connected components $C_j$ of $H_2$
    such that $C_j\subseteq C_i$
    \If{$\vert \mathcal{L}\vert \ge 2$}
      \State choose an edge $xy\in E(H_1[C_i]) \setminus E(H_2[C_i])$
      \State identify the connected components $C_x$ and $C_y$ of $H_2$ that
      contain $x$ and $y$, resp., and
      \State \qquad $v_{\sigma(x)},v_{\sigma(y)} \in \child_S(u_S)$ for which
      $\sigma(x)\preceq_S v_{\sigma(x)}$ and
      $\sigma(y)\preceq_S v_{\sigma(y)}$
    \EndIf
  \EndFor

    \For{\textbf{each} $C_j$ in $\mathcal{L}$}
      \State create vertex $v_j$ as a child of $u_i$

      \If{$\vert \mathcal{L}\vert \ge 2$ and $C_j=C_x$}
        \State $v^*_S\leftarrow v_{\sigma(x)}$
      \ElsIf{$\vert \mathcal{L}\vert \ge 2$ and $C_j=C_y$}
        \State $v^*_S\leftarrow v_{\sigma(y)}$
      \Else 
        \State choose $v^*_S\in \child_S(u_{S})$ such that $\sigma(C_j)\cap
        L(S(v^*_S))\ne\emptyset$
        \EndIf
      \State $\tT(v_j)\leftarrow \tS(u_{S}) - \epsilon$ and
         $\mu(v_j)\leftarrow u_S v^*_S$\;

      \For{\textbf{each} connected component $C_k$ of $H_3$ such that
           $C_k\subseteq C_j$}
        \State identify $v_S\in \child_S(u_{S})$ such that
             $\sigma(C_k)\subseteq L(S(v_S))$
        \State connect \textsc{BuildGeneTree}$(C_k, v_S)$ as a child of $v_j$
      \EndFor 
    \EndFor
  \end{algorithmic}
\end{algorithm}

By the latter arguments we have only constrained choices that were
arbitrary in the original algorithm. All results for
Algorithm~\ref{alg:recognition} (with exception of the complexity results)
therefore remain valid for the modified version. As an immediate
consequence of Lemmas~\ref{lem:algo-explains},
\ref{lem:algo-fully-witnessed}, and~\ref{lem:algo-S5}, we therefore obtain:
\begin{fact}
  \label{obs:mod-algo-explains}
  The modifications of Algorithm~\ref{alg:recognition} summarized in
  Algorithm~\ref{alg:restricted-scenario} ensure that it returns a scenario
  that explains the valid input $\graphs$ and satisfies (S4) and (S5).
\end{fact}

For completeness we show that the modifications do not increase the time
complexity.
\begin{lemma}
  Algorithm~\ref{alg:recognition} with the modifications as summarized in
  Algorithm~\ref{alg:restricted-scenario} can be implemented to run in
  $O(\vert L\vert^4 \log\vert L\vert)$ time (for valid inputs
  $\graphs=(\Gu, \Gg, \Ga, \sigma)$ such that $\sigma$ is surjective).
\end{lemma}
\begin{proof}
  Re-using the arguments in the proof of Lemma~\ref{lem:algo-complexity},
  it suffices to show that, in the modified algorithm, the effort of the
  additional steps in one recursion step on $L'\subseteq L$ and some inner
  vertex $u_S\in V^0(S)$ (excluding the recursive calls) is bounded by
  $O(\vert L'\vert^2)$.
  We have already shown in the proof of Lemma~\ref{lem:algo-complexity} how
  the lists $\mathcal{L}$ of connected components $C_j$ of $H_2$ such that
  $C_j\subseteq C_i$ are obtained using breadth-first search with a total
  effort of $O(\vert L'\vert^2)$ time.  We can store, for each vertex $x\in
  L'$, a pointer to the connected component of $H_2$ in a hash table in
  $O(\vert L'\vert )$ time.  For a given connected component $C_i$ of
  $H_1$, choosing an edge $xy\in E(H_1[C_i]) \setminus E(H_2[C_i])$ is
  easily done by iterating over all pairs of vertices in $C_i$. Since
  distinct connected components of $H_1$ are vertex-disjoint, the overall
  effort for this is again bounded by $O(\vert L'\vert^2)$.  For a given
  connected component $C_i$ of $H_1$, identifying the respective connected
  components $C_x$ and $C_y$ and vertices $v_{\sigma(x)},v_{\sigma(y)} \in
  \child_S(u_S)$ can be done in constant time by querying the
  above-mentioned hash table and the LA data structure, respectively. Since
  $H_1$ has at most $O(\vert L'\vert )$ connected components, the total
  effort for the latter look-ups is bounded by $O(\vert L'\vert )$.
  Finally, checking whether $C_j=C_x$ and $C_j=C_y$ can clearly be done in
  constant time if we compare only pointers to the connected components.
  The total time complexity of the second \textbf{for}-loop in
  Algorithm~\ref{alg:restricted-scenario} is therefore the same as in the
  original algorithm.

  In summary, the total effort of one recursion step (excluding the
  recursive calls) is still bounded by $O(\vert L'\vert^2)$, which
  completes the proof.
\end{proof}

We note that scenario $\scen_2$ in Fig.~\ref{fig:algo-S6-violated} may
  be obtained from Algorithm~\ref{alg:recognition} using the subroutine in 
  Algorithm~\ref{alg:restricted-scenario} if the edge 
  $ab'\in E(H_1[C_i]) \setminus E(H_2[C_i])$
  is chosen (over the alternative choice $a'b$)
  in the \textbf{``if $\vert \mathcal{L}\vert \ge 2$ then''} block.

\begin{lemma}
  \label{lem:algo-satisfies-S6}
  Given a valid input $\graphs = (\Gu, \Gg, \Ga, \sigma)$, the scenario
  $\scen=(T,S,\sigma,\mu,\tT,\tS)$ returned by
  Algorithm~\ref{alg:recognition} with the modifications as summarized in
  Algorithm~\ref{alg:restricted-scenario} satisfies~(S6).
\end{lemma}
\begin{proof}
  Suppose that $\mu(u) \in V^0(S) = V(S) \setminus (L(S) \cup \{0_S\})$.
  Inspection of Algorithm~\ref{alg:recognition} shows that $u$ can only
  have been created in Line~\ref{line:create-u_i} in some recursion step on
  $L'\subseteq L$ and $u_S\in V^0(S)$. In particular, we have $\mu(u)=u_S$
  and $u$ corresponds to some connected component $C_i$ of $H_1$.  Consider
  the intermediate tree $T'$ constructed in the algorithm from which $T$ is
  obtained by adding the planted root $0_T$ and suppression of all inner
  vertices with a single child.  Since $u$ was not suppressed, we must have
  added at least to distinct vertices as children of $u$ in the same
  recursion step.  In particular, the output of the modified algorithm
  satisfies $\mu(v_j)=u_S v_S$ and $\mu(v_{j'})=u_S v'_S$ for two distinct
  children $v_j, v_{j'}$ of $u$ and two distinct vertices $v_S,v'_S \in
  \child_S(u_S)$.  Re-using the arguments in the proof of
  Lemma~\ref{lem:algo-fully-witnessed} and the fact that $\mu(v_j)=u_S v_S
  \prec_S u_S = \mu(u)$, we find a path $P'= (u\eqqcolon v'_1 - v_j
  \eqqcolon v'_2 - \dots - v'_{k'}\coloneqq x)$ in $T'$ from $u$ to some of
  its descendant leaves $x\in L(T')=L(T)$ that passes through $v_j$ and
  does not contain any HGT-edge, i.e., it holds $\mu(v'_{i+1}) \preceq
  \mu(v'_i)$ for all $1\le i < k'$. In particular
  $\sigma(x)=\mu(x)\prec_{S} \mu(v_j) = u_S v_S$.  Therefore, and because
  $T$ is obtained from $T'$ by adding $0_T$ and suppression of all vertices
  with a single child, we have $x\prec_{T} u$ and, moreover, the path $P=
  (u\eqqcolon v_1 - v_2 - \dots - v_k\coloneqq x)$ connecting $u$ and $x$
  in $T$ contains only vertices that are also contained in $P'$ in the same
  order.  We therefore conclude that $\mu(v_{i+1}) \preceq \mu(v_i)$ holds
  for all $1\le i < k$, i.e., $P$ does not contain any HGT-edge.
  Analogously, we find a descendant leaf $y\prec_{S} u$ such that the path
  from $u$ to $y$ in $T'$ passes through $v_{j'}$, the path from $u$ to $y$
  in $T$ does not contain HGT-edges, and furthermore $\sigma(y)\prec_{S}
  u_S v'_S$.

  By construction, we have $\lca_{T'}(x,y)=u$, which implies
  $\lca_{T}(x,y)=u$ since we only added $0_T$ and suppressed the vertices
  with a single child to obtain $T$ from $T'$.  The paths from $u$ to $x$
  and to $y$ in $T$ do not contain HGT-edges. Thus the path from $x$ to $y$
  in $T$ does not contain HGT-edges.  Finally $\sigma(x)\prec_{S} u_S v_S$
  and $\sigma(y)\prec_{S} u_S v'_S$ with $v_S$ and $v'_S$ being distinct
  children of $u_S$ implies $\lca_S(\sigma(x), \sigma(y)) = u_S = \mu(u)$.
  Taken together, the latter arguments imply that $\scen$ satisfies (S6).
\end{proof}

\begin{theorem}
  \label{thm:relaxed-iff-restricted}
  A graph $3$-partition $\graphs=(\Gu, \Gg, \Ga, \sigma)$ can be explained
  by a relaxed scenario if and only if it can be explained by a restricted
  scenario. In particular, Algorithm~\ref{alg:recognition} with the
  modifications summarized in Algorithm~\ref{alg:restricted-scenario}
  constructs a restricted scenario in this case.
\end{theorem}
\begin{proof}
  The \textit{if}-direction trivially holds since every restricted scenario
  is also a relaxed scenario.  Conversely, suppose $\graphs$ is explained
  by a relaxed scenario.  Then Algorithm~\ref{alg:recognition} with the
  modifications as summarized in Algorithm~\ref{alg:restricted-scenario}
  returns a scenario $\scen$ that explains $\graphs$ by
  Lemma~\ref{lem:algo-explains}.  By Lemmas~\ref{lem:algo-fully-witnessed},
  \ref{lem:algo-S5}, and~\ref{lem:algo-satisfies-S6}, respectively, $\scen$
  satisfies~(S4), (S5), and (S6), and thus, it is a restricted scenario.
\end{proof}

\begin{corollary}
  \label{cor:arbitrary-S}
  Let $\graphs=(\Gu, \Gg, \Ga, \sigma)$ be graph $3$-partition with vertex
  coloring $\sigma\colon L\to M$.  If $\graphs=(\Gu, \Gg, \Ga, \sigma)$ can
  be explained by a relaxed scenario, then, for every species tree $S^*$ on
  $M$ that agrees with $(\R_S(\graphs), \F_S(\graphs))$, there is a relaxed
  scenario $\scen=(T,S^*,\sigma,\mu,\tT,\tS)$ that explains $\graphs$.
  Moreover, $\scen$ can be chosen to be a restricted scenario.
\end{corollary}
\begin{proof}
  Suppose $\graphs=(\Gu, \Gg, \Ga, \sigma)$ can be explained by a relaxed
  scenario. By Thm.~\ref{thm:explainable-charac}, therefore, $\graphs$ is a
  valid input for Algorithm~\ref{alg:recognition} with the modifications in
  summarized in Algorithm~\ref{alg:restricted-scenario}.  Since the species
  tree $S$ constructed in Line~\ref{line:S} of
  Algorithm~\ref{alg:recognition} is an arbitrary tree $S^*$ on $M$ that
  agrees with $(\R_S(\graphs), \F_S(\graphs))$, i.e., not necessarily the
  tree constructed by \texttt{MTT} \cite{He:06}, 
  Obs.~\ref{obs:mod-algo-explains} immediately implies that there is a relaxed 
  scenario $\scen=(T,S^*,\sigma,\mu,\tT,\tS)$ that explains $\graphs$.  
  Moreover, if $\scen$ is constructed using the modified algorithm, then it is
  a restricted scenario by Thm.~\ref{thm:relaxed-iff-restricted}.
\end{proof}

\section{Explanation of EDT  Graphs by Relaxed Scenarios}
\label{sect:sole-EDT-recognition}

In the two preceding sections, we have
seen that it can be decided efficiently whether a given vertex-colored graph
$(G,\sigma)$ is an EDT graph provided we also know how the complement
$(\overline{G},\sigma)$ is partitioned into a putative LDT graph
$(\Ga,\sigma)$ and putative PDT graph $(\Gu,\sigma)$. It is of immediate
interest to understand whether the information on $(\Ga,\sigma)$ and
$(\Gu,\sigma)$ is necessary, or whether EDT
  graphs can also be recognized efficiently in isolation.
We consider the following decision problem:

\begin{problem}[\PROBLEM{EDT-Recognition}]\ \\
  \begin{tabular}{ll}
    \emph{Input:}    & A colored graph $(G, \sigma)$.\\
    \emph{Question:} & Is $(G,\sigma)$ an EDT graph?
  \end{tabular}
\end{problem}

As we shall see, \PROBLEM{EDT-Recognition} can be answered in
polynomial-time, if we suppose that the scenario explaining $(G, \sigma)$
is HGT-free while, for the general case, \PROBLEM{EDT-Recognition} is
NP-complete.  We start with a characterization of the EDT graphs that can
be explained by HGT-free relaxed scenarios. For this purpose, it will be
useful to note that edge-less LDT graphs rule out the existence of
HGT-edges in fully witnessed scenarios:
\begin{lemma}
  \label{lem:Gu-empty-HGT-free}
  If a relaxed scenario $\scen$ is fully witnessed and 
  $E(\Gu(\scen))=\emptyset$, then $\scen$ is HGT-free.
\end{lemma}
\begin{proof}
  Suppose for contradiction that $\scen=(T,S,\sigma,\mu,\tT,\tS)$ contains
  an HGT-edge $uv\in E(T)$ (where $v\prec_{T} u$), i.e., $\mu(u)$ and
  $\mu(v)$ are incomparable in $S$.  By assumption, $u$ has a witness $x\in
  L(T)$, and $v$ has a witness $y\in L(T)$. In particular, it holds
  $\sigma(x)=\mu(x)\preceq_{S} \mu(u)$ and $\sigma(y)=\mu(y)\preceq_{S}
  \mu(v)$ which, together with $\mu(u)$ and $\mu(v)$ being incomparable,
  implies that $\mu(u)\prec_{S} \lca_S(\sigma(x),\sigma(y))$. Moreover,
  since $uv$ is an HGT-edge and the path from $u$ to $x$ does not contain
  an HGT-edge, $x$ cannot be a descendant of $v$. Hence, $\lca_{T}(x,y)=u$.
  We now distinguish cases (a) $\mu(u)\in V(S)$ and (b) $\mu(u)\in E(S)$.
  In Case~(a), we have $\tT(u)=\tS(\mu(u))$ by Condition~\AX{S2} and
  $\tS(\mu(u)) < \tS(\lca_S(\sigma(x),\sigma(y)))$ as a consequence of
  $\mu(u)\prec_{S} \lca_S(\sigma(x),\sigma(y))$.  In Case~(b), we have
  $\mu(u)=ab\in E(S)$ and, by Condition~\AX{S3}, $\tT(u) <
  \tS(a)$. Moreover, $\mu(u)\prec_{S} \lca_S(\sigma(x),\sigma(y))$ implies
  $a\preceq_{S} \lca_S(\sigma(x),\sigma(y))$ by the definition of
  $\preceq_{S}$.  Hence, we have $\tT(u) < \tS(a) \le
  \tS(\lca_S(\sigma(x),\sigma(y)))$.  In summary, it holds
  $\tT(\lca_{T}(x,y))=\tT(u)<\tS(\lca_S(\sigma(x),\sigma(y)))$ and thus
  $xy\in E(\Gu(\scen))$ in both cases; a contradiction to
  $E(\Gu(\scen))=\emptyset$.  Therefore, $\scen$ must be HGT-free.
\end{proof}

The recognition of EDT
graphs can be achieved in polynomial-time in the HGT-free case.
\begin{theorem}
  \label{thm:EDT-HGT-free}
  Let $(\Gg=(L,E), \sigma)$ be a vertex-colored graph, and let $\R$ be the
  set of triples such that $\sigma(x)\sigma(y)\vert \sigma(z)\in \R$ iff
  $xz,yz\in E$ and $xy\notin E$ for some $x,y,z\in L$ of pairwise distinct
  colors.  Then $(\Gg,\sigma)$ is an EDT graph that can be explained by
  an HGT-free relaxed scenario if and only if it is a properly colored
  cograph and $\R$ is consistent. In particular, EDT graphs explained by
  HGT-free relaxed scenario can be recognized in $O(\vert L \vert^3 + \vert
  L \vert \vert\R \vert)$ time.
\end{theorem}
\begin{proof}
  Suppose $(\Gg,\sigma)$ is an EDT graph that is explained by the HGT-free
  relaxed scenario $\scen$. By Cor.~\ref{cor:equal-colors} and
  Lemmas~\ref{lem:HGT-free-Gg-cograph}, $(\Gg,\sigma)$ is a properly
  colored cograph. Suppose $xz,yz\in E$ and $xy\notin E$. Since in addition
  $\Gu(\scen)$ is edge-less by Cor.~\ref{cor:Gu-edgeless}, we have
  $xz,yz\notin E(\Ga(\scen))$ and $xy\in E(\Ga(\scen))$. Hence, we obtain
  $\R \subseteq \R_S(\graphs(\scen))$. By
  Thm.~\ref{thm:explainable-charac}, $\R_S(\graphs(\scen))$ and thus also
  its subset $\R$ are consistent.
    
  Now suppose $(\Gg,\sigma)$ is a properly colored cograph and $\R$ is
  consistent. Consider $\graphs=(\Gu\coloneqq (L,\emptyset), \Gg,
  \Ga\coloneqq \overline{\Gg})$.  Since $(\Gu,\sigma)$ is edge-less, it is
  a properly-colored cograph. Since $\Ga$ is the complement of the cograph
  $\Gg$, it is also a cograph. One easily verifies that $\R=\R_S(\graphs)$
  and thus there is a tree $S$ that displays all triples in
  $\R_S(\graphs)$.  Now consider a triple $XZ\vert Y \in \F_S(\graphs)$.
  By construction, this implies that there are $x,y,z\in L$ with pairwise
  distinct colors $X=\sigma(x)$, $Y=\sigma(y)$, and $Z=\sigma(z)$ such (a)
  $xz,yz \in E(\Gg)$ and $xy\notin E(\Gg)$ or (b) $xz,xy \in E(\Gg)$ and
  $yz\notin E(\Gg)$.  In Case~(a), we have $xz,yz \notin E(\Ga)$ and $xy\in
  E(\Ga)$ and thus $S$ displays the informative triple $XY\vert Z\in
  \R_S(\graphs)$.  In Case~(a), we have $xz,xy \notin E(\Ga)$ and $yz\in
  E(\Ga)$ and thus $S$ displays the informative triple $YZ\vert X\in
  \R_S(\graphs)$.  Therefore, the tree $S$ does not display the forbidden
  triple $XZ\vert Y$.  Since $XZ\vert Y \in \F_S(\graphs)$ was chosen
  arbitrarily, we can conclude that $S$ agrees with $(\R_S(\graphs),
  \F_S(\graphs))$.  In summary, therefore, we can apply
  Theorem~\ref{thm:explainable-charac} to conclude that $\graphs$ is
  explained by a relaxed scenario $\scen$.  By
  Theorem~\ref{thm:relaxed-iff-restricted}, $\scen$ can be chosen to be
  fully witnessed.  This together with the fact that $\Gu(\scen)=\Gu$ is
  edge-less and Lemma~\ref{lem:Gu-empty-HGT-free} yields that $\scen$ is
  HGT-free.  In summary, $(\Gg,\sigma)$ is an EDT graph that can be
  explained by a relaxed HGT-free scenario.
  
  Checking whether $(G=(L,E),\sigma)$ is properly colored can be done in
  $O(\vert E\vert)$ time, cographs can be recognized in $O(\vert
  L\vert+\vert E\vert)$ time \cite{Corneil:85}, extraction of $\R$ requires
  $O(\vert L\vert^3)$ time and testing whether $\R$ is consistent can be
  achieved in $O(\vert L\vert \vert \R \vert)$ time \cite{Aho:81}.  Thus,
  EDT graphs can be recognized in time $O(\vert L \vert^3 + \vert L \vert\,
  \vert\R \vert)$ in the HGT-free case.
\end{proof}

The examples in Fig.~\ref{fig:no-EDT} have shown that the connected
components of a given vertex-colored graph $(G,\sigma)$ are not
``independent'' in the sense that $(G,\sigma)$ is an EDT graph if and only
if all of its connected components are EDT graphs, since the components may
impose contradictory constraints on the species tree. However, we will show
next that we can assume w.l.o.g.\ that, if a relaxed scenario $\scen$
explaining $(\Gg,\sigma)$ exists, all pairs $x,y\in L$ that are in distinct
connected components of $\Gg$ form an edge in $\Ga(\scen)$. More precisely,
we have
\begin{lemma}
  Suppose $\graphs = (\Gu, \Gg, \Ga, \sigma)$ is explained by $\scen$ and 
  consider the edge set 
  $F\coloneqq\{xy \mid x,y\in L
  \text{ are in distinct connected components of } \Gg\}$. 
  Then $\graphs'=(\Gu', \Gg, \Ga', \sigma)$ where 
  $\Gu'\coloneqq (L, E(\Gu) \setminus F)$ and
  $\Ga'\coloneqq (L, E(\Ga) \cup F)$
  is explained by a relaxed scenario $\scen'$.
\end{lemma}
\begin{proof}
  Observe first that all pairs $x,y\in L$ that are in distinct connected
  components of $\Gg$ satisfy $xy\in E(\Ga')$.  By
  Theorem~\ref{thm:explainable-charac}, $\Gu$ and $\Gg$ are properly
  colored, $\Gu$ and $\Ga$ are cographs, and $(\R_S(\graphs),
  \F_S(\graphs))$ is consistent.  Since $\Gu'$ is a subgraph of $\Gu$, it
  is still properly colored.
    
  Suppose for contradiction that $\Gu'$ is not a cograph, i.e., it contains
  an induced $P_4=a-b-c-d$. In this case, $ab,bc,cd \in E(\Gu')$ implies
  that $ab,bc,cd \notin F$ and thus, that $a$ and $b$, $b$ and $c$ as well
  as $c$ and $d$ are contained in the same connected component of $\Gg$.
  Consequently, $a$, $b$, $c$, and $d$ are contained in a single connected
  component of $\Gg$, which implies that $ac,bd,ad\notin F$.  Therefore,
  $a-b-c-d$ is also an induced $P_4$ in $\Gu$; a contradiction.  Now
  suppose for contradiction that $\Ga'$ contains an induced
  $P_4=a-b-c-d$. In this case, $ac,bd,ad \notin E(\Ga')$ implies $ac,bd,ad
  \notin \Ga$ and $ac,bd,ad \notin F$.  The latter in particular implies
  that $a$, $b$, $c$, and $d$ are contained in a single connected component of
  $\Gg$ and thus $ab,bc,cd \notin F$.  It follows that $ab$, $bc$, and $cd$
  must also be edges in $\Ga$ and, thus, $a-b-c-d$ is an induced $P_4$ in
  $\Ga$; a contradiction. In summary, $\Gu'$ and $\Ga'$ are cographs.
    
  We continue with showing that $(\R_S(\graphs'), \F_S(\graphs'))$ remains
  consistent.  Suppose $XY\vert Z\in \R_S(\graphs')$, i.e., there are $x,y,z\in
  L$ with pairwise distinct colors $X=\sigma(x)$, $Y=\sigma(y)$, and
  $Z=\sigma(z)$ such that (a') $xz, yz \in E(\Gu')$ and $xy \notin
  E(\Gu')$, or (b') $xy \in E(\Ga')$ and $xz, yz \notin E(\Ga')$.  In both
  cases, we can apply similar arguments as before to conclude that
  $xy,xz,yz\notin F$. Thus, $xz, yz \in E(\Gu)$ and $xy \notin E(\Gu)$, and
  $xy \in E(\Ga)$ and $xz, yz \notin E(\Ga)$, respectively. This in turn
  implies $XY\vert Z\in \R_S(\graphs)$.  Hence, we have
  $\R_S(\graphs')\subseteq\R_S(\graphs)$.  Moreover, $\F_S(\graphs')$ does
  only depend on the (non-)edges of $\Gg$ and since $\Gg$ remained
  unchanged in $\graphs'$, we have $\F_S(\graphs') = \F_S(\graphs)$.  The
  latter two arguments together with $(\R_S(\graphs), \F_S(\graphs))$ being
  consistent imply that $(\R_S(\graphs'), \F_S(\graphs'))$ is also
  consistent.
  
  In summary, $\Gu'$ and $\Gg$ are properly colored, $\Gu'$ and $\Ga'$ are
  cographs, and $(\R_S(\graphs'), \F_S(\graphs'))$ is consistent.
  Theorem~\ref{thm:explainable-charac} therefore implies that $\graphs'$ is
  explained by a relaxed scenario $\scen'$.
  \end{proof}

\begin{corollary}
  \label{cor:EDT-CCs}
  If $(\Gg,\sigma)$ is an EDT graph, then it is explained by a relaxed
  scenario $\scen$ that satisfies $xy\in E(\Ga(\scen))$ for all $x,y\in L$
  that are contained in distinct connected components of $\Gg$.
\end{corollary}

Let us now turn the general case of \textsc{EDT-Recognition}. We show
 that it is NP-hard by reducing from a problem of deciding whether there
is a tree that displays a given set of fan triples and a suitable choice of
rooted triples. The precise problem statement requires some definitions.
Let $U$ be a set.  Let $\CF$ be a set of fan triples whose leaves are in
$U$, and let $\CR$ be a set of unordered pairs of rooted triples of the
form $\rpair{x}{y}{z}$ with $x,y,z \in U$.  We say that a tree $S^*$ on the
leaf set $U$ \emph{satisfies} \emph{$(\CF, \CR)$} if the following holds:
\begin{itemize}
\item[] For each $x\vert y\vert z \in \CF$, $S^*$ displays
  $x\vert y\vert z$;
\item[] For each $\rpair{x}{y}{z} \in \CR$, $S^*$ displays either
  $xy\vert z$ or $xz\vert y$.
\end{itemize}
This suggests the following decision problem.
\begin{problem}[\PROBLEM{$(\CF,\CR)$-Satisfiability}]\ \\
  \begin{tabular}{ll}
    \emph{Input:}    & A tuple $(U, \CF, \CR)$ where $U$ is a set,
                       $\CF$ is a set of fan triples and \\
                     & $\CR$ is a set of pairs of rooted triples of
                       the form $\rpair{x}{y}{z}$.\\
    \emph{Question:} & Does there exist a tree $S^*$ on leaf set $U$
                       that satisfies $(\CF, \CR)$?
  \end{tabular}
\end{problem}

Jansson et al.\ \cite{jansson2018determining} showed that a slightly different version of
\PROBLEM{$(\CF,\CR)$-Satisfiability}, known as
\PROBLEM{$(F^{+-})$-Consistency}, is NP-hard.  In the
\PROBLEM{$(F^{+-})$-Consistency} problem the input are two sets $F^+$ and
$F^-$ of fan triples and one asks for a tree that displays all fan triples
in $F^+$ but none of the ones in $F^-$.  The latter is equivalent to asking
for a tree that that displays all fan triples in $F^+$ and that displays
for every $x\vert y\vert z\in F^-$ exactly one of the triples $xy\vert z$,
$xz\vert y$, or $yz\vert x$.  This translated to a slightly different
version of \PROBLEM{$(\CF, \CR)$-Satisfiability} by requiring (i) the
elements of $\CR$ to be of the form $\{xy\vert z, xz\vert y, yz\vert x\}$
and (ii) that one of the three triples must be displayed by the final
tree. For our purposes, we must restrict $\CR$ to pairs of triples instead
of triple sets of size $3$.  The NP-hardness proof in
\cite{jansson2018determining} can be adapted to establish the following
result:
\begin{theorem} \label{thm:CFCR-NPc}
  \PROBLEM{$(\CF,\CR)$-Satisfiability} is NP-complete.
\end{theorem}
\begin{proof}
See Appendix. 
\end{proof}
Theorem \ref{thm:CFCR-NPc}, in turn, can be used to prove 
\begin{theorem}\label{thm:EDT-rec-NPc}
\PROBLEM{EDT-Recognition} is NP-complete. Moreover, it remains NP-complete
if the input graph $(G,\sigma)$ is a cograph.
\end{theorem}
\begin{proof}
See Appendix. 
\end{proof}

\section{Explanation of PDT Graphs by Relaxed Scenarios}

If only the information of $\Gu \in \graphs$ is available, it can be tested
whether $\Gu$ is an LDT graph and, in the affirmative case, a relaxed
scenario that explains $\Gu$ can be constructed in polynomial-time
\cite{Schaller:21f}. In contrast, we have seen above that the problem of
recognizing an EDT graph is NP-hard (Theorem~\ref{thm:EDT-rec-NPc}).
This begs the question whether recognition of PDT graphs is an
  easy or hard task.
  
\begin{theorem}\label{thm:char-PDT}
  A graph $(G,\sigma)$ is a PDT graph if and only if the following
  conditions are satisfied:
  \begin{enumerate}[noitemsep]
  \item $G$ is a cograph, and
  \item $(\overline G,\sigma)$ is properly colored, and 
  \item The set of triples $R(G) \coloneqq \{\sigma(x)\sigma(y)\vert
    \sigma(z) \colon xy \in E(G) \text{ and } xz, yz \notin E(G) \text{ and
    }\\ \sigma(x),\sigma(y),\sigma(z) \text{ are pairwise distinct}\}$
    is consistent.
  \end{enumerate}
  In particular, it can be verified if $(G,\sigma)$ is a PDT graph and, in the
  affirmative, a scenario that explains $(G,\sigma)$ can be constructed in
  polynomial time.
\end{theorem}
\begin{proof}
  Suppose that $(G,\sigma)$ is a PDT graph. Hence, there is a relaxed
  scenario $\scen$ such that $G=\Ga(\scen)$. By Lemma \ref{lem:Ga-cograph},
  $G$ must be a cograph. Since $G=\Ga(\scen)$, its complement $\overline G$
  comprises all edges of $\Gg(\scen)$ and $\Gu(\scen)$. By
  Cor.\ \ref{cor:equal-colors}, $\Gg(\scen)$ and $\Gu(\scen)$ are always
  properly colored and so $(\overline G,\sigma)$ is also properly
  colored. The set $R(G)$ is precisely the set of triples as specified in
  Def.\ \ref{def:inf-forb-triples}(b') and, in particular, $R(G)\subseteq
  \R_S(\graphs)$ where $\graphs=(\Gu(\scen), \Gg(\scen), \Ga(\scen),
  \sigma)$.  By Theorem \ref{thm:explainable-charac}, $(\R_S(\graphs),
  \F_S(\graphs))$ is consistent, an thus in particular $R(G)$ is
  consistent.

  Conversely, assume that $(G,\sigma)$ satisfies Conditions (1), (2)
    and (3). Consider $\graphs=(\Gu, \Gg, \Ga, \sigma)$ such that $\Ga =
  G$, $\Gg=(V(G),\emptyset)$ and $\Gu = \overline G$. Since $G$ is a
  cograph and $\Gu = \overline G$, Prop.\ \ref{prop:cograph} implies that
  $\Gu$ is a cograph. Moreover, by Condition (2), $(\Gu,\sigma)$ is a
  properly colored cograph. Since there are no edges in $\Gg$, it follows
  that $\Gg$ is also a properly colored cograph. Since $\Gg$ is edge-less,
  we have $\F_S(\graphs) = \emptyset$. Moreover, since $G$ is the
  complement of $\Gu$, Def.\ \ref{def:inf-forb-triples}(b') and the
  definition of $R(G)$ imply $R(G) = \R_S(\graphs)$. Condition (3) now
  implies that $\R_S(\graphs)$ is consistent. Together with $\F_S(\graphs)
  = \emptyset$ this implies that $(\R_S(\graphs), \F_S(\graphs))$ is
  consistent.  Hence, all conditions of Theorem
  \ref{thm:explainable-charac} are satisfied and we conclude that there is
  a relaxed scenario that explains $\graphs=(\Gu, \Gg, \Ga, \sigma)$. In
  particular, $G=\Ga$ is a PDT graph.
  Re-using the arguments in the proof of Lemma~\ref{lem:algo-complexity}, 
  we can construct a scenario for $\graphs=(\Gu, \Gg, \Ga, \sigma)$ (and thus 
  for $G=\Ga$ in $O(\ell^4 \log \ell)$ where $\ell=\max(\vert L \vert, \vert 
  \sigma(L) \vert)$.
\end{proof}

We note that PDT graphs can be recognized faster than the construction
of an explaining scenario with the help of Theorem~\ref{thm:char-PDT}.
Cographs can be recognized in $O(\vert V\vert+\vert E\vert)$ time
\cite{Corneil:85} and $O(\vert V\vert^2)$ operations are sufficient to
verify that the complement of $G$ is properly colored. The triple set
$R(G)$ contains at most $O(\vert\sigma(V)\vert^3)$ triples which can be
constructed in $O(\vert V\vert^3)$ time. The Aho \emph{et al.} algorithm
checks triple consistency in $O(\vert R\vert\,\vert V\vert)$ time. Hence,
PDT graphs can be recognized in $O( \vert V\vert(\vert V\vert^2 +
\vert\sigma(V)\vert^3))$ time.

\section{Orthology and Quasi-Orthology}
\label{sect:orthology}

Most of the mathematical results concerning orthology have been obtained in
an HGT-free setting. There, a pair of genes $x$ and $y$ is orthologous if
their last common ancestor $\lca_T(x,y)$ coincides with the last common
ancestor of the two species in which they reside \cite{Fitch:00}.
Thus, we expect a close connection between orthology and the graph
$\Gg(\scen)$.  Thm.~\ref{thm:EDT-HGT-free} in the previous section,
furthermore, is reminiscent of the characterization of orthology graphs
that can be reconciled with species trees in HGT-free duplication/loss
scenarios \cite{HernandezRosales:12a,Hellmuth:17}. We therefore close
this contribution by connecting the graph $\Gg(\scen)$ with different
notions of orthology in scenarios with HGT that have been discussed in
the literature.

Disagreements on the ``correct'' definition of orthology in the
presence of HGT stem for the fact that, in general, pairs of genes
originating from a speciation event may be separated by HGT, and thus
become xenologs. They may even eventually reside in the same species and
therefore appear as paralogs. Choanozoa, for example, have two CCA-adding
enzymes, one vertically inherited through the eukaryotic lineage, the other
horizontally acquired from a bacterial lineage \cite{Betat:15a}. To
accommodate such differences, Darby et al.\ \cite{Darby:17} proposed a classification of
subtypes of xenology and, in line with \cite{Fitch:00}, reserve the terms
\emph{ortholog} and \emph{paralog} to situations in which the path between
$x$ and $y$ does not contain an HGT event. In this section, we briefly
survey notions of orthology that have ``natural'' definitions in the
setting of relaxed scenarios and explore their mathematical properties and
their relationships with EDT graphs. 

\begin{definition}\label{def:wqo}
  Let $\scen=(T,S,\sigma,\mu,\tT,\tS)$ be a relaxed scenario.  Two distinct
  vertices $x,y\in L(T)$ are \emph{weak quasi-orthologs} if
  $\mu(\lca_T(x,y))\in V^0(S)$.
\end{definition}
Def.~\ref{def:wqo} is, in essence, Walter Fitch's original, purely
event-based definition of orthology \cite{Fitch:70}.  The graph
$\wQO(\scen)$ with vertex set $L(T)$ and the weak quasi-orthologous pairs
as its edges is the \emph{weak quasi-orthology graph} of $\scen$.

In later work, Walter M.\ Fitch \cite{Fitch:00} emphasizes the condition that ``the common
ancestor lies in the cenancestor (i.e., the most recent common ancestor) of
the taxa from which the two sequences were obtained'', which translates to
the following notion:
\begin{definition}
  \label{def:ortho}
  Let $\scen=(T,S,\sigma,\mu,\tT,\tS)$ be a relaxed scenario. Then two
  distinct genes $x,y\in L(T)$ are \emph{strict quasi-orthologs} if
  $\mu(\lca_T(x,y))=\lca_S(\sigma(x),\sigma(y))$.
\end{definition}
The graph $\sQO(\scen)$ with vertex set $L(T)$ and the strict
quasi-orthologous pairs as its edges is the \emph{strict quasi-orthology
graph} of $\scen$. By Obs.~\ref{obs:generic}, all edges of $\Gg$ form
  strictly quasi-orthologous pairs in the scenarios produced by
  Algorithm~\ref{alg:recognition}.

Later definitions explicitly exclude xenologs
\cite{Gray:83,Fitch:00}. Translating the concept of orthology used by
Darby et al.\ \cite{Darby:17} to our notation yields
\begin{definition}
  Let $\scen=(T,S,\sigma,\mu,\tT,\tS)$ be a relaxed scenario.  Two distinct
  vertices $x,y\in L(T)$ are \emph{weak orthologs} if $\mu(\lca_T(x,y))\in
  V^0(S)$ and $\lambda(e)=0$ for all edges $e$ along the path between $x$
  and $y$ in $T$.
\end{definition}
The graph $\wO(\scen)$ with vertex set $L(T)$ and the pairs of weak
orthologs as its edges will be called the \emph{weak orthology graph} of
$\scen$.  The most restrictive notion of orthology is obtained by enforcing
both the matching of last common ancestors and the exclusion of horizontal
transfer:
\begin{definition}
  Let $\scen=(T,S,\sigma,\mu,\tT,\tS)$ be a relaxed scenario.  Two distinct
  vertices $x,y\in L(T)$ are \emph{strict orthologs} if $\mu(\lca_T(x,y)) =
  \lca_S(\sigma(x), \sigma(y))$ and $\lambda(e)=0$ for all edges $e$ along
  the path between $x$ and $y$ in $T$.
\end{definition}
The graph $\sO(\scen)$ with vertex set $L(T)$ and the pairs of (strict)
orthologs as its edges will be called the \emph{(strict) orthology graph}
of $\scen$. We note that strict orthologs also appear in the definition of
property (S6): A relaxed scenario satisfies (S6) if and only if $\mu(u)\in
V^0(S)$ implies that there is a pair of strict orthologs $x$ and $y$ with
$\lca_T(x,y)=u$.  The alternative notions of orthology and the proposed
terminology are summarized in Table~\ref{tab:ortho-terminology}.

\begin{table}
  \caption{Summary of the alternative notions of orthology in the presence
    of HGT events.}
  \label{tab:ortho-terminology}
  \centering
  \renewcommand{\arraystretch}{1.3}
  \begin{tabular}{|l||c|c|}
    \hline
    Reconciliation condition & HGT irrelevant  & HGT excluded \\
    \hline
    $\mu(\lca_T(x,y)) \in V^0(S)$  & $\wQO(\scen)$ &  $\wO(\scen)$ \\
    & weak quasi-ortholog  & weak ortholog \\
    \hline
    $\mu(\lca_T(x,y)) = \lca_S(\sigma(x), \sigma(y))$ &
    $\sQO(\scen)$ & $\sO(\scen)$ \\
    & strict quasi-ortholog & (strict) ortholog \\
    \hline
  \end{tabular}
\end{table}

From $\mu(\lca_T(x,y))=\lca_S(\sigma(x),\sigma(y))$, we obtain
$\mu(\lca_T(x,y))\in V(S)$. Furthermore, if $x$ and $y$ are distinct, then
$\lca_T(x,y)$ is not a leaf and (S1) in the definition of relaxed scenarios
implies that $\mu(\lca_T(x,y))$ is also not a leaf.  Hence we have:
\begin{fact}
  If $x,y\in L$ are distinct and
  $\mu(\lca_T(x,y))=\lca_S(\sigma(x),\sigma(y))$, then $\mu(\lca_T(x,y))\in
  V^0(S)$ for every relaxed scenario $\scen=(T,S,\sigma,\mu,\tT,\tS)$.
\end{fact}

As an immediate consequence, every strict quasi-ortholog is a weak
quasi-ortholog and every strict ortholog is a weak ortholog.  Furthermore
strict or weak orthologs are strict or weak quasi-orthologs, respectively.
In terms of the corresponding graphs, we therefore have the following
subgraph relations:
\begin{equation}
  \label{eq:ortho-subgraph-rel}
  \sO(\scen)\subseteq \sQO(\scen), \qquad
  \wO(\scen)\subseteq \wQO(\scen), \qquad
  \sQO(\scen)\subseteq \wQO(\scen), \qquad
  \sO(\scen)\subseteq \wO(\scen).
\end{equation}
That is, we have $\sO(\scen)\subseteq \sQO(\scen)\subseteq \wQO(\scen)$ and
$\sO(\scen)\subseteq \wO(\scen)\subseteq \wQO(\scen)$, while $\sQO(\scen)$
and $\wO(\scen)$ are incomparable w.r.t.\ the subgraph relation.

\begin{lemma}\label{lem:wqo-cograph}
  The weak quasi-orthology graph $\wQO(\scen)$ and the weak orthology graph
  $\wO(\scen)$ are cographs for every relaxed scenario $\scen$.
\end{lemma}
\begin{proof}
 Let $\scen=(T,S,\sigma,\mu,\tT,\tS)$ be a relaxed scenario.  Consider the
 labeling $t\colon V^0(T)\to\{0,1\}$ with $t(u)=1$ iff $\mu(u)\in
 V^0(S)$. We have $xy\in E(\wQO(\scen))$ if and only if
 $t(\lca_T(x,y))=1$. Thus $(T,t)$ is a cotree that explains
 $\wQO(\scen)$. By Prop.\ \ref{prop:cograph}, $\wQO(\scen)$ is a cograph.
 
 Consider  $(T,t)$ and remove all HGT-edges from $T$ to obtain the
 forest $(T^*,t)$.  Although the tree(s) in $(T^*,t)$ are not necessarily
 phylogenetic, we can obtain a cograph $G$ with edges $xy\in E(G)$
 precisely if $x,y$ are leaves of a connected component of $(T^*,t)$ and
 $t(\lca_{T^*}(x,y))=1$.  One easily verifies that any two leaves $x$ and
 $y$ in a connected component of $T^*$ satisfy
 $\lca_{T^*}(x,y)=\lca_T(x,y)$.  Therefore, $xy\in E(G)$ precisely if the
 path connecting $x$ and $y$ in $T$ does not contain an HGT edge and
 $t(\lca_T(x,y))=1$ (or, equivalently $\mu(u)\in V^0(S)$).  Consequently,
 $G=\wO(\scen)$ and thus, $\wO(\scen)$ is a cograph.
\end{proof}
It is worth noting that $xy\in E(\wQO(\scen))$ does not imply
$\sigma(x)\ne\sigma(y)$, i.e., $(\wQO(\scen),\sigma)$ is not necessarily
properly colored. The genes $a$ and $a'$ in Figure~\ref{fig:2-colP4} serve
as an example. Now consider the two relaxed scenarios $\scen$ as shown in
Fig.~\ref{fig:EDT-P4-4col}.  In both cases, one observes that $\Gg(\scen) =
\sQO(\scen)$. In each case, $\Gg(\scen)$ contains an induced
$P_4$. Therefore, we obtain
\begin{fact}
  In general, $\sQO(\scen)$ is not a cograph.  
\end{fact}
  
\begin{lemma}
  \label{lem:so-cograph}
  The strict orthology graph $\sO(\scen)$ is a cograph for every relaxed
  scenario $\scen$. 
\end{lemma}
\begin{proof}
  Let $\scen=(T,S,\sigma,\mu,\tT,\tS)$ be a relaxed scenario. Note that
  $\sO(\scen)\subseteq\wO(\scen)$. Furthermore, if $xx'\in E(\wO(\scen))$,
  then $x$ and $x'$ are leaves in the same subtree of the forest $F(T)$
  obtained by removing all HGT edges from $T$, i.e., $x$ and $x'$ are
  witnesses of $\lca_T(x,x')$.  By definition, we have
  $\sO(\scen)\ne\wO(\scen)$ if and only if there are two vertices $x,x'\in
  L(T)$ with $\mu(\lca_T(x,x'))\in V^0(S)$ but $\mu(\lca_T(x,x'))\neq
  \lca_S(\sigma(x),\sigma(x'))$, and there is no HGT-edge on the path
  between $x$ and $x'$ in $T$. Note that the latter condition is equivalent
  to $x$ and $x'$ being witnesses of $\lca_T(x,x')$.  In this case, $xx'\in
  E(\wO(\scen))$ but $xx'\notin E(\sO(\scen))$ and Lemma~\ref{lem:order}
  implies $\lca_S(\sigma(x),\sigma(x'))\prec_S \mu(\lca_T(x,x'))$.  In the
  following, set $p\coloneqq \lca_T(x,x')$, $w\coloneqq \mu(p)$,
  $\sO\coloneqq \sO(\scen)$ and $\wO\coloneqq \wO(\scen)$.
  
  We proceed by modifying $(T,\tT)$ and the reconciliation map $\mu$ to
  obtain a scenario $\scen' = (T',S,\sigma,\mu',\tT',\tS)$ such that
  $\sO=\sO(\scen')$ remains unchanged and the edge $xx'$ is removed from
  $\wO$. This, in particular, ensures that
  $\sO\subseteq\wO(\scen')\subsetneq\wO$ holds.

  Since $\lca_S(\sigma(x),\sigma(x'))\prec_S w = \mu(p)$, and both $x$ and
  $x'$ are witnesses of $p$, there is a unique child $w^*\in\child_S(w)$
  such that $\lca_S(\sigma(x),\sigma(x'))\preceq_S w^*$.  For this vertex
  $w^*$, let $A^*\subseteq \child_T(p)$ be the subset of all children $q$
  of $p$ that satisfy (i) $q$ has a witness and (ii) for every witness $y$
  of $q$ holds $\sigma(y)\in L(S(w^*))$. By construction, the unique
  children $q_x$ and $q_{x'}$ of $p$ that satisfy $x\preceq_T q_x$ and
  $x'\preceq_T q_{x'}$ are contained in $A^*$, i.e., $A^*\neq
  \emptyset$. Moreover, for any two distinct $q_1,q_2\in A^*$ and all
  $x_1\in L(T(q_1))$ and $x_2\in L(T(q_2))$ such that $x_1$ is a witness of
  $q_1$ and $x_2$ is a witness $q_2$, we have
  $\lca_S(\sigma(x_1),\sigma(x_2))\preceq_S w^*$.  Note that $pq$ cannot be
  an HGT-edge of $T$ for all $q\in A^*$, since incomparability of $\mu(p)$
  and $\mu(q)$ would imply that at least one edge $uv$ along the path from
  $q$ to its witness $x_q$ must satisfy that $\mu(u)$ and $\mu(v)$ are
  incomparable (otherwise, condition~(ii) in the construction of $A^*$ is
  is not possible).  Thus, if $pq\in E(T)$ is an HGT edge for some
  $q\in\child_T(p)$, then $q\notin A^*$.
      
  Now construct a modified gene tree $T'$ as follows: If $A^*= \child_T(p)$
  we set $T'=T$ and relabel $p$ as $p^*$.  Otherwise, we insert an
  additional vertex $p^*$ into $T$ that has $p$ as its parent and the
  vertices $q_i\in A^*$, $1\le i\le \vert A^*\vert$ as its children. Note
  that by construction $w^*$ has at least 2 children. The time map for the
  modified tree is set by $\tau_{T'}(v)=\tau_{T}(v)$, $v\in V(T)$, and
  $\tau_{T'}(p^*)=\tau_{T}(p)-\epsilon$ for sufficiently small
  $\epsilon>0$.  Since we started with a relaxed scenario that explains
  $\sO$, $T'$ remains a phylogenetic tree.  Moreover, we define the
  modified reconciliation $\mu'$ by setting $\mu(p^*)=ww^*\in E(S)$ and
  $\mu'(v)=\mu(v)$ for all $v\in V(T')\setminus\{p^*\}$ and set
  $\scen'\coloneqq(T',S,\mu',\sigma,\tau_{T'},\tS)$. By construction,
  $\lca_{T'}(x,x')=p^*$ and thus, $\mu(p^*)\in E(S)$ implies $xx'\notin
  E(\wO(\scen'))$.  Furthermore, if $\lca_T(y_1,y_2)=p$, $y_1\in L(T(q_1))$
  for some $q_1\in A^*$ and $y_2\in L(T(q_2))$ for some $q_2\in
  \child_T(p)\setminus A^*$, then $\lca_{T'}(y_1,y_2)=p$ because $y_2$ is
  not a descendant of $p^*$ in $\scen'$. Finally, if $\lca_{T}(y_1,y_2)\ne
  p$, then $\lca_{T'}(y_1,y_2)=\lca_{T}(y_1,y_2)$. The latter two arguments
  together with the fact that the reconciliation maps for $T$ and $T'$
  coincide for all vertices distinct from $p^*$ imply
  $\sO(\scen')=\sO$. Furthermore, $x_1x_2\in E(\wO(\scen'))$ if and only if
  $x_1x_2\in E(\wO)$ and $\lca_{T'}(x_1,x_2)\neq p^*$.  In particular,
  $\vert E(\wO(\scen')\vert < \vert E(\wO)\vert$. The modification of
  $\scen$ also preserves witnesses: if $x$ is a witness of $v\ne p$ in
  $\scen$ then $x$ remains a witness of $v$ in $\scen'$; if $x$ is a
  witness of $p$ in $\scen$ then it is a witness of $p^*$ in $\scen'$ and,
  since $pp^*$ is not a HGT-edge, $x$ remains a witness of $p$. Thus $q\in
  A^*$ has a witness $x$ that is also a witness of $p$ in both $\scen$ and
  $\scen'$, and a witness of $p^*$ in $\scen'$. In particular, therefore,
  $p^*q$ with $q\in A^*$ is not an HGT edge.  Conversely, if $pq\in E(T)$
  is an HGT edge in $\scen$, $pq$ is also an HGT edge in $\scen'$ because
  $\mu'(p)=\mu(p)$ and $\mu'(q)=\mu(q)$ and $S$ remains unchanged.  The
  latter argument holds for all HGT edges in $\scen$, resp., $\scen'$.
  Therefore, $uv$ is an HGT-edge in $\scen$ if and only if $uv$ it an HGT
  edge in $\scen'$.  In particular, therefore, if the path from $u\in
  V^0(T)$ to the leaf $x\in L(T)$ is HGT-free in $\scen$, then it is also
  HGT-free in $\scen'$.

  Repeating this construction produces a finite sequence of scenarios
  $\scen=\scen_0,\scen_1,\dots,\scen_k$ with the same strict orthology
  graphs $\sO=\sO(\scen_1)=\dots=\sO(\scen_k)$ and in each step strictly
  reduces the number of edges in the weak orthology graph, i.e.,
  $\wO(\scen_{i})\subsetneq \wO(\scen_{i-1})$ for $1\le i\le k$ as long as
  in $\scen_{i-1}$ there is a vertex $p$ with a set $A^*$ with $\vert
  A^*\vert\ge 2$. Eventually we arrive at a relaxed scenario $\scen_k$ with
  a refined gene tree $T_k$ that contains no vertex $p$ with set $A^*$ as
  defined above. In $\scen_k$, therefore, $w=\mu_k(\lca_{T_k}(x,y))\in V^0$
  implies $\lca_S(\sigma(x),\sigma(y))=w$, which in turn implies
  $\wO(\scen_k)=\sO(\scen_k)=\sO$. The assertion now follows since
  $\wO(\scen_k)$ is a cograph by Lemma \ref{lem:wqo-cograph}.
\end{proof} 

The modification of a relaxed scenario $\scen$ in the proof of
Lemma~\ref{lem:so-cograph} only affects the last common ancestors of pairs
of genes $x,x'$ with $\mu(\lca_T(x,x;))\succ_S
\lca_S(\sigma(x),\sigma(x'))$ and thus $xy\in E(\Gu)$. Furthermore, in the
modified scenario $\scen'$, by construction we still have
$\mu(\lca_{T'}(x,x;))\succ_S \lca_S(\sigma(x),\sigma(x'))$, since either
$\lca_{T'}(x,x')=\lca_{T}(x,x')$ or
$\tau_T(\lca_{T'}(x,x'))=\tau_T(\lca_{T'}(x,x'))-\epsilon$ for an
arbitrarily small $\epsilon$. Therefore, we have
$\graphs(\scen)=\graphs(\scen')$ in each step, which immediately implies
\begin{proposition}
  A graph 3-partition $\graphs$ is explained by a relaxed scenario if
  and only if it is explained by a relaxed scenario satisfying
  $\sO(\scen)=\wO(\scen)$.
\end{proposition}

Finally, we show that every valid input $\graphs=(\Gu, \Gg, \Ga,
  \sigma)$ has an explanation such that the EDT graph $\Gg$ represents the
  strict quasi-orthologs.  This explanation can, in particular, by obtained
  with Alg.~\ref{alg:recognition}.  To see this, we first provide
  \begin{lemma}\label{lem:Gg-subset-sQO}
    Let $\scen$ be a relaxed scenario. Then $\sQO(\scen)\subseteq
    \Gg(\scen)$.
  \end{lemma}
  \begin{proof}
    Assume that $xy\in E(\sQO(\scen))$. Thus we have $x\neq y$ and
    $\mu(\lca_T(x,y))=\lca_S(\sigma(x), \sigma(y))\in V(S)$, which in
    turn yields $\tS(\mu(\lca_T(x,y))) = \tS(\lca_S(\sigma(x),
    \sigma(y)))$. Together with \AX{S2}, this implies that
    $\tS(\lca_S(\sigma(x), \sigma(y))) = \tT(\lca_T(x,y))$ and,
    therefore, $xy\in E(\Gg(\scen))$. Hence, we have
    $\sQO(\scen)\subseteq \Gg(\scen)$.
  \end{proof}
  \begin{lemma}
    \label{lem:alg1-Gg=SQO}
    If $\scen$ is a scenario produced by Algorithm~\ref{alg:recognition} to
    explain the valid input $\graphs=(\Gu,\Gg,\Ga,\sigma)$, then
    $\Gg=\sQO(\scen)$.
  \end{lemma}
  \begin{proof}
    Obs.~\ref{obs:generic} implies that $\Gg\subseteq \sQO(\scen)$ for
    every scenario $\scen$ produced by
    Algorithm~\ref{alg:recognition}. Conversely, every scenario $\scen$
    produced by Algorithm~\ref{alg:recognition} with input
    $\graphs=(\Gu,\Gg,\Ga,\sigma)$ is relaxed
    (cf.\ Lemma~\ref{lem:algo-explains}) and satisfies, in particular, $\Gg
    = \Gg(\scen)$. Hence, we can apply Lemma~\ref{lem:Gg-subset-sQO} to
    conclude that $\sQO(\scen)\subseteq \Gg(\scen) = \Gg$.
  \end{proof}
  It is important to note, however, that there are scenarios for which
  $\Gg\subseteq\sQO(\scen)$ is not true. As an example, consider the
  scenario $\scen$ in Figure~\ref{fig:simple-scenarios}(top row, middle) in
  which $xy\in \Gg(\scen)$ but $\mu(\lca_T(x,y))\neq
  \lca_S(\sigma(x),\sigma(y))$ and thus, $xy\notin \sQO(\scen)$.

\paragraph{Generic Scenarios.} 
It will sometimes be useful to assume that time maps are generic in the
sense that two inner vertices of the gene or species tree have the same
time stamp only if they belong to the same biological event. For our
purposes, it seems sufficient to rule out that concurrent nodes are mapped
to different positions in the species tree, i.e., we postulate the
following ``genericity'' axiom for evolutionary scenarios:
\begin{itemize}
\item[\AX{G}] If $\tT(v)=\tS(U)$ for $v\in V^0(T)$ and $U\in V^0(S)$,
  then $\mu(v)=U$.
\end{itemize}
Axiom \AX{G} stipulates that no two distinct speciation \emph{events},
i.e., inner nodes of the species tree are concurrent and that no other
evolutionary event (duplication or horizontal transfer) happens concurrent
with a speciation. Note that two vertices of the gene tree ``belong''
to the same speciation event if they are reconciled with the same vertex of
$S$. Thus $u,u'\in V(T)$ with $\mu(u)=\mu(u')\in V(S)$ are considered as
the same speciation event and thus also necessarily have the same time
stamp $\tT(u)=\tT(u')$.

As an immediate consequence of \AX{G}, we observe that
$\tT(\lca_T(x,y))=\tS(U)$ implies $\mu(\lca_T(x,y))=U$. Conversely, since
$T$ is phylogenetic, every $v\in V^0(T)$ (except the planted root) is the
last common ancestor of some pair of vertices, and $\mu(0_T)=0_S$, we can
equivalently express \AX{G} as
\begin{itemize}
\item[\AX{G'}] If $\tT(\lca_T(x,y))=\tS(U)$ for $x,y\in L(T)$ and $U\in
  V^0(S)$, then $\mu(\lca_T(x,y))=U$.
\end{itemize}

\begin{definition}
  A relaxed scenario satisfying \AX{G}, or equivalently \AX{G'}, is called
  \emph{generic}.
\end{definition}

We note in passing that it is not a trivial endeavor to modify a
  relaxed scenario $\scen$ to a generic one $\scen'$ such that
  $\graphs(\scen) = \graphs(\scen')$.  Simply adjusting the time maps is,
  in general, not enough. For example, consider scenario
  $\scen_3=(T,S,\sigma,\mu,\tT,\tS)$ in
  Figure~\ref{fig:HGT-free-Gg-quasi-ortho}(C).  Without adjusting the
  reconciliation map $\mu$, any generic scenario
  $\scen'=(T,S,\sigma,\mu,\tT',\tS')$ would satisfy $ab\notin
  E(\Gg(\scen'))$ although $ab\in E(\Gg(\scen_3))$. Hence, additional
  effort is needed to adjust $\mu$, i.e., to map $\lca_T(a,b)$ to
  $\lca_S(\sigma(a),\sigma(b))$ instead of mapping it to the edge
  $\rho_S\sigma(c)$. However, for every scenario $\scen$, there exists a
  (possibly alternative) scenario $\scen'$ that is computed using
  $\graphs(\scen)$ as input for Algorithm~\ref{alg:recognition} in
  conjunction with Algorithm~\ref{alg:restricted-scenario}.  Therefore,
  $\scen'$ satisfies $\graphs(\scen') = \graphs(\scen)$ and the conditions
  provided in Observation~\ref{obs:generic} and
  \ref{obs:mod-algo-explains}. These strong constraints on $\scen'$ might
  be helpful in transforming it into a generic scenario.

\begin{theorem}
  \label{thm:Gg-iff-quasi-ortho}
  For a generic scenario $\scen=(T,S,\sigma,\mu,\tT,\tS)$ it always holds
  that $\Gg(\scen) = \sQO(\scen)$ and thus, $\Gg(\scen)\subseteq
  \wQO(\scen)$. In particular, if $\scen$ is HGT-free or $S$ and $T$ are
  binary, then $\sQO(\scen)$ is a cograph.
\end{theorem}
\begin{proof}
  Let $\scen=(T,S,\sigma,\mu,\tT,\tS)$ be a generic scenario. 
  Assume first that $xy\in E(\Gg(\scen))$. By definition, $x\neq y$ and 
  $\tT(\lca_T(x,y))=\tS(\lca_S(\sigma(x),\sigma(y)))$. 
  By \AX{G'}, $\mu(\lca_T(x,y))=\lca_S(\sigma(x),\sigma(y))$. 
  Hence, $xy\in E(\sQO(\scen))$ and, therefore, 
  $\Gg(\scen)\subseteq \sQO(\scen)$. 
  By Lemma \ref{lem:Gg-subset-sQO}, we have
    $\sQO(\scen)\subseteq\Gg(\scen)$ and, thus, $\sQO(\scen) =
    \Gg(\scen)$. By Equ.~\eqref{eq:ortho-subgraph-rel}, we have
  $\Gg(\scen) = \sQO(\scen)\subseteq \wQO(\scen)$.  Moreover, $\sQO(\scen)
  = \Gg(\scen)$ together with Lemma~\ref{lem:HGT-free-Gg-cograph} and
  Theorem~\ref{thm:binary->Gg-cograph} implies that $\sQO(\scen)$ is a
  cograph whenever $\scen$ is HGT-free or $S$ and $T$ are binary.
\end{proof}

Note that a pair of weak quasi-orthologs $x,y\in L(T)$ may have arisen in a
speciation and have been transferred to the species $\sigma(x)$ and
$\sigma(y)$ in which they are found at later points in time. Thus
$\tT(\lca_T(x,y)) \lessgtr \tS(\lca_S(\sigma(x),\sigma(y))$ is possible,
see Figure~\ref{fig:weak-quasi-ortholog} for two examples.  Consequently,
$\wQO(\scen) \neq \Gg(\scen)$ is possible for generic scenarios.

\begin{figure}[t]
  \centering
  \includegraphics[width=0.45\textwidth]{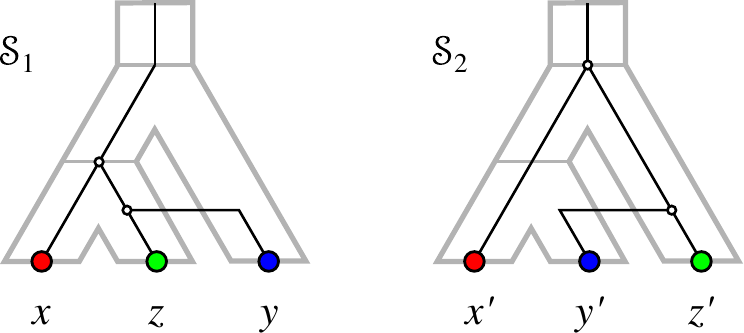}
  \caption{The two pairs $x$ and $y$ as well as $x'$ and $y'$ are weak
    quasi-orthologs in $\scen_1=(T,S,\sigma,\mu,\tT,\tS)$ and
    $\scen_2=(T',S',\sigma',\mu',\tau_{T'},\tau_{S'})$, respectively, but
    it holds $\tT(\lca_T(x,y)) < \tS(\lca_S(\sigma(x),\sigma(y))$ and
    $\tau_{T'}(\lca_{T'}(x',y')) >
    \tau_{S'}(\lca_{S'}(\sigma'(x'),\sigma'(y'))$.}
  \label{fig:weak-quasi-ortholog}
\end{figure}

As an immediate consequence of Lemma~\ref{lem:order}, equality between
$\wQO(\scen)$ and $\Gg(\scen)$ also holds for HGT-free scenarios. In
particular, by definition, $\sO(\scen) = \sQO(\scen)$. Hence, together with
Lemma~\ref{lem:HGT-free-Gg-cograph}, we obtain
\begin{corollary}
  \label{cor:HGT-free-Gg-equals-Theta}
  Every relaxed scenario $\scen$ without HGT-edges satisfies
  $\Gg(\scen)=\sQO(\scen)=\sO(\scen)$.  In this case, $\sQO(\scen)$ is a
  cograph.
\end{corollary}

\begin{corollary}
  Let $\scen$ be a generic scenario. Then $\Gg(\scen)=\wQO(\scen)$ if and
  only if $\mu(\lca_T(x,y))=\lca_S(\sigma(x),\sigma(y))$ for all $xy\in
  E(\wQO(\scen))$, which holds if and only if $\sQO(\scen)=\wQO(\scen)$. In
  this case, $\sQO(\scen)$ is a cograph.
  \label{cor:sameLCA}
\end{corollary}
The example in Figure~\ref{fig:HGT-free-Gg-quasi-ortho}(C) show that the 
condition \AX{G} cannot be dropped in Cor.~\ref{cor:sameLCA}.

\begin{figure}[ht]
  \centering
  \includegraphics[width=0.85\textwidth]{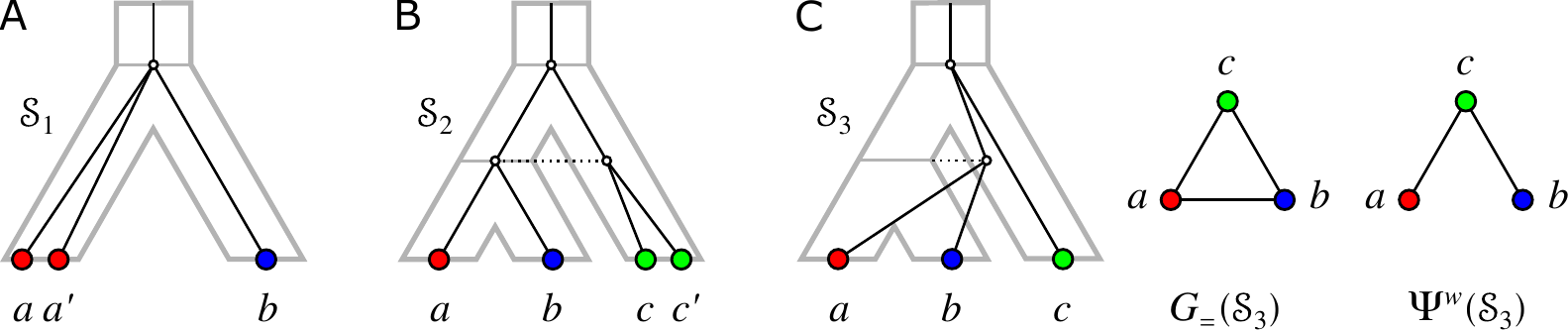}
  \caption{(A) A HGT-free relaxed scenario where $\Gg(\scen_1)\subsetneq
    \wQO(\scen_1)$.  The vertices $a$ and $a'$ are weak quasi-orthologs but
    $aa'\notin E(\Gg(\scen_1))$.  (B) An HGT-free, non-generic
      relaxed scenario. (C) A non-generic relaxed scenario for which
    $\Gg(\scen_3)\ne \wQO(\scen_3)$ even though
    $\mu(\lca_T(x,y))=\lca_S(\sigma(x),\sigma(y))$ holds for all
    $xy\in E(\wQO(\scen_3))$.}
  \label{fig:HGT-free-Gg-quasi-ortho}
\end{figure}

Equ.\ \ref{eq:ortho-subgraph-rel} and Thm.~\ref{thm:Gg-iff-quasi-ortho}
immediately imply
\begin{corollary}
  Every generic scenario $\scen$ satisfies
  $\sO(\scen) \subseteq \sQO(\scen) = \Gg(\scen)\subseteq \wQO(\scen)$.
\end{corollary}

\section{Concluding Remarks}

We have developed a complete characterization of graph $3$-partition
$\graphs$ on a species-colored set of vertices that can be explained by an
relaxed scenario $\scen$ (Thm.~\ref{thm:explainable-charac}). We showed,
furthermore, that whenever such an explaining relaxed scenario exists, one
can also find explanations from a much more restricted class of scenarios
that are fully witnessed and satisfy certain natural constraints for
``speciation events'' (Thm.~\ref{thm:relaxed-iff-restricted}). The
existence of such scenarios can be tested in polynomial time, and in the
positive case, both relaxed and restricted scenarios explaining the input
$3$-partition can be constructed, again in polynomial time.  If only the
information of $\Gg \in \graphs$ is available, it can be tested in
polynomial-time as whether $\Gg$ is an EDT graph in the HGT-free case
(cf.\ Thm.~\ref{thm:EDT-HGT-free}), while the problem becomes NP-hard for
general relaxed scenarios (cf.\ Thm.~\ref{thm:EDT-rec-NPc}).  In contrast,
PDT graphs can be recognized in polynomial-time
(cf.\ Thm.~\ref{thm:char-PDT}).  These approaches extend earlier work on
LDT graphs, which serve as the basis for indirect methods for the inference
of HGT events \cite{Schaller:21f}.  If only the information of $\Gu \in
\graphs$ is available, it can be tested whether $\Gu$ is an LDT graph and,
in the affirmative case, a relaxed scenario that explains $\Gu$ can be
constructed in polynomial-time \cite{Schaller:21f}.

Relaxed scenarios also can be used to formalize Walter Fitch's concept of
xenologous gene pairs \cite{Fitch:00,Darby:17}.  Given a relaxed scenario
$\scen=(T,S,\sigma,\mu,\tT,\tS)$, we define the \emph{xenology relation}
$R$ by setting $(x,y)\in R$ precisely if $x,y\in L(T)$ and the unique path
connecting $x$ and $y$ in $T$ contains an HGT edge. The resulting graph
$\gfitch(\scen) \coloneqq (L(T), R)$ is known as \emph{symmetrized Fitch
graph} \cite{Geiss:18a,Hellmuth:2019a,Hellmuth:18a}. It is always a
properly colored multipartite graph. Thm.~5 in \cite{Schaller:21f} shows
that for every properly colored multipartite graph there is a relaxed
scenario $\scen$ such that $\Gu(\scen) = \gfitch(\scen)$. On the other
hand, by \cite[Thm.~4]{Schaller:21f}, the LDT graph $\Gu(\scen)$ is always
a subgraph of $\gfitch(\scen)$ for every relaxed scenario $\scen$. Thus,
for every $\scen$ and every $xy\in \Gu(\scen)$, the two genes $x$ and $y$
are separated by at least on HGT-event. There are examples of relaxed
scenarios $\scen$ for which $\Gu(\scen)\neq \gfitch(\scen)$
(cf.\ \cite[Fig.~7]{Schaller:21f}).  Whether $\Gu(\scen) \subsetneq
\gfitch(\scen)$ or $\Gu(\scen) = \gfitch(\scen)$ heavily depends on the
particular scenario $\scen$. Given $\graphs = (\Gu,\Gg,\Ga,\sigma)$, which
may be estimated empirically from sequence similarity data, an explaining
scenario $\scen$ is not uniquely determined in general. This begs the
question whether there is a relaxed scenario $\scen$ that explains
$\graphs$ and satisfies $\Gu=\gfitch(\scen)$. To see that this is not the
case, consider $\graphs\coloneqq \graphs(\scen_2)$, where $\scen_2$ is the
scenario as in Fig.\ \ref{fig:EDT-P4-3col}. In this case, $\Gu$ is not a
complete multipartite graph and thus, $\Gu \subsetneq \gfitch(\scen)$ for
every relaxed scenario $\scen$ that explains $\Gu$. Consequently, the
information on HGT-events is not always provided entirely by the knowledge
of $\Gu$ alone. The graphs $\Gg$ and $\Ga$ thus may add additional
information for the inference of HGT. It will therefore be an interesting
topic for future work to understand how to employ $\graphs
=(\Gu,\Gg,\Ga,\sigma)$ to detect HGT-events and to which extend HGT-events
are uniquely determined for a given $\graphs$.

Relaxed scenarios provide a very general framework in which the concepts of
orthology, paralogy, and xenology can be studied in a rigorous manner. In
Section ``Orthology and Quasi-Orthology'', 
 we compared different concepts of orthology
that have been proposed for situations with horizontal transfer. We obtained
simple results describing the mutual relationships of the corresponding
variants of ``orthology graphs'' on $L(T)$, and their relations with
$\Gg$. With the exception of the strict quasi-orthology graph
$\wQO(\scen)$, the alternative notions lead to colored cographs similar to
the HGT-free case, see \cite{Hellmuth:13a}.  The latter connections are of
practical importance since the EDT graph $\Gg$, or the 3-partition graphs,
can be estimated from sequence similarities. It will be interesting,
therefore, to explore if techniques similar to those employed by
Schaller et al.\ \cite{Schaller:21a} can be used to identify the edges on $\Gg$ that do not
correspond to orthology-relationships.

We found that, similar to LDT graphs, PDT graphs are also cographs. This is
in general not the case for EDT graphs, although EDT graphs are perfect
(Prop.~\ref{prop:EDT-perfect-graph}). If both gene tree and species tree
are binary, i.e., fully resolved, then the EDT graph is a cograph.
However, not all proper vertex colorings of a cograph result in an EDT
graph (Fig.~\ref{fig:no-EDT}). It remains an interesting open problem to
characterize the ``EDT-colorings'' of cographs in analogy to the
hc-colorings of cograph that appear in the context of reciprocal best match
graphs \cite{Geiss:20a,Valdivia:23a}. Moreover, it is at least of
theoretical interest to ask how difficult it is to decide whether a
suitable coloring $\sigma$ exists such that $(\Gu,\Gg,\Ga,\sigma)$ is
explained by a relaxed scenario. Finding such a coloring corresponds to
assigning species to genes, a problem that arises in metagenomics.  Indeed,
when DNA is extracted from bulk samples taken from the environment, the
species that contains each sequence is unknown since they belong to members
of a diverse population (for instance, microbial or fungal).  Popular
techniques to recover a species assignment include sequence similarity
analysis~\cite{teeling2004tetra} and phylogenetic
reconstructions~\cite{darling2014phylosift}.  Since our approaches combine
these two ideas, it will be interesting to see whether EDT-colorings can be
useful in the context of metagenomics.

The reconciliation of $T$ and $S$ implicitly determines what kind of
evolutionary event corresponds to a vertex $v\in V^0(T)$. Given a relaxed
or restricted scenario $\scen$, the assignment of an event label $t(v)\in
Q$ from some pre-defined set $Q$ of event types is, of course, a matter of
biological interpretation of $\scen$. The definitions of ``DTL scenarios''
as in \cite{Tofigh:11,Bansal:12,Stolzer:12} assign event labels
to the inner vertices of $T$ that then must satisfy certain consistency
conditions with the local behavior of the reconciliation map
$\mu$. Event labelings $t:V^0(T)\to Q$ also play a key role in orthology
detection in duplication/loss scenarios
\cite{HernandezRosales:12a,Lafond:14,Hellmuth:17,Schaller:21a}. In relaxed
scenarios, it is not always possible to assign event types that match with
straightforward biological interpretations in an unambiguous manner. For
example, from a biological perspective, \emph{speciation events} are
usually defined as ``passing on the entire ancestral genome to each
offspring lineage''. In Fig.~\ref{fig:2-colP4}, however, $\lca_T(a,a')$
describes a gene duplication that occurs together with the speciation
event. As noted in \cite[Fig.2]{Stadler:20a}, this issue already arises in
the setting of DL-scenarios with multifurcating trees even in HGT-free
scenarios that satisfy the speciation constraint \AX{S6}, see also
\cite{Geiss:20b}.  Some further pertinent results on event-based
reconciliation in the presence of HGT were discussed by
N{\o}jgaard et al.\ \cite{Nojgaard:18a}. These point out subtle differences for non-binary
species trees in the definition of event-based DTL-scenarios \cite{Tofigh:11}
and suggest a natural notion of event-annotated relaxed scenarios. Because
of these difficulties we have avoided to consider event types as a formal
level in this contribution. Instead, these issues will be the focus of a
forthcoming contribution.

It is reassuring that a graph $3$-partition $\graphs$ that can be explained
by a relaxed scenario can always also be explained by a restricted
scenario. This begs the question, however, whether there is a simple, local
editing algorithm that converts a ``true'' scenario in a restricted or at
least a fully witnessed one. In the case of HGT-free scenarios, there is a
simple rule to exclude ``non-observable'' vertices in $T$: in this
restricted setting, it suffices to recursively remove all deleted genes and
all inner vertices with a single child \cite{HernandezRosales:12a}. The
situation seems to be much less obvious for relaxed scenarios, since these
models are somewhat more general than ``event-driven'' scenarios.  For
instance, relaxed scenarios allow multiple descendants from nodes $v\in
V(T)$ with $\mu(v)\in V(S)$. As a consequence, is seems difficult to
interpret a vertex $v$ that is reconciled with a vertex in the species tree
as a ``speciation event'' in the strict sense. The exact meaning of
``events'', therefore, deserves a more detailed analysis in the setting of
relaxed scenarios.

\begin{appendix}
\section*{Appendix}


\subsection*{Proof of Lemma~\ref{lem:algo-rs}}
\label{appx:additional-proofs}

In this section, we show in detail that, given a valid input
$\graphs=(\Gu,\Gg,\Ga,\sigma)$ with vertex set $L$,
Algorithm~\ref{alg:recognition} indeed returns a relaxed scenario
$\scen=(T,S,\sigma,\mu,\tT,\tS)$ such that $L(T)=L$. The proof parallels
the arguments in the proof of Thm.~2 in \cite{Schaller:21f}.

\begin{proof}[Proof of Lemma~\ref{lem:algo-rs}]
  Let $\sigma\colon L\to M$ and set $\R=\R_S(\graphs)$ and
  $\F=\F_S(\graphs)$. By a slight abuse of notation, we will simply write
  $\mu$ and $\tT$ also for restrictions to subsets of $V(T)$.  By
  assumption, $(\R,\F)$ is consistent, and thus, a tree $S$ on $M$ that
  displays $\Sri$ exists, and can be constructed in Line \ref{line:S}
  e.g.\ using \texttt{MTT} \cite{He:06}.  By Lemma~\ref{lem:arbitrary-tT},
  we can always construct a time map $\tS$ for $S$ satisfying $\tS(x)=0$
  for all $x\in L(S)$ in Line~\ref{line:tS}.  By definition,
  $\tS(y)>\tS(x)$ must hold for every edge $yx\in E(S)$, and thus, we
  obtain $\epsilon>0$ in Line~\ref{line:epsilon}.

  Recall that $\sigma(L')\subseteq L(S(u_S))$ holds in every recursion step
  by Obs.~\ref{obs:color-subset} and note that we reach the
  \emph{else}-block starting in Line~\ref{line:aux-graphs} only if $u_S$ is
  not a leaf.  Therefore, the auxiliary graphs $H_1$, $H_2$, and $H_3$ are
  well-defined and there is a vertex $v^*_S\in\child_S(u_S)$ such that
  $\sigma(C_j)\cap L(S(v^*_S))\ne\emptyset$ for every connected component
  $C_j$ of $H_2$ in Line~\ref{line:choose-v-S}, and a vertex $v_S\in
  \child_S(u_{S})$ such that $\sigma(C_k)\subseteq L(S(v_S))$ for every
  connected component $C_k$ of $H_3$ in
  Line~\ref{line:choose-v-S-for-class}.  Moreover, $\parent_S(u_{S})$ is
  always defined since we have $u_S=\rho_S$ and thus $\parent_S(u_S)=0_S$
  in the top-level recursion step, and recursively call the function
  \texttt{BuildGeneTree} on vertices $v_S$ such that $v_S\prec_S u_S$.

  In summary, all assignments are well-defined in every recursion step.  It
  is easy to verify that the algorithm terminates since, in each recursion
  step, we either have that $u_S$ is a leaf, or we recurse on vertices
  $v_{S}$ that lie strictly below $u_S$.  We argue that the resulting tree
  $T'$ is a (not necessarily phylogenetic) tree on $L$ by observing that,
  in each step, each $x\in L'$ is either attached to the tree as a leaf (if
  $u_S$ is a leaf) or passed down to a recursion step on some connected
  component of $H_3$ since each connected component $C_k$ of $H_3$
  satisfies $C_k\subseteq C_j$ for some connected component $C_j$ of $H_2$
  which in turn satisfies $C_j\subseteq C_i$ for some connected component
  $C_i$ of $H_1$.  Nevertheless, $T'$ is turned into a phylogenetic tree
  $T$ by suppression of degree-two vertices in Line~\ref{line:Tphylo}.
  Finally, $\mu(x)$ and $\tT(x)$ are assigned for all vertices $x\in
  L(T')=L$ in Line~\ref{line:mu-tT-leaves}, and for all newly created inner
  vertices in Lines~\ref{line:mu-tT-inner1}, \ref{line:mu-tT-inner2},
  and~\ref{line:mu-tT-inner3}.

  Before we continue to show that $\scen$ is a relaxed scenario, we first
  show that the conditions for time maps and time consistency are satisfied
  for $(T', S, \sigma, \mu, \tT, \tS)$:
  \par\noindent
  \begin{claim}
    \label{clm:tT-mu-in-T-prime}
    For all $x,y \in V(T')$ with $x\prec_{T'} y$, we have $\tT(x)<\tT(y)$.
    Moreover, for all $x\in V(T')$, the following statements are true:
    \vspace{-0.02in}
    \begin{description}
      \item[(i)] if $\mu(x)\in V(S)$, then $\tT(x)=\tS(\mu(x))$, and
      \item[(ii)] if $\mu(x)=(a,b)\in E(S)$, then $\tS(b)<\tT(x)<\tS(a)$.
    \end{description}
  \end{claim}
  \begin{claim-proof}
    Recall that we always write an edge $uv$ of a tree $T$ such that
    $v\prec_T u$.  For the first part of the statement, it suffices to show
    that $\tT(x)<\tT(y)$ holds for every edge $yx\in E(T')$, and thus to
    consider all vertices $x\neq \rho_{T'}$ in $T'$ and their unique
    parent, which will be denoted by $y$ in the following.  Likewise, we
    have to consider all vertices $x\in V(T')$ including the root to show
    the second statement.  The root $\rho_{T'}$ of $T'$ corresponds to the
    vertex $\rho'$ created in Line~\ref{line:create-rho} in the top-level
    recursion step on $L$ and $\rho_{S}$.  Hence, we have $\mu(\rho_{T'}) =
    \parent_S(\rho_S)\rho_S = 0_S\rho_S \in E(S)$ and
    $\tT(\rho_{T'})=\tS(\rho_S) +\epsilon$
    (cf.\ Line~\ref{line:mu-tT-inner1}).  Therefore, we have to show
    Subcase~(ii).  Since $\epsilon>0$, it holds that
    $\tS(\rho_S)<\tT(\rho_{T'})$.  Moreover, $\tS(0_S)-\tS(\rho_{S})\ge
    3\epsilon$ holds by construction, and thus
    $\tS(0_S)-(\tT(\rho_{T'})-\epsilon)\ge 3\epsilon$ and
    $\tS(0_S)-\tT(\rho_{T'})\ge 2\epsilon$, which together with
    $\epsilon>0$ implies $\tT(\rho_{T'})<\tS(0_S)$.

    We now consider the remaining vertices
    $x\in V(T')\setminus\{\rho_{T'}\}$. Every such vertex $x$ is introduced
    into $T'$ in some recursion step on $L'$ and $u_S$ in
    exactly one of the following four ways:
    \begin{enumerate}[noitemsep,nolistsep]
      \item[(a)] $x\in L(T')$ is a leaf attached to some inner
      vertex $\rho'$ in Line~\ref{line:attach-leaf},
      \item[(b)] $x=u_i$ is created in Line~\ref{line:create-u_i},
      \item[(c)] $x=v_j$ is created in Line~\ref{line:create-v_j}, and
      \item[(d)] $x=w_k\coloneqq \texttt{BuildGeneTree}(C_k,v_S)$ is
        attached to the tree in Line~\ref{line:recursive-call}.
    \end{enumerate}
    Note that if $x=\rho'$ is created in Line~\ref{line:create-rho}, then
    $\rho'$ is either the root of $T'$, or equals a vertex $w_k\coloneqq
    \texttt{BuildGeneTree}(C_k,v_S)$ that is attached to the tree in
    Line~\ref{line:recursive-call} in the ``parental'' recursion step.

    In Case~(a), we have that $x\in L(T')$ is a leaf and attached to some
    inner vertex $y=\rho'$. Since $u_S$ must be a leaf in this case, and
    thus $\tS(u_S)=0$, we have $\tT(y)=0+\epsilon=\epsilon$ and $\tT(x)=0$
    (cf.\ Lines~\ref{line:mu-tT-inner1} and~\ref{line:mu-tT-leaves}).
    Since $\epsilon>0$, this implies $\tT(x)<\tT(y)$.  Moreover, we have
    $\mu(x)=\sigma(x)\in L(S)\subset V(S)$
    (cf.\ Line~\ref{line:mu-tT-leaves}), and thus have to show Subcase~(i).
    Since $u_S$ is a leaf and $\sigma(L')\subseteq L(S(u_S))$, we conclude
    $\sigma(x)=u_S$.  Thus we obtain $\tT(x)=0=\tS(u_S)=\tS(\mu(x))$.

    In Case~(b), we have that $x=u_i$ is created in
    Line~\ref{line:create-u_i} and attached as a child to some vertex
    $y=\rho'$ created in the same recursion step.  Thus, we have
    $\tT(y)=\tS(u_S)+\epsilon$, $\tT(x)=\tS(u_S)$ and $\mu(x)=u_S\in V(S)$
    (cf.\ Lines~\ref{line:mu-tT-inner1} and~\ref{line:mu-tT-inner2}).
    Therefore and because$\epsilon>0$, it holds $\tT(x)<\tT(y)$ and
    Subcase~(i) is satisfied.

    In Case~(c), we have that $x=v_j$ is created in
    Line~\ref{line:create-v_j} and attached as a child to some vertex
    $y=u_i$ created in the same recursion step.  Thus, we have
    $\tT(y)=\tS(u_S)$ and $\tT(x)=\tS(u_S)-\epsilon$
    (cf.\ Lines~\ref{line:mu-tT-inner2} and~\ref{line:mu-tT-inner3}).
    Therefore and since $\epsilon>0$, it holds $\tT(x)<\tT(y)$.  Moreover,
    we have $\mu(x)=u_S v^*_S \in E(S)$ for some $v^*_S\in\child_S(u_S)$.
    Hence, we have to show Subcase~(ii). By a similar calculation as
    before, $\epsilon>0$, $\tS(u_S)-\tS(v^*_S)\ge 3\epsilon$ and
    $\tT(x)=\tS(u_S)-\epsilon$ imply $\tS(v^*_S)<\tT(x)<\tS(u_S)$.

    In Case~(d), $x=w_k\coloneqq \texttt{BuildGeneTree}(C_k,v_S)$ is
    attached to the tree in Line~\ref{line:recursive-call} and equals
    $\rho'$ as created in Line~\ref{line:create-rho} in some ``child''
    recursion step with $v_S\in\child_S(u_S)$. Thus, we have
    $\tT(x)=\tS(v_S)+\epsilon$ and $\mu(x)=u_S v_S \in E(S)$
    (cf.\ Line~\ref{line:mu-tT-inner1}). Moreover, $x$ is attached as a
    child of some vertex $y=v_j$ as created in
    Line~\ref{line:create-v_j}. Thus, we have $\tT(y)=\tS(u_S)-\epsilon$.
    By construction and since $u_S v_S \in E(S)$, we have
    $\tS(u_S)-\tS(v_S)\ge 3\epsilon$.  Therefore, $(\tT(y)+\epsilon) -
    (\tT(x)-\epsilon) \ge 3\epsilon$ and thus $\tT(y)- \tT(x) \ge
    \epsilon$. This together with $\epsilon>0$ implies $\tT(x)<\tT(y)$.
    Moreover, since $\mu(x)=u_S v_S \in E(S)$ for some
    $v_S\in\child_S(u_S)$, we have to show Subcase~(ii).  By a similar
    calculation as before, $\epsilon>0$, $\tS(u_S)-\tS(v_S)\ge 3\epsilon$
    and $\tT(x)=\tS(v_S)+\epsilon$ imply $\tS(v_S)<\tT(x)<\tS(u_S)$.
  \end{claim-proof}

  The tree $T$ is obtained from $T'$ by first adding a planted root $0_T$
  (and connecting it to the original root) and then suppressing all inner
  vertices except $0_T$ that have only a single child in Line
  \ref{line:Tphylo}.  In particular, $T$ is a planted phylogenetic tree by
  construction.  The root constraint (S0) $\mu(x)=0_S$ if and only if
  $x=0_T$ also holds by construction
  (cf.\ Line~\ref{line:mu-tT-planted-root}).  Since we clearly have not
  contracted any outer edges $(y,x)$, i.e.\ with $x\in L(T')$, we conclude
  that $L(T')=L(T)=L$.  As argued before, we have $\tT(x)=0$ and
  $\mu(x)=\sigma(x)$ whenever $x\in L(T')=L(T)$
  (cf.\ Line~\ref{line:mu-tT-leaves}).  Since, in addition, all other
  vertices are mapped by $\mu$ to some edge of $S$, inner vertex, or $0_S$
  (cf.\ Lines~\ref{line:mu-tT-inner1}, \ref{line:mu-tT-inner2},
  \ref{line:mu-tT-inner3}, and \ref{line:mu-tT-planted-root}), the leaf
  constraint (S1) is satisfied.

  By construction, we have $V(T)\setminus \{0_T\} \subseteq V(T')$.
  Moreover, suppression of vertices clearly preserves the
  $\preceq$-relation between all vertices $x,y\in V(T)\setminus \{0_T\}$.
  Together with Claim~\ref{clm:tT-mu-in-T-prime}, this implies
  $\tT(x)<\tT(y)$ for all vertices $x,y\in V(T)\setminus \{0_T\}$ with
  $x\prec_{T} y$.  For the single child $\rho_T$ of $0_T$ in $T$, we have
  $\tT(\rho_T)\le \tS(\rho_S)+\epsilon$ where equality holds if the root of
  $T'$ was not suppressed and thus is equal to $\rho_T$.  Moreover,
  $\tT(0_T)=\tS(0_S)$ and $\tS(0_S)-\tS(\rho_S)\ge 3\epsilon$ hold by
  construction.  Taken together the latter two arguments imply that
  $\tT(\rho_T)<\tT(0_T)$.  In particular, we obtain $\tT(x)<\tT(y)$ for all
  vertices $x,y\in V(T)$ with $x\prec_{T} y$.  Hence, $\tT$ is a time map
  for $T$, which, moreover, satisfies $\tT(x)=0$ for all $x\in L(T)$.

  To show that $\scen=(T,S,\sigma,\mu, \tT,\tS)$ is a relaxed scenario, it
  remains to show the two time consistency constraints (S2) and (S3) in
  Def.~\ref{def:rs}.  For $0_T$, we have $\tT(0_T)=\tS(0_S)=\tS(\mu(0_T))$.
  Hence, condition in (S2) is satisfied for $0_T$.  The remaining vertices
  of $T$ are all vertices of $T'$ as well.  The latter two arguments
  together with Claim~\ref{clm:tT-mu-in-T-prime} imply that conditions (S2)
  and (S3) are also satisfied, and thus $\scen$ is a relaxed scenario.

\end{proof}

\subsection*{Hardness of  EDT Graph Recognition}
\label{sec-appx:hardness}

To establish the NP-hardness of \textsc{$(\CF,\CR)$-Satisfiability} and
\textsc{EDT-Recognition}, we start from
\begin{problem}[\PROBLEM{3-Set Splitting}]\ \\
  \begin{tabular}{ll}
    \emph{Input:}    & A finite set $U$ and a collection
                       $B=\{B_1,\dots,B_m\}$ of subsets of $U$ \\ 
                     & s.t.\ $\vert B_i\vert=3$ for all $i$.  \\
    \emph{Question:} & Is there a partition $\{U_1, U_2\}$ of $U$ into
                       two sets such that, for each $B_j \in B$,\\
		     & we have $B_j\cap U_1\neq \emptyset$ and
                       $B_j\cap U_2\neq \emptyset$. \\
		       & In other words, none of the $B_j \in B$ is
                       entirely contained in either $U_1$ or $U_2$. 
  \end{tabular}
\end{problem}

Lov{\'a}sz \cite{lovasz1973coverings} showed that the ``unrestricted'' version
  of \PROBLEM{3-Set Splitting}, in which elements in $B_j \in B$ have size
  $\vert B_i\vert \leq 3$ instead of $\vert B_i\vert= 3$, is
  NP-complete. There does not seem to be a published proof for the
  NP-completeness of the ``restricted'' variant of \PROBLEM{3-Set
    Splitting}. For completeness, we include a simple argument starting
  from
\begin{problem}[\PROBLEM{monotone NAE-3-SAT}]\ \\
  \begin{tabular}{ll}
    \emph{Input:} & Given a set of clauses $C = \{C_1, \dots, C_m\}$
                    over a set $U$ of Boolean variables \\
	          & s.t.\ $\vert C_i\vert =3$ for all $i$ 
		    and $C_i$ contains no negated variables. \\
    \emph{Question:}& Is there a truth assignment to $U$ such 
      that in each $C_i$\\
		  & not all three literals are set to true?
  \end{tabular}
\end{problem}

As shown by Porschen et al.\ \cite[Thm.~3]{PORSCHEN20141}, \PROBLEM{monotone NAE-3-SAT} is
NP-complete. Its is straightforward to see that \PROBLEM{monotone
  NAE-3-SAT} and \PROBLEM{3-Set Splitting} are equivalent in the following
sense: Interpret the $C_i\in C$ as sets and put $B=C$. Then $(C,U)$ is a
yes-instance of \PROBLEM{monotone NAE-3-SAT} if and only if $(B,U)$ is a
yes-instance of \PROBLEM{3-Set Splitting} because we can obtain a solution
$\{U_1, U_2\}$ for $(B,U)$ from a solution for $(C,U)$ by setting
$U_1\coloneqq \{x\in U \mid x \text{ is true}\}$ and $U_2\coloneqq
U\setminus U_1$. Conversely, a solution for $(C,U)$ is obtained from a
solution $\{U_1, U_2\}$ for $(B,U)$ by assigning ``true'' exactly to all
$x\in U_1$.  Consequently, we have
\begin{proposition}\label{prop:3-set-splitting}
\PROBLEM{3-Set Splitting} is NP-complete.
\end{proposition}

We are now in the position to prove NP-completeness of
\PROBLEM{$(\CF,\CR)$-Satisfiability} (Thm.~\ref{thm:CFCR-NPc}).

\begin{proof}[Proof of Theorem \ref{thm:CFCR-NPc}.]
  Given a tree $S^*$, it can be verified in polynomial-time as whether
  $S^*$ satisfies $(\CF, \CR)$. Hence,
  $\textsc{$(\CF,\CR)$-Satisfiability}\in \text{NP}$.  To show NP-hardness
  we use a reduction from \PROBLEM{3-Set Splitting}.

  Given an instance $(U, B)$ of \PROBLEM{3-Set Splitting}, construct an
  instance $(U', \CF, \CR)$ of \textsc{$(\CF, \CR)$-Satisfiability} as
  follows. For $B_j \in B$, we order its three elements arbitrarily and
  write $B_j = \{b_j^1, b_j^2, b_j^3\}$. Let
  $U' \coloneqq U \cup \{x, z', z''\} \cup \{\alpha_j : 1 \leq j \leq m\}$
  and let
  \begin{align*}
    \CF \coloneqq&\{x\vert z'\vert z''\} \cup
                         \bigcup_{j=1}^m \{\ x \vert  b_j^1 \vert
                         \alpha_j,\ b_j^2 \vert  b_j^3 \vert  \alpha_j\ \}, \\
    \CR \coloneqq& \{\ \rpair{u_i}{z'}{z''} : u_i \in U\ \}
  \end{align*}
  It is easy to verify that this reduction can be performed in polynomial
  time. We show that there exists a 3-set splitting of $B$ if and only if
  there exists a tree $S^*$ that satisfies $(\CF, \CR)$.

  \begin{figure}
    \centering \includegraphics[width=0.85\textwidth]{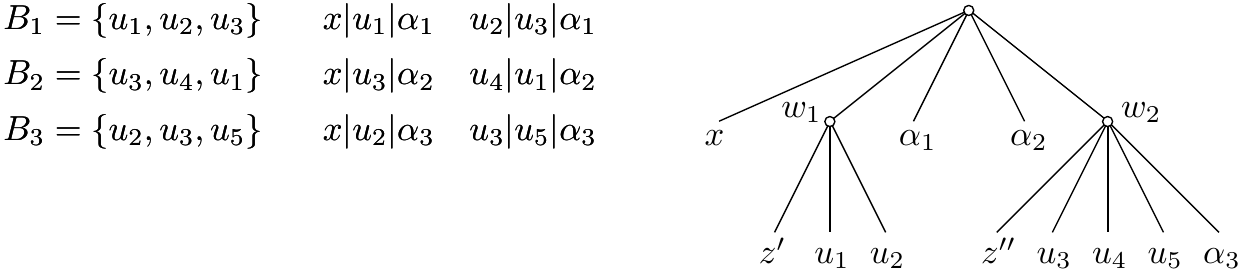}
    \caption{An example instance of \textsc{Set-Splitting} with
      solution $U_1 = \{u_1, u_2\}, U_2 = \{u_3, u_4, u_5\}$.  The elements
      of $B_1, B_2, B_3$ are in the order chosen by the reduction.  The
      constructed tree $S^*$ is shown, along with the
      $x\vert b_j^1\vert \alpha_j$ and $b_j^2\vert b_j^3\vert \alpha_j$ fan
      triples that must be displayed.  Note that each of the three cases in
      which two elements of $B_j$ have the same parent occurs. }
    \label{fig:setsplitting}
  \end{figure}

  Assume first that $(U,B)$ is a yes-instance of \PROBLEM{3-Set Splitting},
  i.e., there is a partition $\{U_1, U_2\}$ of $U$ such that
  $\vert B_j \cap U_1\vert \in \{1, 2\}$ for each $B_j \in B$.  We
  construct a tree $S^*$ that satisfies $(\CF, \CR)$, see
  Figure~\ref{fig:setsplitting} for an illustrative example.  Start
  with $S^*$ as the tree in which the root has three children
  $x, w_1, w_2$.  Then, add each element of $\{z'\} \cup U_1$ as a child of
  $w_1$, and add each element of $\{z''\} \cup U_2$ as a child of
  $w_2$. Notice that $S^*$ displays $x\vert z'\vert z''$ as required by
  $\CF$. Moreover, because each $u_i$ has either $z'$ or $z''$ as a sibling
  but not both, $S^*$ displays either $u_i z' \vert z''$ or
  $u_i z'' \vert z'$ for each $u_i \in U$, and thus satisfies the
  constraints in $\CR$. We next add the remaining $\alpha_j$ leaves as
  children of existing vertices of $S^*$, which cannot alter the triples
  and fan triples gathered so far.

  For each $B_j \in B$, exactly two of $b_j^1, b_j^2$ and $ b_j^3$ have the
  same parent $w \in \{w_1, w_2\}$ in $S^*$, because $\{U_1, U_2\}$ is a
  3-set splitting. There are three cases, and in each one, we let the
  reader verify that $S^*$ displays $x\vert b_j^1\vert \alpha_j$ and
  $b_j^2\vert b_j^3\vert \alpha_j$:
  \newline
  if either $b_j^1$ and $b_j^2$ or $b_j^1$ and $b_j^3$ have the same parent
  $w$, then add $\alpha_j$ as a child of the root of $S^*$;
  \newline
  if $b_j^2$ and $b_j^3$ have the same parent $w$, then add $\alpha_j$
  as a child of $w$.

  It is now straightforward to verify that $S^*$ satisfies $(\CF,\CR)$.

  Suppose now that $(U', \CF, \CR)$ is a yes-instance of
  \textsc{$(\CF,\CR)$-Satisfiability}, i.e., there
  exists a tree $S^*$ that satisfies $(\CF, \CR)$.  By the construction of
  $\CR$, for $u_i \in U$, $S^*$ displays either $u_i z' \vert z''$ or
  $u_i z'' \vert z'$.  We claim that the partition $\{U_1, U_2\}$ where
  \begin{align*}
    U_1 \coloneqq& \{u_i : S^* \mbox{ displays } u_i z' \vert z'' \} 
    \text{ and}\\
    U_2 \coloneqq& \{u_i : S^* \mbox{ displays } u_i z'' \vert z' \}
  \end{align*}
  is a 3-set splitting of $B$. In fact, $U_1\cap U_2=\emptyset$, since
  $S^*$ cannot display both $u_i z' \vert z''$ and $u_i z'' \vert z'$ at
  the same time. Moreover, by construction of $\CR$ and since $S^*$
  satisfies $(\CF, \CR)$, at least one of the triples $u_i z' \vert z''$
  and $u_i z'' \vert z'$ must be displayed by $S^*$ for all $u_i\in
  U$. Consequently, $U_1\cup U_2=U$.

  Assume, for contradiction, that $\{U_1, U_2\}$ is not a 3-set splitting
  of $B$. Hence, there is a $B_j = \{b_j^1, b_j^2, b_j^3\}$ in $B$ such
  that either $B_j \subseteq U_1$ or $B_j \subseteq U_2$.  First, suppose
  that $B_j \subseteq U_1$. By construction of $U_1$, $S^*$ displays $b_j^1
  z' \vert z''$, $b_j^2 z' \vert z''$, and $ b_j^3 z' \vert z''$. Since
  $S^*$ displays $x\vert z'\vert z'' \in \CF$, we have $r \coloneqq
  \lca_{S^*}(x, z') = \lca_{S^*}(x, z'') = \lca_{S^*}(z', z'')$.  Let $y'$
  be the unique child of $r$ such that $z'\preceq_{S^*} y'$, and note that
  $x$ and $z''$ are not descendants of $y'$. Since $S^*$ displays $b_j^1 z'
  \vert z''$, $b_j^2 z' \vert z''$, and $b_j^3 z' \vert z''$, it follows
  that $b_j^1$, $b_j^2$, and $b_j^3$ are all descendants of $y'$.  Now,
  $\alpha_j$ cannot be a descendant of $y'$, as otherwise $S^*$ would
  display $b_j^1 \alpha_j \vert x$, as opposed to the fan triple $x\vert
  b_j^1\vert \alpha_j \in \CF$ that $S^*$ must display. On the other hand,
  if $\alpha_j$ is not a descendant of $y'$, then $b^j_2,b^j_3\prec_{S^*}
  y'$ implies that $S^*$ displays $b_j^2 b_j^3 \vert \alpha_j$, a
  contradiction since $b_j^2 \vert b_j^3 \vert \alpha_j \in \CF$.  Hence,
  $B_j \subseteq U_1$ is not possible. By interchanging the roles of $z'$
  and $z''$ and using similar arguments, one shows that $B_j \subseteq U_2$
  is not possible either. In summary, $\{U_1, U_2\}$ is a 3-set splitting.
\end{proof}

We are now in the position to prove NP-completeness of
\PROBLEM{EDT-Recognition} (Thm.~\ref{thm:EDT-rec-NPc}).

\begin{proof}[Proof of Theorem \ref{thm:EDT-rec-NPc}]
  First note that the problem is in NP, since a scenario that explains a
  given instance $(G, \sigma)$ can easily be verified in polynomial time.
  We show that \textsc{EDT-Recognition} is NP-hard by reduction from the
  \textsc{$(\CF, \CR)$-Satisfiability} problem.  Let $(U, \CF, \CR)$ be an
  instance of \textsc{$(\CF, \CR)$-Satisfiability}. We proceed by
  constructing a corresponding instance $(G, \sigma)$ of
  \textsc{EDT-Recognition} as the disjoint union of colored graphs $F_t$
  for all $t\in\CF$ and $R_t$ for all $t\in \CR$.

  The color set $\Sigma$ comprises a distinct color $\sigma(u)$ for each
  $u\in U$, and a distinct color $\sigma(t)$ for each $t\in\CR$.  Note that
  for each pair of triples $t = \rpair{x}{y}{z} \in \CR$ a single color
  $\sigma(t)$ is used. Hence, $\Sigma$ contains $\vert U \vert + \vert
  \CR\vert$ colors.
  
  \medskip\noindent For each $t \coloneqq x\vert y\vert z \in \CF$, we
  define $F_t$ as the vertex colored graph with
  \begin{itemize}
  \item[] vertex set $V(F_t)\coloneqq \{x_t,y_t,z_t,x'_t,y'_t,z'_t\}$,
  \item[] edge set $E(F_t) \coloneqq \{x_ty_t,y_tz_t,x'_tz'_t,z'_ty'_t\}$,
    and
  \item[] vertex coloring $\sigma(x_t)=\sigma(x_t')=\sigma(x)$,
    $\sigma(y_t)=\sigma(y_t')=\sigma(y)$, and
    $\sigma(z_t)=\sigma(z_t')=\sigma(z)$.
  \end{itemize}
  By construction $F_t$ consists of two connected components, namely the
  two $P_3$s $x_t-y_t-z_t$ and $x'_t-z'_t-y'_t$ on three colors.  In
  particular, $F_t$ is properly colored. Moreover, $F_t$ and $F_{t'}$ are
  vertex disjoint for distinct $t,t'\in \CF$ even though $t$ and $t'$ may
  have leaves in common and thus, the vertices in $V(F_t)$ and $V(F_{t'})$
  may share colors.
    
 \medskip\noindent
  For each $t \coloneqq \rpair{x}{y}{z} \in \CR$ we define $R_t$ as the vertex
  colored graph with 
  \begin{itemize}
  \item[] vertex set $V(R_t) \coloneqq \{x_t, y_t, z_t, w_t, y_t', z_t',
    w_t'\}$,
  \item[] edge set $E(R_t) \coloneqq \{x_t w_t, x_t y_t, w_t y_t, w_t z_t,
    w_t' y_t', y_t' z_t'\}$, and
  \item[] vertex coloring 
    $\sigma(x_t) = \sigma(x)$,
    $\sigma(y_t) = \sigma(y_t') = \sigma(y)$, 
    $\sigma(z_t) = \sigma(z_t') = \sigma(z)$, and
    $\sigma(w_t) = \sigma(w_t') = \sigma(t)$. 
  \end{itemize}
  By construction, $R_t$ consists of two connected components, a so-called
  \emph{paw graph} on the four vertices $x_t$, $y_t$, $z_t$, and $w_t$ and
  the $P_3$ $w_t' - y_t' - z_t'$. In particular, $R_t$ is properly colored.
  Again, $R_t$ and $R_{t'}$ for distinct $t,t'\in\CR$ are vertex disjoint
  but may share certain colors.  Since $\CF \cap \CR = \emptyset$, we have
  $t\ne t'$ for any $F_t$ and $R_{t'}$, i.e., each $t$ unambiguously refer
  to either a subgraph $F_t$ or a subgraph $R_t$ of $(G,\sigma)$. The
  graphs $F_t$ and $R_t$ are illustrated in Fig.~\ref{fig:NPc-gadgets}(A)
  and (B), respectively.
  
  \begin{figure}[h]
    \centering 
    \includegraphics[width=0.7\textwidth]{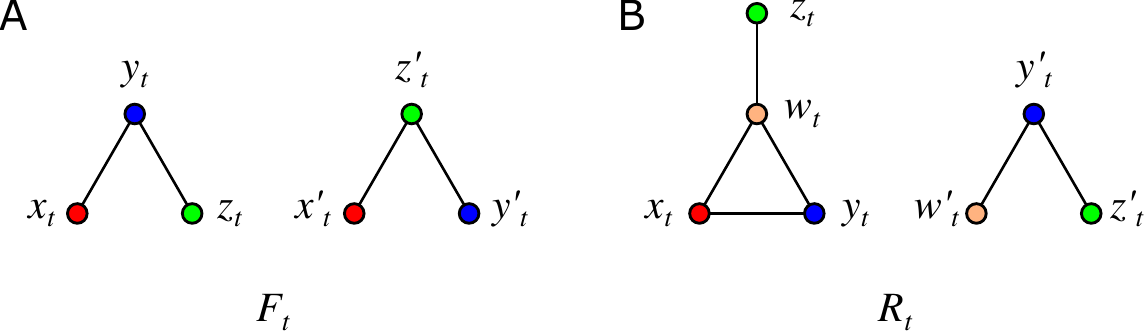}
    \caption{The graphs $F_t$ and $R_t$ as constructed in the proof of
      Theorem~\ref{thm:EDT-rec-NPc}.}
     \label{fig:NPc-gadgets}
  \end{figure}

  Since $F_t$ and $R_t$ can be constructed in constant time for each
  $t\in\CF\cup\CR$, the graph $(G,\sigma)$ can be constructed in polynomial
  time. Every connected component of $G$ is either a paw component'' or a
  ``$P_3$ component''. By construction, any two vertices that are in the
  same connected component of $(G,\sigma)$ have different colors. Thus
  $(G,\sigma)$ is properly colored.
  
  We proceed by showing that there exists a tree $S^*$ that satisfies
  $(\CF, \CR)$ if and only if there exists a relaxed scenario $\scen$ that
  explains $(G, \sigma)$.  As we shall see, $F_t$ ensures that the species
  tree $S^*$ displays the fan triple
  $\sigma(x)\vert\sigma(y)\vert\sigma(z)$, while $R_t$ enforces the species
  tree to display either $\sigma(x)\sigma(y)\vert\sigma(z)$ or
  $\sigma(x)\sigma(z)\vert\sigma(y)$.

  \begingroup
  \renewcommand{\a}{\tilde{a}}
  \renewcommand{\b}{\tilde{b}}
  \renewcommand{\c}{\tilde{c}}
  \newcommand{\p}{\tilde{p}}
  \newcommand{\q}{\tilde{q}}
  \renewcommand{\r}{\tilde{r}}
  \renewcommand{\u}{\tilde{u}}
  \renewcommand{\v}{\tilde{v}}
  \newcommand{\w}{\tilde{w}}  
  \newcommand{\x}{\tilde{x}}
  \newcommand{\y}{\tilde{y}}
  \newcommand{\z}{\tilde{z}}

  In the following we simplify the notation and denote the color of a
  vertex $u$ in $G$ by $\tilde u$ instead of $\sigma(u)$.

  Suppose first that $(G, \sigma)$ is a yes-instance of
  \textsc{EDT-Recognition} and thus, there exists a relaxed scenario $\scen
  = (T,S,\sigma,\mu,\tT,\tS)$ that explains $(G, \sigma)$. We show that
  there exists a tree $S^*$ that satisfies $(\CF,\CR)$. Consider $\graphs =
  (\Gu(\scen), \Gg(\scen), \Ga(\scen), \sigma)$, where by assumption
  $\Gg(\scen) = G$. By Prop.\ \ref{prop:triples}, the species tree $S$ of
  $\scen$ agrees with $(\R_S(\graphs), \F_S(\graphs))$.

  We claim that $S_{\vert \x\y\z}$ coincides with the fan triple $\x\vert
  \y\vert \z$ for every $t = x \vert y \vert z \in \CF$.  To see this,
  consider the subgraph $F_t$ in $G$. It contains $x_t-y_t-z_t$ and $x'_t -
  z'_t - y'_t$ as induced $P_3$s. By Definition~\ref{def:inf-forb-triples},
  therefore, $\x \y \vert \z$, $\x \z \vert \y$, and $\y \z\vert \x$ are
  forbidden triples of $\F_S(\graphs)$, and thus $S_{\vert \x\y\z}$ must
  display $\x\vert \y\vert \z$ as claimed.  We next claim that for each $t
  = \rpair{x}{y}{z} \in \CR$, $S_{\vert \x\y\z}$ is either $\x\y\vert \z$
  or $\x\z\vert \y$. Consider the subgraph $R_t$ in $G$.  It contains $w_t'
  - y_t' - z_t'$ as an induced $P_3$. By
  Definition~\ref{def:inf-forb-triples}, therefore, $\w \y \vert \z$ and
  $\y \z \vert \w$ are forbidden triples of $\F_S(\graphs)$. We argue next
  that $y_t z_t \in E(\Gu(\scen))$. To this end, suppose for contradiction
  that $y_t z_t \in E(\Ga(\scen))$. This together with
  Definition~\ref{def:inf-forb-triples} and $w_t y_t, w_tz_t \in
  E(G)=E(\Gg(\scen))$ implies that $\y \z \vert \w$ is an informative
  triple of $\R_S(\graphs)$; a contradiction to $\y \z \vert \w$ being a
  forbidden triple. Together with $y_t z_t \notin E(G)$, this leaves $y_t
  z_t \in E(\Gu(\scen))$ as the only possibility. Now consider $x_t z_t$,
  which is not an edge in $G=\Gg(\scen)$. We have the two possibilities
  $x_t z_t \in E(\Gu(\scen))$ and $x_t z_t \in E(\Ga(\scen))$. Again using
  Definition~\ref{def:inf-forb-triples}, $x_t z_t, y_t z_t \in
  E(\Gu(\scen))$ and $x_t y_t \notin E(\Gu(\scen))$ yield the informative
  triple $\x \y \vert \z$ in the former case; and $x_t z_t \in
  E(\Ga(\scen))$ and $x_t y_t, y_t z_t \notin E(\Ga(\scen))$ yield the
  informative triple $\x \z \vert \y$. Hence, in either case, $S_{\vert
    \x\y\z}$ is either $\x \y\vert \z$ or $\x \z \vert \y$, as claimed.

  We now construct a tree $S^*$ that satisfies $(\CF, \CR)$ from $S$ as
  follows.  We first set $S' \coloneqq S_{\vert \{\tilde{u} : u \in U \}}$.
  In other words, $S'$ is the minimal phylogenetic subtree of $S$ that
  connects all leaves that are distinct from $\w_t$ for $t \in \CR$.
  Moreover, since $w_t$ is not part of any of the aforementioned triples
  and fan triples, the tree $S'$ still displays, for every $t = x \vert y
  \vert z \in \CF$, the fan triple $\x\vert \y\vert \z$ and, for every $t =
  \rpair{x}{y}{z} \in \CR$, either the triple $\x\y\vert \z$ or the triple
  $\x\z\vert \y$. The tree $S^*$ obtained from $S'$ by relabeling, for each
  $u\in U$, the leaf $\u$ by $u$ therefore satisfies $(\CF, \CR)$.
  
  Suppose that $(U, \CF, \CR)$ is a yes-instance of
  \textsc{$(\CF,\CR)$-Satisfiability} and thus, there exists a tree $S^*$
  on leaf set $U$ that satisfies $(\CF, \CR)$. We first construct a graph
  3-partition $\graphs = (\Gu, \Gg, \Ga, \sigma)$ and then use
  Theorem~\ref{thm:explainable-charac} to argue that
  $\graphs$ can be explained by some relaxed scenario.

  We start by setting $\Gg\coloneqq G$ and proceed as follows:
  \begin{enumerate}[label=(A{{\arabic*}})]
  \item \label{item:rule-ccs} for any two distinct connected components
    $H_1$ and $H_2$ of $G$ and any $x \in H_1, y \in H_2$, add $xy$ to
    $E(\Ga)$;
  \item \label{item:rule-F} for each $t = x\vert y\vert z \in \CF$, 
    add $x_t z_t$ and $x_t' y_t'$ to $E(\Gu)$;
  \item \label{item:rule-R} for each $t = \rpair{x}{y}{z} \in \CR$, add
    $y_t z_t$ and $w_t' z_t'$ to $E(\Gu)$ and, for $x_t z_t$, there are two
    cases:
    \begin{enumerate}
    \item 
      if $S^*$ displays $xy\vert z$, then add $x_t z_t$ to $E(\Gu)$;
    \item 
      if $S^*$ displays $xz\vert y$, then add $x_t z_t$ to $E(\Ga)$.
    \end{enumerate}
    Note that no other case is possible since $S^*$ satisfies $(\CF,\CR)$.
  \end{enumerate}

  This completes the construction of $\graphs$. Since
  rules~\ref{item:rule-F} and \ref{item:rule-R} assign an edge in either
  $\Ga$ or $\Gu$ to every non-adjacent pair of vertices within the same
  connected component, i.e., induced $P_3$ or paw graph of $\Gg$, and
  rule~\ref{item:rule-ccs} covers all edges between these connected
  components, $\graphs$ is a graph 3-partition.

  \begin{claim}\label{claim:gsmall}
    For each $ab \in E(\Gu)$, $a$ and $b$ are in the same connected
    component of $G$. Moreover, the connected components of $\Gu$ are
    isolated edges or induced $P_3$s.
    
    \smallskip\normalfont
    \noindent
    \emph{Proof of Claim \ref{claim:gsmall}.}  Only Steps~\ref{item:rule-F}
    and~\ref{item:rule-R} add edges to $\Gu$, and they only add edges
    between vertices of the same $P_3$ or paw component of $G$.  Moreover,
    in each such component, these steps never add more than two edges to
    $\Gu$, and so the connected components of $\Gu$ are isolated edges or
    induced $P_3$s, as claimed.
    \hfill$\diamond$
  \end{claim}

  \begin{claim}\label{claim:propcolored}
    The graphs $\Gu$ and $\Gg$ are properly colored.  
    
    \smallskip\normalfont
    \noindent
    \emph{Proof of Claim \ref{claim:propcolored}.}  Because $\Gg = G$, the
    graph $\Gg$ is properly colored by construction. As for $\Gu$, the
    endpoints of $\Gu$ edges always belong to the same $P_3$ on three
    colors or $P_4$ on four colors in $G$ by Claim~\ref{claim:gsmall}, and
    they have a different color by construction.
    \hfill$\diamond$
  \end{claim}

  \begin{claim}\label{claim:cographs}
    The graphs $\Gu$ and $\Ga$ are cographs.
    
    \smallskip\normalfont
    \noindent
    \emph{Proof of Claim \ref{claim:cographs}.}  For $\Gu$, this holds
    because its connected components have at most $3$ vertices by
    Claim~\ref{claim:gsmall} and, thus, it cannot contain an induced $P_4$.
    Now consider the graph $\Gg\cup\Gu$. Since only Steps~\ref{item:rule-F}
    and~\ref{item:rule-R} add edges to $\Gu$, and they only add edges
    between vertices of the same $P_3$ or paw component of $G$, the
    connected components of $\Gg\cup \Gu$ all have $3$ or $4$ vertices. In
    particular, upon inspection of Fig.~\ref{fig:NPc-gadgets} and
    Steps~\ref{item:rule-F} and~\ref{item:rule-R}, one easily verifies that
    none of these components contains an induced $P_4$. Therefore,
    $\Gg\cup\Gu$ must be a cograph.  Finally, since $\graphs$ is a graph
    3-partition, $\Ga$ is the complement graph of $\Gg\cup\Gu$ and thus
    also a cograph.
    \hfill$\diamond$
  \end{claim}
  
  By Theorem~\ref{thm:explainable-charac}, it remains to show that
  $(\R_S(\graphs), \F_S(\graphs))$ is consistent. To this end, we construct
  a species tree $S$ that agrees with $(\R_S(\graphs),
  \F_S(\graphs))$. First, we set $S\coloneqq S^*$ and, for each $u \in U$,
  relabel the leaf $u$ in $S$ to $\u$.  Second, we insert the remaining
  leaves $\{\w_t \colon t \in \CR\}$ to $S$. To this end, for each $t
  \coloneqq \rpair{x}{y}{z} \in \CR$, we add $\w_t$ as a child of
  $\lca_S(\y, \z)$. We note that if $S$ contains a fan triple $\a\vert
  \b\vert \c$ (resp.\ rooted triple $\a\b\vert \c$) for $\a,\b, \c \in
  \Sigma$, then after inserting a leaf as a child of an existing vertex of
  $S$, the tree $S$ still displays $\a\vert \b\vert \c$ or $\a\b\vert \c$,
  respectively.  Therefore, each insertion of a leaf $\w_t$ preserves the
  triples and fan triples that are already displayed by $S$.
    
  We continue by showing that $S$ agrees with $(\R_S(\graphs), \F_S(\graphs))$.

  \begin{claim}\label{claim:display-R}
    The species tree $S$ displays every triple in $\R_S(\graphs)$.
    
    \smallskip\normalfont
    \noindent
    \emph{Proof of Claim \ref{claim:display-R}.}  Suppose that there are
    $a,b,c \in V(G)$ that imply an informative triple $\sigma(a) \sigma(b)
    \vert \sigma(c) \in \R_S(\graphs)$ (we refrain from using $x,y,z$ as in
    Definition~\ref{def:inf-forb-triples} to avoid confusion with the $x_t,
    y_t, z_t$ vertices).  Together with
    Definition~\ref{def:inf-forb-triples}, this implies that one of the
    following two cases holds: (1) $ac, bc \in E(\Gu)$ and $ab\notin
    E(\Gu)$ or (2) $ab \in E(\Ga)$ and $ac, bc \notin E(\Ga)$.
  
    \textit{Case~(1): $ac, bc \in E(\Gu)$ and $ab \notin E(\Gu)$.}  By
    rule~\ref{item:rule-ccs}, vertices of distinct connected components of
    $G$ are connected by edges in $\Ga$. Since $ac, bc \in E(\Gu)$, the
    vertices $a,b$ and $c$ must be contained in the same connected
    component of $G$.  Clearly, each $P_3$ component contains at most one
    edge in $\Gu$ (since two of the three possible edges are edges in
    $G=\Gg$).  Therefore, $a,b,c$ must be part of a paw component belonging
    to an $R_t$ subgraph, with $t = \rpair{x}{y}{z} \in \CR$. In
    particular, we must have $a=x_t$, $b=y_t$, and $c=z_t$ (noting that the
    roles of $a$ and $b$ are interchangeable). Since $x_t z_t = ac \in
    E(\Gu)$, $S^*$ must display $xy\vert z$ according to
    rule~\ref{item:rule-R} and, thus, $S$ displays $\x \y \vert \z =
    \sigma(a) \sigma(b) \vert \sigma(c)$.
    
    \textit{Case~(2): $ab \in E(\Ga)$ and $ac, bc \notin E(\Ga)$.}  By
    rule~\ref{item:rule-ccs}, vertices of distinct connected components of
    $G$ are connected by edges in $\Ga$. Since $ac, bc \notin E(\Ga)$, the
    vertices $a,b$ and $c$ must be contained in the same connected
    component of $G$. Since we never add $\Ga$ edges between vertices in a
    $P_3$ component, $a,b,c$ must be part of a paw component belonging to
    an $R_t$ subgraph, with $t = \rpair{x}{y}{z} \in \CR$. In particular,
    we must have $a=x_t$ and $b=z_t$ (again, the roles of $a$ and $b$ are
    interchangeable). Since $x_t z_t = ab \in E(\Ga)$, $S^*$ must display
    $xz\vert y$ according to rule~\ref{item:rule-R} and, thus, $S$ displays
    $\x \z \vert \y$. By construction of $S$, $\w_t$ is a child of
    $\lca_S(\y,\z)$. Together with $S$ displaying $\x \z \vert \y$, this
    implies that $S$ also displays $\x \z \vert \w_t$. For $c$, the two
    possibilities $c=y_t$ and $c=w_t$ remain, for which we obtain
    $\sigma(a) \sigma(b) \vert \sigma(c) = \x \z \vert \y$ and $\sigma(a)
    \sigma(b) \vert \sigma(c) = \x \z \vert \w_t$, respectively. Hence,
    $\sigma(a) \sigma(b) \vert \sigma(c)$ is displayed by $S$ in both
    cases.
    
    In summary, $S$ displays every informative triple of $\R_S(\graphs)$.
    \hfill$\diamond$
  \end{claim}

  \begin{claim}\label{claim:nodisplay-F}
    The species tree $S$ does not display any triple in $\F_S(\graphs)$.

    \smallskip\normalfont
    \noindent    
    \emph{Proof of Claim \ref{claim:nodisplay-F}.}  Suppose that there are
    vertices $a,b,c \in V(G)$ that imply a forbidden triple $\sigma(a)
    \sigma(b) \vert \sigma(c) \in \F_S(\graphs)$.  By
    Definition~\ref{def:inf-forb-triples}, we have (1) $ab, bc \in E(\Gg)$
    and $ac \notin E(\Gg)$ or (2) $ab, ac \in E(\Gg)$ and $bc \notin
    E(\Gg)$. In the following, we consider only Case~(1), since analogous
    arguments apply in Case~(2).  Because $ab,bc \notin E(\Ga)$, we know
    that $a$, $b$ and $c$ are contained in the same connected component of
    $G$.
    
    Suppose that $a$, $b$, and $c$ are in the same $P_3$ component of some
    $F_t$ subgraph where $t = x\vert y\vert z \in \CF$. Thus
    $\{\sigma(a),\sigma(b),\sigma(c)\} = \{\x,\y,\z\}$.  In this case,
    since $S^*$ contains $x\vert y\vert z$, $S$ contains $\x\vert \y\vert
    \z = \sigma(a) \vert \sigma(b) \vert \sigma(c)$ and thus does not
    contain the forbidden triple implied by $a,b,c$.
    
    Suppose that $a$, $b$, and $c$ are in the same $P_3$ component of some 
    $R_t$ component where $t = \rpair{x}{y}{z} \in \CR$. Thus
    $\{\sigma(a),\sigma(b),\sigma(c)\} = \{\w_t,\y,\z\}$. Since we have added 
    $\w_t$ as a child of $\lca_S(\y,\z)$, $S$ contains $\w_t\vert \y\vert
    \z = \sigma(a) \vert \sigma(b) \vert \sigma(c)$ and thus does not
    contain the forbidden triple implied by $a,b,c$.
  
    Finally, suppose that $a$, $b$, and $c$ are in the same paw component
    of some $R_t$ component where $t = \rpair{x}{y}{z} \in \CR$. Then
    either (i) $a = y_t$, $b = w_t$, $c = z_t$; (ii) $a = z_t$, $b = w_t$,
    $c = y_t$; (iii) $a = x_t$, $b = w_t$, $c = z_t$; or (iv) $a = z_t$, $b
    = w_t$, $c = x_t$. In Cases~(i) and~(ii), we again have
    $\{\sigma(a),\sigma(b),\sigma(c)\} = \{\w_t,\y,\z\}$ and, as argued
    before, $S$ does not contain the forbidden triples implied by $a,b,c$.
    Now consider Cases~(iii) and~(iv), and thus $\sigma(a)\sigma(b) \vert
    \sigma(c) = \x \w_t \vert \z$ and $\sigma(a)\sigma(b) \vert \sigma(c) =
    \z \w_t \vert \x$, respectively.  Since $S^*$ displays either $xy\vert
    z$ or $xz\vert y$, $S$ displays $\x\y\vert \z$ or $\x\z\vert \y$. Since
    we have moreover added $\w_t$ as a child of $\lca_S(\y,\z)$, $S$
    displays $\x \vert \z \vert \w_t$ or $\x \z \vert \w_t$,
    respectively. Hence, $S$ displays none of the two forbidden triples
    obtained in Cases~(iii) and~(iv).
    
    Taken together, $S$ does not display a triple in $\F_S(\graphs)$.
    \hfill$\diamond$
  \end{claim}
  \endgroup

  We have constructed the graph 3-partition $\graphs =
  (\Gu,\Gg,\Ga,\sigma)$ such that $\Gu$ and $\Gg$ are properly colored by
  Claim~\ref{claim:propcolored}, $\Gu$ and $\Ga$ are cographs by
  Claim~\ref{claim:cographs}, and $(\R_S(\graphs),\F_S(\graphs))$ is
  consistent by Claim~\ref{claim:display-R} and
  Claim~\ref{claim:nodisplay-F}. By Theorem~\ref{thm:explainable-charac},
  this implies that $\graphs$ can be explained by a relaxed scenario
  $\scen$. Since $\Gg(\scen) = \Gg = G$, we can conclude that $G$ is an EDT
  graph.
   
  In summary, we have established that \textsc{EDT-Recognition} is
  NP-complete. Moreover, the graph $G$ constructed in the reduction from
  the \textsc{$(\CF,\CR)$-Satisfiability} problem is a cograph because it
  does not contain a $P_4$ as an induced subgraph. Therefore
  \textsc{EDT-Recognition} remains NP-hard if the input graph is a cograph.
  \end{proof}
  
  \end{appendix}

\bibliographystyle{spbasic}      
\bibliography{edt-graph}   

\end{document}